%% file: neurips_2026.tex
\RequirePackage{fix-cm}
\documentclass[letterpaper,11pt]{article}
\PassOptionsToPackage{table}{xcolor}

\usepackage[utf8]{inputenc}
\usepackage[T1]{fontenc}
\usepackage{times}
\usepackage[margin=1in,top=0.45in]{geometry}
\usepackage{microtype}
\usepackage{amsmath,amssymb,amsfonts,amsthm,mathtools}
\input{math_commands.tex}

\usepackage{graphicx}
\usepackage[export]{adjustbox}
\usepackage{tabularx}
\usepackage{booktabs,multirow}
\usepackage{caption}
\usepackage[table]{xcolor}
\usepackage{enumitem}
\usepackage[most]{tcolorbox}
\usepackage{xspace}
\usepackage{placeins}
\usepackage{needspace}
\usepackage{titlesec}
\setlist{leftmargin=5mm}

\definecolor{TblBlue}{HTML}{EAF1F8}
\definecolor{TblTeal}{HTML}{E9F4F1}
\definecolor{TblSand}{HTML}{F8F1E5}
\definecolor{TblHeader}{HTML}{E9EEF3}
\definecolor{TblOurs}{HTML}{DDEEE9}
\definecolor{TblAvg}{HTML}{F3F5F7}
\definecolor{TblInk}{HTML}{254458}

\usepackage{natbib}
\usepackage{url}
\usepackage[colorlinks=true,linkcolor=cyan,citecolor=cyan,filecolor=magenta,urlcolor=blue]{hyperref}
\usepackage[nameinlink,capitalise]{cleveref}
\hypersetup{
  pdftitle={Lluna: Collaborative Unimodal Models for Missing-Modality Inference},
  pdfauthor={Haoyue Bai, Yinyi Luo, Heqing Huang, Xiaoyuan Yi, James Evans, Jindong Wang}
}

\newcommand{\luna}{\textsc{Lluna}\xspace}
\newcommand{\pcc}{\textsc{PCC}\xspace}
\newtheorem{theorem}{Theorem}
\newtheorem{corollary}{Corollary}

\newcommand{\wjdd}[1]{}
\newcommand{\wjd}[1]{}
\newcommand{\bhy}[1]{}

\let\appendixaddcontentsline\addcontentsline
\renewcommand{\addcontentsline}[3]{}

\definecolor{abstractbg}{RGB}{230,242,250}
\definecolor{abstractborder}{RGB}{200,210,220}
\titleformat{\section}
  {\sffamily\bfseries\Large}
  {\thesection}
  {1em}
  {}

\renewenvironment{abstract}
{%
\begin{center}
\begin{tcolorbox}[
    colback=abstractbg,
    colframe=abstractborder,
    boxrule=0.6pt,
    arc=6pt,
    width=\textwidth,
    left=8pt,
    right=8pt,
    top=8pt,
    bottom=8pt
]
}
{%
\end{tcolorbox}
\end{center}
}

\newcommand{\papertitle}{%
\sffamily\bfseries\fontsize{16}{19}\selectfont
Lluna: Collaborative Unimodal Models for\\
Missing-Modality Inference
}

\newcommand{\paperauthors}{%
\sffamily
Haoyue Bai$^1$, Yinyi Luo$^2$, Heqing Huang$^1$,\\[3pt]
Xiaoyuan Yi$^3$, James Evans$^4$, Jindong Wang$^1$
}

\newcommand{\paperaffiliations}{%
$^1$College of William and Mary \quad $^2$Carnegie Mellon University\\[3pt]
$^3$Microsoft \quad $^4$University of Chicago
}

\begin{document}
\thispagestyle{empty}

\noindent
\includegraphics[height=.6cm]{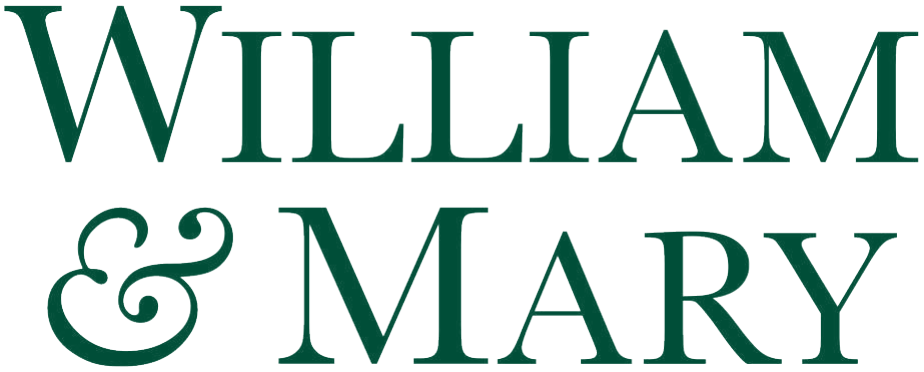}\hspace{0.45cm}%
\includegraphics[height=.6cm]{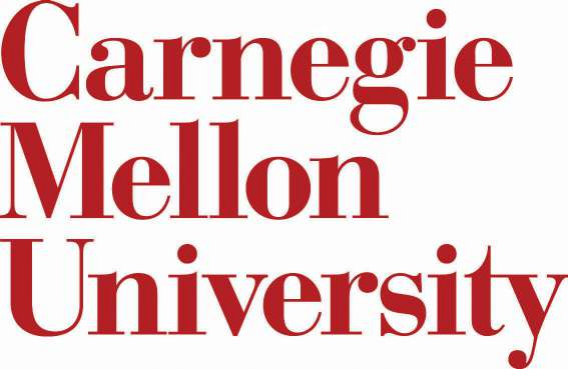}\hspace{0.45cm}%
\includegraphics[trim=218bp 219bp 220bp 218bp,clip,height=.5cm]{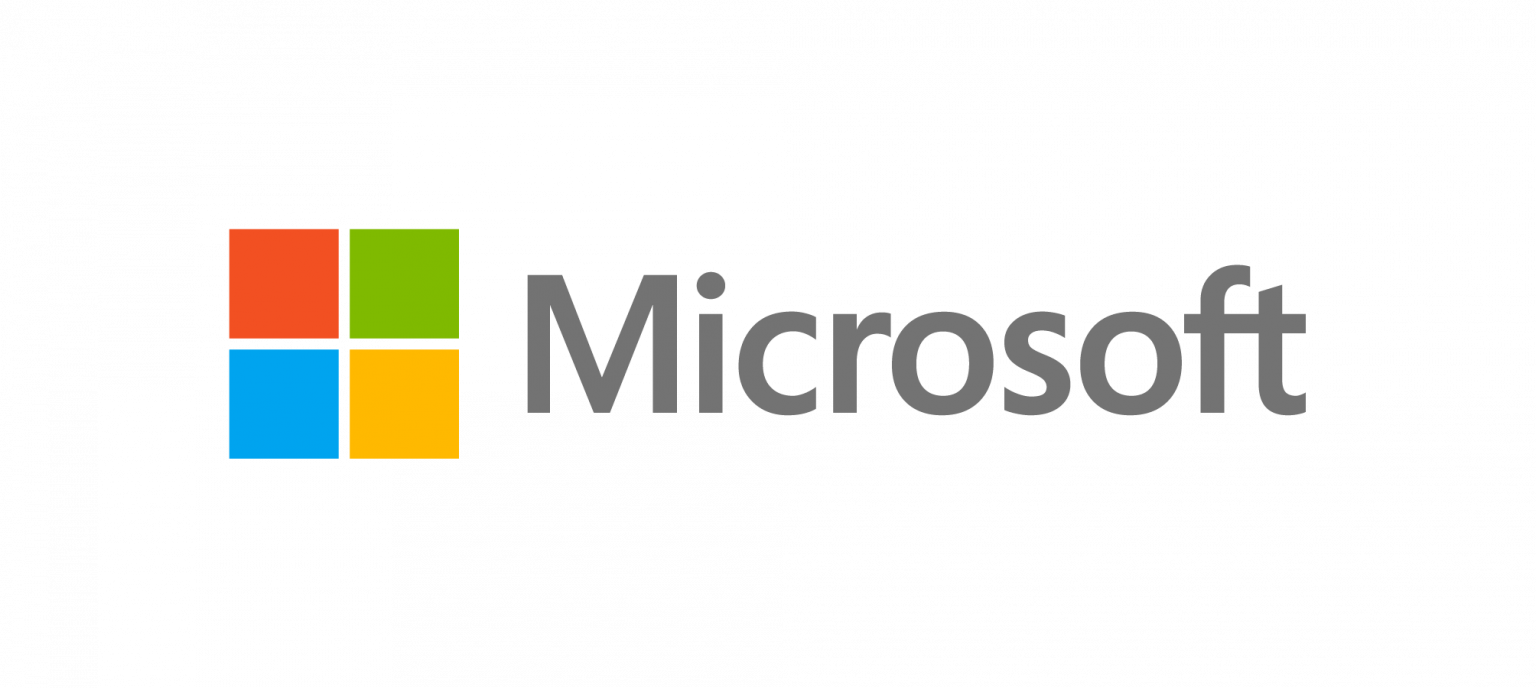}\hspace{0.45cm}%
\includegraphics[trim=18bp 192bp 18bp 190bp,clip,height=.5cm]{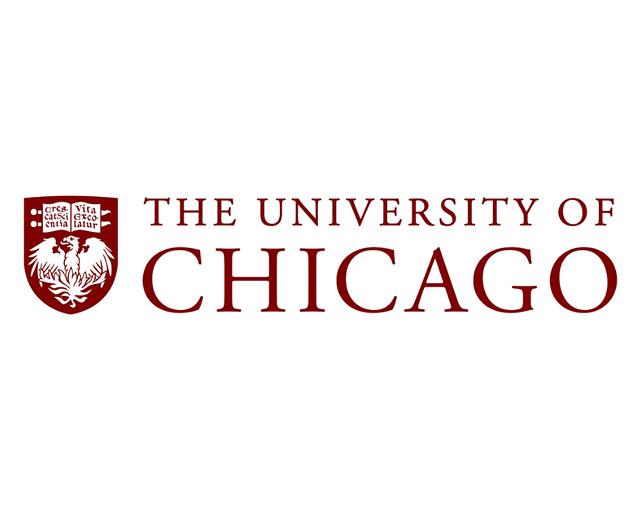}

\vspace{-.15in}
\noindent\rule{\textwidth}{0.8pt}
\vspace{0.65cm}

\begin{center}
    {\papertitle\par}
    \vspace{0.45cm}
    {\large \paperauthors\par}
    \vspace{0.2cm}
    {\normalsize \paperaffiliations\par}
\end{center}

\input{section/abstract}
\input{section/introduction}
\input{section/related_work}
\input{section/method}

\input{section/experiments}
\input{section/conclusion}

\bibliographystyle{plainnat}
\bibliography{refs}

\clearpage
\appendix
\let\addcontentsline\appendixaddcontentsline
\input{section/appendix}
\end{document}

%% file: math_commands.tex
\usepackage{amsmath,amsfonts,bm}

\def\1{\bm{1}}

\DeclareMathAlphabet{\mathsfit}{\encodingdefault}{\sfdefault}{m}{sl}
\SetMathAlphabet{\mathsfit}{bold}{\encodingdefault}{\sfdefault}{bx}{n}



%% file: section/abstract.tex
\begin{abstract}
Pretrained unimodal models offer a flexible route to multimodal prediction through collaboration, but this collaboration can become less accurate than a surviving unimodal model when modalities are missing. 
We study this setting using fixed surviving-model anchors and find that collaborative gains can weaken or reverse as inputs disappear, even when the anchor's own input is unchanged. Our diagnostics further reveal reduced recovery of modality-specific information after aggregation and state-dependent collaborative representations.
We propose \luna, which combines a direct private prediction path from adapted unimodal features with a common graph-based collaboration branch that provides a residual correction. Progressive Collaboration Consistency organizes richer-to-poorer guidance along single-modality removals, aligning normalized common representations while task supervision adapts predictions to the available evidence. Across eight datasets spanning affective computing, scene understanding, and medical prediction, \luna achieves the best evaluated average performance over non-full input states, improving Mean-Missing by 4.26\% on average relative to the strongest evaluated baseline. The fixed-anchor diagnostic also shows positive gains in mean collaboration under minimal inputs on all eight tasks. Code is available at https://github.com/AIFrontierLab/Lluna.
\end{abstract}

%% file: section/introduction.tex
\section{Introduction}

Multimodal systems are essential for perceiving and reasoning about the real world \citep{baltrusaitis2019multimodal,alayrac2022flamingo,girdhar2023imagebind}.
In practice, however, modalities may be unavailable: audio streams may be missing, sensors may fail, and clinical measurements may be uncollected \citep{wu2026missingmodalitysurvey,wang2023shaspec,yun2024flexmoe,li2025simmlm}.
Prediction with missing modalities is therefore an important problem in multimodal learning.

Most missing-modality methods learn a single multimodal predictor, using strategies such as modality dropout, aggregation of available representations, or shared-specific modeling \citep{neverova2016moddrop,havaei2016hemis,wang2023shaspec,yun2024flexmoe,li2025simmlm,xie2026ccsd}.
In contrast, with the abundance of publicly-available pretrained foundation models \citep{osborne2024ai}, we study a different regime where independently pretrained unimodal models are composed through \textit{model collaboration}.
Specifically, model collaboration can align heterogeneous representations or combine model outputs, as in Freeze-Align \citep{maniparambil2025freezealign} and FoE \citep{wang2024foe}.
This raises a basic question: does collaboration remain effective when some modalities become unavailable?
As shown in \Cref{fig:motivation}(a), across affective computing, scene understanding, and medical prediction, our experiments show that collaboration gains got smaller as modalities become unavailable and can even turn negative.
A surviving unimodal model can then be more accurate than the collaborative predictor, even though its own input remains unchanged.

We next examine how collaboration represents the available evidence. For the illustrative runner in \Cref{fig:motivation}(b), ``I finally did it'' conveys achievement, a relieved voice adds emotional context, and the video reveals physical effort. These complementary cues illustrate the \emph{private signal}, i.e., the modality-specific information collaboration should retain.
Yet feature-level probes find that this signal is less recoverable after collaboration across various mechanisms and training policies (\Cref{fig:motivation}(b), Appendix~\ref{app:probe_multimethod}).
The second challenge is to coordinate collaboration across available missing-modality states.
As video and audio disappear subsequently, the same event is represented through different cues and interactions (\Cref{fig:motivation}(c)).
This calls for direct access to modality-specific features together with cross-state guidance for the common collaborative branch.

\begin{figure}[t]
  \captionsetup{labelfont=bf,justification=raggedright,singlelinecheck=false,skip=5pt}
  \centering
  \includegraphics[width=.9\textwidth]{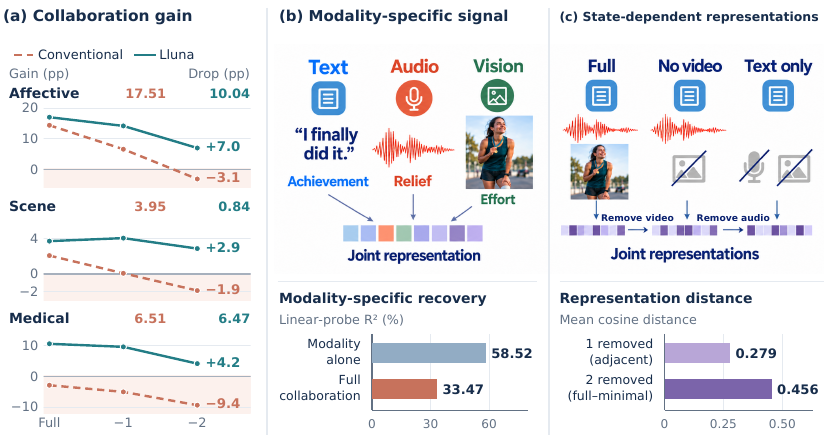}
  \vspace{-.1in}
  \caption{\textbf{Collaboration under missing modalities.} \textbf{(a)} Domain-mean accuracy gain in percentage points (pp) over fixed surviving unimodal models; annotations give the absolute decline from Full to two removals in pp. \textbf{(b)} Attn. pool's private-signal recovery, averaged over eight tasks, using the linear-probe coefficient of determination ($R^2$). \textbf{(c)} Cosine distances for adjacent and Full-to-Minimal state pairs, averaged over six three-modality tasks. Details: Appendix~\ref{sec:diagnosis}.}
  \label{fig:motivation}
  \vspace{-.1in}
\end{figure}

In this paper, we propose \textbf{\luna}, a collaboration strategy for unimodal models for missing-modality inference.
As shown in \Cref{fig:overview}, \luna pairs \emph{private prediction} with \emph{common collaboration}. Model-specific adapters provide a direct path from available features to prediction, while a common branch exchanges messages on a compatibility graph weighted by prediction agreement. Its output supplies a learned residual correction to the private prediction. Single-modality removals connect the full input to every nonempty availability state, decomposing larger evidence changes into local steps.
Progressive Collaboration Consistency (\pcc) uses these steps to train a shared model: richer states guide the common representations of their poorer neighbors, and intermediate states in turn provide targets for subsequent steps.
We further establish a local-to-global stability bound, relating the drift of normalized common representations between distant availability states to accumulated discrepancies along single-modality transitions.

Experiments across eight datasets in three application domains show that \luna improves Mean-Missing performance over the strongest evaluated baseline (among collaboration methods, missing-modality approaches, and Omni models) on every task, averaging $4.26\%$ relative improvement. The collaboration diagnostic also shows positive mean gains under minimal inputs on all eight tasks. Figure~\ref{fig:motivation}a further shows a $37.9\%$ reduction in the domain-averaged absolute decline in collaboration gain after two modality removals, compared with conventional collaboration.

Our contributions are:
\begin{itemize}[leftmargin=2em]
\setlength\itemsep{0em}
    \item We identify loss of collaborative benefit under missing modalities and characterize the accompanying changes in private-signal recovery and collaborative representations.
    \item We propose \luna, which retains a direct prediction path, adds a common residual correction, and applies progressive richer-to-poorer consistency to the common representation.
    \item Extensive experiments in three domains show improved Mean-Missing performance on every task, complemented by positive mean collaboration gains under minimal inputs in the diagnostic.
\end{itemize}

%% file: section/related_work.tex
\FloatBarrier
\section{Related Work}
\label{sec:related_work}
\textbf{Model collaboration.}
Model collaboration combines complementary predictions and representations through ensemble weighting and learned feature interactions \citep{baltrusaitis2019multimodal,wolpert1992stacking,jacobs1991mixtures,vaswani2017attention,lee2019set}. Pretrained-model ensembling and routing often combine models within the same modality \citep{wang2024foe,lu2024collaboration,yadav2024moerging}. At the feature level, shared embeddings align signals across modalities \citep{girdhar2023imagebind}, frozen-encoder alignment connects pretrained unimodal models \citep{maniparambil2025freezealign}, and graph-based aggregation supports variable modality sets \citep{patapati2026clarga}.

\textbf{Multi-agent Systems.} LLM-based MAS coordinate instruction-following agents through dialogue, roles, and task planning \citep{li2023camel,wu2024autogen,hong2024metagpt,talebirad2023multi}. \luna uses relatively compact unimodal models, including encoders without a conversational interface, and requires neither instruction following nor agentic behavior. Moreover, modality loss changes the available evidence and can make some models unavailable. \luna learns collaboration over representations, rather than assuming all participants communicate via language.

\textbf{Learning with missing modalities.}
Missing-modality learning adapts representations, fusion, and supervision to incomplete inputs~\citep{wu2026missingmodalitysurvey}.
Earlier work uses modality dropout and aggregation over available-modality representations~\citep{neverova2016moddrop,havaei2016hemis}.
Other approaches include shared-specific feature learning, modality-aware routing, and cross-state distillation~\citep{wang2023shaspec,yun2024flexmoe,li2025simmlm,baek2026mars,xie2026ccsd}, as well as quality-aware fusion and reasoning from incomplete observations~\citep{meng2026giml,xu2026limssr}.
Most of these methods address missing inputs within a single task-specific multimodal predictor.
In contrast, \luna studies missing-modality inference as a flexible model-collaboration problem among independently pretrained unimodal models, where modality loss changes both the available evidence and the set of models that can participate in collaboration.

\textbf{Omni and unified models.}
Integrating modalities into one predictor offers another route, from pretrained vision--language coupling for few-shot learning \citep{alayrac2022flamingo} to Omni models combining vision, audio, and language for understanding and interaction \citep{xu2025qwenomni,tong2025interactiveomni,ye2025omnivinci,gemmateam2026gemma4,zhang2025streamomni}. Unified multimodal models bring understanding and generation into a common modeling framework \citep{luo2026torchumm, wang2024emu3,xie2025showo,deng2025bagel}.

%% file: section/method.tex
\Needspace{8\baselineskip}
\section{Lluna}
\label{sec:method}


We are given labeled training examples $\{(x,y)\}$, where $y$ is the target label and $x$ has $M$ different modalities: $\mathcal{M}=\{M_i\}^M_{i=1}$, with $M_i$ denoting the $i$th modality.
Missing-modality learning aims to learn a function $f$ to infer the labels for a testing set $\{\tilde{x}\}$ that has at most $M-1$ modalities.

\subsection{Overview of \luna}


Different from existing missing-modality approaches that learn a single model~\citep{meng2026giml,li2025simmlm,wang2023shaspec}, \luna instead learns a collaboration strategy from a set of $N$ \textit{heterogeneous, unimodal} models, $\mathcal{F}=\{F_i\}^N_{i=1}$ for robust prediction.
Different $F_i$s may specialize in different modalities, e.g., $F_1$ can be Qwen3-1.7B~\citep{yang2025qwen3} as a textual model, and $F_2$ can be Whisper~\citep{whisper} as an audio model.
All $F_i$s remain \textit{frozen} in training phase.

To sustain beneficial collaboration under missing modalities, \luna builds on a private prediction path that retains direct access to the available modality-specific features (Figure~\ref{fig:overview}).
Using the compatibility graph, the common collaboration path learns interactions among available models to refine this base prediction. Progressive Collaboration Consistency (\pcc) further uses richer availability states to guide the common representations learned under progressively reduced inputs.

\begin{figure}[!htb]
\centering
\includegraphics[width=.9\textwidth]{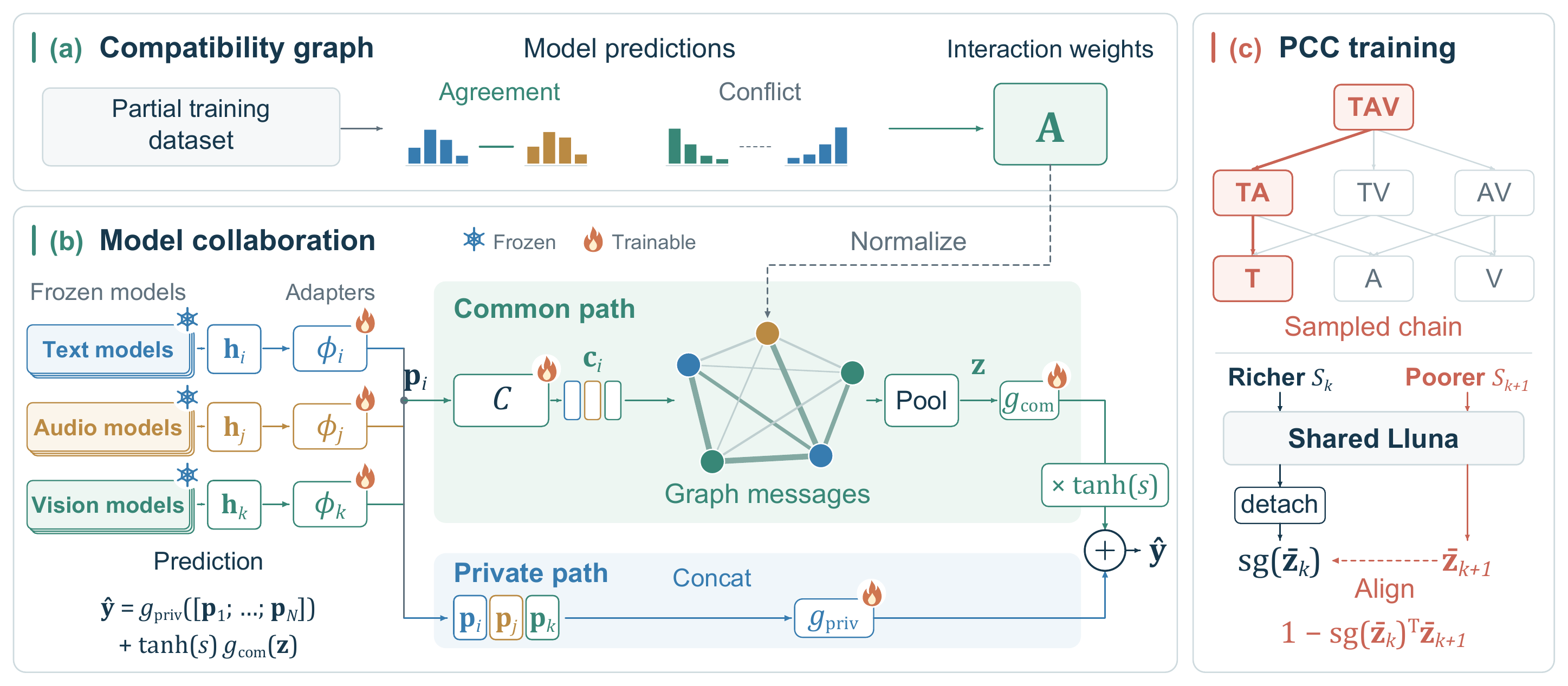}
\captionsetup{skip=4pt}
\caption{\textbf{Overview of \luna.} \textbf{(a)} Compatibility graph building using training-partition prediction agreement. \textbf{(b)} Residual collaboration of private and common features.
\textbf{(c)} PCC training. Consecutive availability states $S_k,S_{k+1}$ differ by one removed modality; here $k$ indexes chain positions. \pcc aligns their unit-normalized consensuses $\bar{\mathbf z}_{k+1}$ and $\bar{\mathbf z}_k$, with stop-gradient $\operatorname{sg}$ on the richer state. T/A/V denote text/audio/vision. Inference uses only the observed state.}
\label{fig:overview}
\vspace{-.15in}
\end{figure}

\Needspace{6\baselineskip}
\subsection{Compatibility Graph Building}
\label{sec:compatibility}

An availability state $S\subseteq\mathcal M$ specifies the nonempty subset of observed modalities.
Since models in $\mathcal F$ may differ in modality and representation space, \luna first adapts each frozen feature $\mathbf h_i$ with a model-specific adapter $\phi_i$, giving $\mathbf p_i=\phi_i(\mathbf h_i)$, and then maps it into a common interaction space using a shared aligner $C$, with $\mathbf c_i^{(0)}=C(\mathbf p_i)$.
Below, $i,j$ index models; unless stated otherwise, message passing, pooling, and their normalization involve only available models.
We omit state indices when describing a single state.

Heterogeneous models can provide complementary information, but effective collaboration also requires their evidence to be sufficiently compatible rather than strongly conflicting.
We therefore use agreement in task predictions as a soft measure of compatibility to organize cross-model interactions (\Cref{fig:overview}(a)).
For each dataset, we fit a task head to each model's frozen training features and define $A_{ij}=\operatorname{sim}(\widehat{\mathbf y}_i^{(n)},\widehat{\mathbf y}_j^{(n)})$, where $\widehat{\mathbf y}_i^{(n)}$ is $F_i$'s task-head prediction for $n$ partial training examples.
The symmetric similarity $\operatorname{sim}$ uses class probabilities for classification and standardized scalar predictions for regression.
The graph is estimated once from the training partition and kept fixed throughout collaboration training and inference.
For each availability state, we restrict $\mathbf A$ to available models and normalize incoming weights for each target model $j$ as $\widehat A_{ij}=A_{ij}/\sum_{i'}A_{i'j}$.
Positive self-loops ensure a valid incoming connection even for an isolated available model (More details in Appendix~\ref{app:graph_configuration}).

\subsection{Residual Coordination of Common and Private Information}

Our analysis shows that modality-specific signals can become less accessible after collaboration (Appendix~\ref{app:probe_multimethod}), suggesting that common collaboration should refine, rather than replace, the direct evidence provided by individual models.
Accordingly, \luna decomposes prediction into a private base prediction from adapted model features and a common residual correction from cross-model interaction (\Cref{fig:overview}(b)).
The private head $g_{\mathrm{priv}}$ is an MLP that reads the adapted features $\mathbf p_i$ in a fixed model order, with zeros in unavailable slots, to form the base prediction.
The common branch instead operates on the aligned features $\mathbf c_i^{(0)}$ through the compatibility graph.
Using standard graph message passing~\citep{gilmer2017neural}, each aligned feature  is updated as $\mathbf c_j^{(\ell+1)}=\operatorname{LayerNorm}\!\left(\mathbf c_j^{(\ell)}+\gamma(\mathbf c_j^{(\ell)})\sum_i\widehat A_{ij}U(\mathbf c_i^{(\ell)})\right)$, where $U$ is a learned message map and $\gamma$ is a sigmoid gate.
After $L$ layers, attention pooling produces the common consensus $\mathbf z=\operatorname{AttentionPool}(\{\mathbf c_i^{(L)}\})$ (details in Appendix~\ref{app:implementation}).

Crucially, common collaboration refines rather than replaces prediction from model-specific evidence.
A common head $g_{\mathrm{com}}$, implemented as an MLP, maps $\mathbf z$ to a residual correction:
\begin{equation}
\widehat{\mathbf y}
=
g_{\mathrm{priv}}([\mathbf p_1;\ldots;\mathbf p_N])
+\tanh(s)\,g_{\mathrm{com}}(\mathbf z),
\label{eq:prediction}
\end{equation}
where $s$ is a learned residual-scaling scalar~\citep{bachlechner2021rezero}.
Both $g_{\mathrm{priv}}$ and $g_{\mathrm{com}}$ additionally receive binary availability indicators, omitted from the notation for clarity.

\subsection{Training with Progressive Collaboration Consistency}

Missing modalities change both the available evidence and the resulting collaborative representation.
Single-modality removals form a connected graph over nonempty availability states, providing a local structure for richer-to-poorer guidance.
Along sampled transitions, \pcc aligns the poorer state's common representation with its richer neighbor, while task supervision allows the overall prediction to adapt to the remaining evidence.
Intermediate states can in turn guide subsequent removals, yielding \textit{progressive richer-to-poorer} transfer across availability states.

As shown in \Cref{fig:overview}(c), we organize nonempty modality subsets into removal chains $S_0\supset S_1\supset\cdots$, where each transition $S_k\rightarrow S_{k+1}$ removes exactly one modality; for example, $\mathrm{TAV}\rightarrow\mathrm{TA}\rightarrow\mathrm{T}$.
Let $\mathbf z_k$ denote the common consensus at state $S_k$ and $\bar{\mathbf z}_k=\mathbf z_k/\|\mathbf z_k\|_2$ its unit-normalized form.
For each adjacent transition, \pcc aligns the poorer representation to its richer neighbor while stopping the gradient through the richer state:
\begin{equation}
\mathcal{L}_{\mathrm{PCC}}
=
\sum_{(S_k,S_{k+1})\in\mathcal E}
\omega_{k,k+1}
\left[
1-\operatorname{sg}(\bar{\mathbf z}_k)^\top
\bar{\mathbf z}_{k+1}
\right],
\label{eq:pcc}
\end{equation}
where $\operatorname{sg}$ denotes stop-gradient, $\mathcal E$ contains the sampled adjacent transitions, and $\omega_{k, k+1}$ balances contributions across different availability states.
Because an intermediate state can receive guidance from a richer neighbor and subsequently guide a poorer one, these local consistency constraints propagate progressively along the removal chain.
Normalizing $\mathbf z_k$ focuses the guidance on the direction of the common consensus while allowing its magnitude to vary with the available evidence (details in Appendix~\ref{app:pcc_configuration}).
Appendix~\ref{app:derivations} further relates accumulated local discrepancies to representation changes between distant states.
The overall objective combines state-wise task supervision, auxiliary supervision of the common prediction, and progressive consistency:
\begin{equation}
\mathcal{L}
=
\mathbb{E}\!\left[
\mathcal{L}_{\mathrm{task}}(y,\widehat{\mathbf y})
+
\lambda_{\mathrm{com}}
\mathcal{L}_{\mathrm{task}}
\!\left(y,g_{\mathrm{com}}(\mathbf z)\right)
\right]
+
\lambda_{\mathrm{PCC}}\mathcal{L}_{\mathrm{PCC}}.
\label{eq:objective}
\end{equation}
Here, $\lambda_{\mathrm{com}}$ and $\lambda_{\mathrm{PCC}}$ control the auxiliary and consistency terms.
Task supervision adapts the complete prediction to each availability state, while \pcc regularizes only the common collaboration.

\textbf{Efficient chain sampling.}
Training with randomly dropped modalities is an established strategy for incomplete-input robustness~\citep{neverova2016moddrop,havaei2016hemis,wang2023shaspec}.
However, exhaustively evaluating all $2^M-1$ nonempty availability states at every training step becomes increasingly expensive as the number of modalities grows.
To improve training efficiency while retaining progressive guidance, \luna samples $K$ maximal removal chains on the same training examples, reusing states and adjacent transitions shared across chains.
Validation and test evaluate all legal availability states, while inference evaluates only the observed state.

\Needspace{8\baselineskip}
\subsection{Local-to-Global Representation Stability}
\label{sec:local_global_stability}

Although \pcc regularizes only adjacent availability states, local control extends along the availability graph. Any two nonempty states can be connected through a sequence of single-modality changes.

\begin{theorem}[Local-to-global representation stability]
\label{thm:representation_stability}
Fix an input and a path $\{S_i\}_{i=0}^D$ of $D\geq1$ single-modality
transitions. Let $d_k=1-\bar{\mathbf z}_{k-1}^{\top}\bar{\mathbf z}_k$
denote the cosine discrepancy between consecutive normalized common
representations. If the common consensus is nonzero at every state,
then
\begin{equation}
\|\bar{\mathbf z}_0-\bar{\mathbf z}_D\|_2
\leq
\sum_{k=1}^{D}\sqrt{2d_k}
\leq
\sqrt{2D\sum_{k=1}^{D}d_k}.
\label{eq:main_stability}
\end{equation}
\end{theorem}
Theorem~\ref{thm:representation_stability} shows that controlling local PCC discrepancies also limits representation drift between more distant availability states. Thus, progressive consistency provides pathwise global control without directly forcing distant evidence configurations to match. The proof is given in Appendix~\ref{app:derivations}.

%% file: section/experiments.tex
\section{Experiments}
\label{sec:experiments}

\subsection{Setup}
\textbf{Benchmarks and protocol.}
Our experiments cover eight datasets in three domains: \emph{affective computing}, \emph{scene understanding}, and \emph{medical prediction}. Affective computing includes CH-SIMS v2 (\textbf{SIMS})~\citep{liu2022chsims}, CMU-MOSI (\textbf{MOSI})~\citep{zadeh2016mosi}, MELD~\citep{poria2019meld}, and UR-FUNNY v2 (\textbf{URF})~\citep{hasan2019urfunny}. Scene understanding includes derived protocols from Daily-Omni (\textbf{Daily})~\citep{zhou2025dailyomni} and WorldSense (\textbf{WS})~\citep{hong2025worldsense}. Medical prediction includes the four-site CirCor DigiScope task (\textbf{CirCor})~\citep{oliveira2022circor} and MCI-to-dementia conversion prediction on ADNI~\citep{mueller2005adni}.
The main comparison evaluates every nonempty modality subset on a fixed test split. We define \textbf{Mean-Missing (MM)} as the task metric averaged equally over all nonempty, incomplete modality subsets in a run. For trainable same-pool methods, we select one checkpoint per run by validation MM for all test states and report the mean and standard deviation (SD) across five runs. Each dataset uses one \luna configuration shared by five runs. Task definitions, data splits, and training details are in Appendix~\ref{app:experimental_setup}; aggregation follows Appendix~\ref{app:multiseed_protocol}.

\input{tables/multiseed_overall}

\textbf{Baselines.}
We compare against three groups.
(1) \textit{General collaboration methods} include mean ensembling (Mean ens.) and concatenation MLP~\citep{baltrusaitis2019multimodal}, learned logit stacking (Logit stack)~\citep{wolpert1992stacking}, attention pooling (Attn. pool), soft MoE~\citep{jacobs1991mixtures}, FoE~\citep{wang2024foe}, Set Transformer (Set Trans.)~\citep{lee2019set}, and CLARGA~\citep{patapati2026clarga}.
(2) \textit{Missing-modality methods} include GIML~\citep{meng2026giml}, MARS~\citep{baek2026mars}, LIMSSR~\citep{xu2026limssr}, SimMLM~\citep{li2025simmlm}, ShaSpec~\citep{wang2023shaspec}, and Flex-MoE~\citep{yun2024flexmoe}.
(3) \textit{Omni comparators} include Qwen2.5-Omni 3B/7B~\citep{xu2025qwenomni}, InteractiveOmni 8B~\citep{tong2025interactiveomni}, OmniVinci 9B~\citep{ye2025omnivinci}, Gemma 4 12B IT~\citep{gemmateam2026gemma4}, and Stream-Omni-8B~\citep{zhang2025streamomni}.

\textbf{Shared model pool.}
For each dataset, \luna and all general collaboration and missing-modality baselines use the same frozen pool and identical pre-extracted representations, data splits, and modality-availability masks. They use these representations or derived task outputs, without updating or replacing upstream encoders or accessing additional raw inputs. Removing a modality jointly masks every representation derived from it. This controls information available for collaboration; trainable parameter counts are reported separately in Appendix~\ref{app:refresh_resources}. The main comparison uses a nine-model pool. The shared text/audio/vision pool uses three publicly available models per modality from different families. Omni baselines instead process available raw inputs with their frozen backbones. Appendix~\ref{app:expert_pools} specifies the task-specific pools and interfaces.

\subsection{Main results}
We evaluate \luna and its competitors using Mean-Missing performance.
In \Cref{tab:multiseed_overall}, \textbf{Avg. rank} summarizes cross-dataset performance: we rank all 21 methods by their Mean-Missing mean separately on each dataset, then average each method's eight ranks with equal weight.
Full-input results, with all modalities available, and results for each individual missing-input state are reported in Appendix~\ref{app:multiseed_views}.
\luna improves MM over the strongest evaluated baseline on all eight datasets, averaging $4.26\%$ relative improvement. For example, WS accuracy rises from $31.20\%$ to $33.90\%$, an $8.65\%$ relative gain, while SIMS MAE decreases from $0.3590$ to $0.3424$. These results demonstrate \luna's effectiveness in combining frozen unimodal models under missing modalities across classification and regression tasks.

\begin{table}[htbp]
\centering
\begin{minipage}{.4\textwidth}
\centering
\vspace{-.1in}
\caption{\textbf{Missing-input training.}}
\label{tab:supplement_missing_input_main}
\vspace{-.1in}
\input{tables/supplement_missing_input_main}
\end{minipage}
\end{table}

\textbf{Training with missing-input views.}
To test whether \luna's advantage persists under missing-input training, we evaluate ten representative strong baselines on one dataset per domain. The collaboration baselines receive missing-input supervision and select checkpoints by validation Mean-Missing. These controls include I$^2$MoE~\citep{xin2025i2moe} and use the shared frozen model pool. The three Omni models retain frozen backbones and use a shared task head fitted on complete and missing-input views. In Table~\ref{tab:supplement_missing_input_main}, \luna exceeds all ten controls on each task: the strongest control reaches $0.3546$ versus $0.3421$ MAE on SIMS, $32.64\%$ versus $33.82\%$ accuracy on WS, and $82.08\%$ versus $84.17\%$ balanced accuracy on ADNI. \luna remains ahead of all evaluated controls, supporting its effectiveness when the baselines are trained and selected for missing-input performance. Appendix~\ref{app:missing_input_controls} gives the training protocol and all test-view results.

\subsection{Ablation Study}
\label{sec:refresh_ablation_components}

\textbf{Private vs. common paths.}
To assess the contribution of each prediction path, we compare \luna with w/o Private and w/o Common. The former removes the direct private prediction; the latter removes the common branch. Removing either branch worsens MM on all eight tasks (\Cref{tab:refresh_ablation_core}a). On ADNI, balanced accuracy drops by $18.54$ pp without private prediction and by $4.79$ points without common collaboration, supporting the complementary roles of the two paths.
To examine how the paths should be combined, Equal addition fixes the common residual scale to one. The learned combination performs better on most tasks, supporting more consistent cross-task performance by adapting the contribution of the common prediction. More analyses are in Appendix~\ref{sec:refresh_ablation_details}.

\textbf{Graph interaction controls.}
To test whether prediction compatibility helps organize model interaction, we compare the fitted compatibility graph with Uniform Graph, which assigns equal weights over all available nodes (\Cref{tab:refresh_ablation_core}b). \luna outperforms Uniform Graph on seven tasks and ties on ADNI. These results favor compatibility-guided interaction over uniform aggregation.

\textbf{Consistency-pair selection.}
To examine which state relationships should be regularized, we compare \luna with w/o PCC, Full-to-Minimal, and All-Pairs, retaining the architecture and remaining training settings (Table~\ref{tab:refresh_ablation_core}c). W/o PCC removes the consistency objective. Full-to-Minimal aligns the full state directly with each single-modality state, whereas All-Pairs aligns every pair of distinct nonempty states. \luna's adjacent-state PCC outperforms removing PCC and both alternative pair selections on most tasks. These results favor concentrating consistency on single-modality transitions. Further details are given in Appendix~\ref{sec:refresh_ablation_details}.

\textbf{Direction of consistency.}
To test the value of directional guidance, w/o Stop-gradient retains the PCC loss but allows its gradients to update both the richer and poorer states. This worsens MM on most tasks (Table~\ref{tab:refresh_ablation_core}c), favoring richer-to-poorer guidance over mutual adjustment.

\input{tables/refresh_ablation_core}

\begin{figure}[!t]
\centering
\begin{minipage}[t]{.278\linewidth}
\centering
\includegraphics[width=\linewidth]{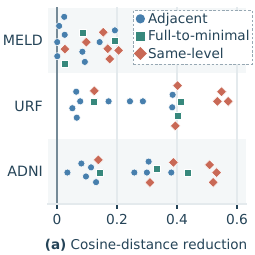}
\end{minipage}\hspace{.0145\linewidth}%
\begin{minipage}[t]{.356\linewidth}
\centering
\includegraphics[width=\linewidth]{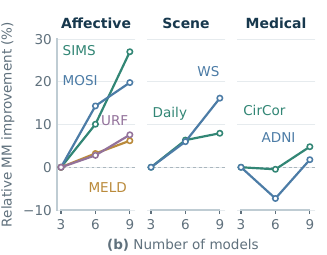}
\end{minipage}\hspace{.0145\linewidth}%
\begin{minipage}[t]{.237\linewidth}
\centering
\includegraphics[width=\linewidth]{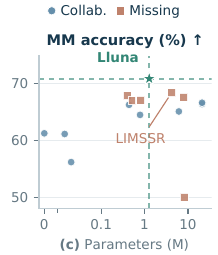}
\end{minipage}

\captionsetup{skip=4pt,justification=raggedright,singlelinecheck=false}
\caption{\textbf{Analyses of \luna.} (a) Reduction in common-state drift with PCC. (b) Relative MM gains over three-model pools. (c) MM versus trainable parameters on URF. Additional pool-size and parameter results are in Appendix~\ref{app:refresh_resources}.}
\label{fig:lluna_analysis}
\label{fig:refresh_diag_common}
\label{fig:model_pool_scaling_main}
\label{fig:refresh_resource_parameters}
\vspace{-.1in}
\end{figure}

\FloatBarrier

\textbf{Consistency target.}
To identify where cross-state consistency is most effective, we apply PCC to the common representation (\luna, or PCC-Common), the private representation (PCC-Private), or both (PCC-Both), while retaining the two prediction paths (Table~\ref{tab:refresh_ablation_core}d). PCC-Common outperforms both alternatives on all eight tasks. These results support focusing consistency on the common representation while preserving a separate path for modality-specific information.

\textbf{Effect of PCC on common-representation drift.}
To examine how PCC shapes relationships among availability states, we compare common-state cosine distances in matched models with and without PCC. For each state pair, absolute reduction is $d_{\mathrm{off}}-d_{\mathrm{on}}$: the pair's mean distance without PCC minus that with PCC, each averaged over the same test examples. Figure~\ref{fig:refresh_diag_common}(a) shows individual pair responses on MELD, URF, and ADNI. Both adjacent and more distant state pairs exhibit contraction in these examples, illustrating that local consistency can reduce drift beyond the directly regularized transitions. This complements the MM evidence in Table~\ref{tab:refresh_ablation_core}c.

\FloatBarrier
\Needspace{12\baselineskip}


\begin{figure}[!t]
\centering
\begin{minipage}[t]{2.09in}
\centering
\vspace{0pt}
\includegraphics[width=\linewidth]{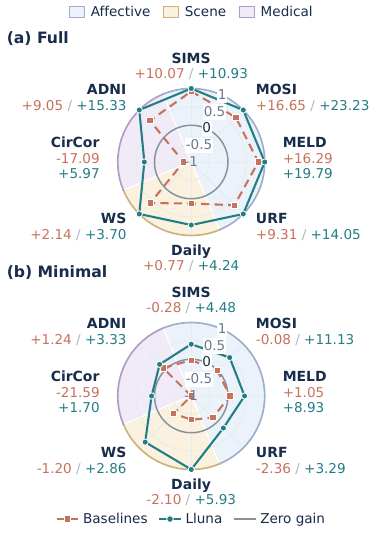}
\captionsetup{skip=4pt,justification=raggedright,singlelinecheck=false}
\caption{Collaboration gains.}
\label{fig:refresh_diag_cooperation}
\end{minipage}\hspace{0.16in}%
\begin{minipage}[t]{3.25in}
\centering
\vspace{0pt}
\includegraphics[width=\linewidth]{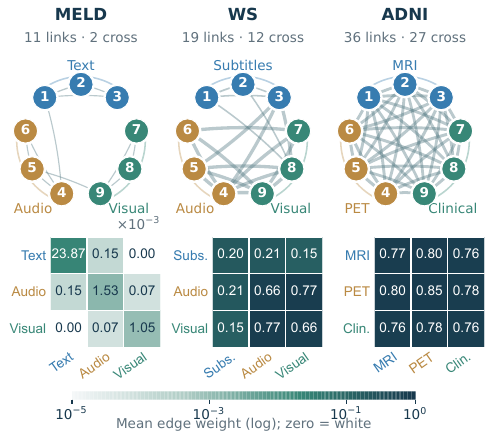}
\captionsetup{skip=4pt,justification=raggedright,singlelinecheck=false}
\caption{Model compatibility across tasks.}
\label{fig:analysis_model_graphs}
\end{minipage}
\vspace{-.1in}
\end{figure}

\subsection{Understanding Model Collaborations}
\label{sec:graph_cases}
\label{sec:refresh_resources}

\textbf{Varying model-pool size.}
We vary \luna's pool with $3,6,9$ models to examine the benefit of additional models (Figure~\ref{fig:model_pool_scaling_main}(b)). Nine models achieve the best observed MM, with relative improvements of $1.77\%$--$27.02\%$ over the three-model pools. Gains vary by task, and intermediate expansions are not uniformly beneficial. Appendix~\ref{app:refresh_resources} reports the scores and pool composition.

\textbf{Collaboration gains.}
\label{sec:refresh_diag_cooperation}
We test whether collaboration remains beneficial as inputs disappear. For each modality, we fix a unimodal reference model and measure collaboration accuracy minus its accuracy, in percentage points, across states retaining that modality. Figure~\ref{fig:refresh_diag_cooperation} compares \luna with the equally weighted mean of seven baselines. From full to minimal inputs, the baseline mean gain declines on all eight tasks and is negative at Minimal on six. \luna retains positive gains on every task and exceeds the baseline mean at both endpoints, supporting beneficial collaboration under severe missingness. Appendix~\ref{app:fixed_anchor_transitions} gives the baseline identities, aggregation protocol, and complete results.

\textbf{Visualization analysis.}
We compare the training-estimated compatibility graphs of MELD, WS, and ADNI (Figure~\ref{fig:analysis_model_graphs}). Sparse cross-modality bridges on MELD contrast with stronger audio--visual coupling on WS and dense connections on ADNI, illustrating task-dependent interaction patterns. Appendix~\ref{app:expert_pools} provides model identities and visualization details.


\subsection{Training Efficiency and Hyperparameter Sensitivity}
\label{sec:refresh_efficiency}
\label{sec:refresh_ablation_pairs}
\label{sec:refresh_ablation_l9}
\label{sec:refresh_parameter_comparison}
\label{sec:sampling_analysis}

\textbf{PCC budgets and hyperparameters.}
To assess sensitivity to consistency coverage and model configuration, we vary the number of adjacent PCC pairs per step and use a nine-configuration design covering PCC strength, private/common width ratio, and graph depth. Larger pair budgets are not uniformly better: SIMS favors four pairs, whereas URF favors nine. Averaged by factor level, the base PCC strength performs best on all four affective tasks, while preferred widths and depths vary. These results support task-specific choices rather than uniformly increasing coverage or capacity. Appendix~\ref{sec:refresh_ablation_details} gives the complete settings and results.

\textbf{Trainable parameters.}
To examine the trainable capacity needed for strong missing-input performance, we compare MM with the number of optimized collaboration parameters, excluding frozen backbones and separately fitted frozen heads. \luna achieves the best MM among the compared methods on all eight tasks. On URF (Figure~\ref{fig:refresh_resource_parameters}(c)), it reaches $70.91\%$ accuracy with $1.29$M parameters, versus $68.46\%$ with $4.25$M for the strongest baseline, LIMSSR. This $2.45$-point gain with $69.8\%$ fewer parameters illustrates effective collaboration with moderate trainable capacity. Appendix~\ref{app:refresh_resources} reports the complete comparison and counting protocol.


\textbf{Chain sampling.}
To test whether fewer PCC comparisons can preserve missing-input performance, we compare task-specific sparse configurations with full-lattice training. Their equal-task mean relative MM change is $-0.28\%$ across eight tasks. On MELD, one sampled chain reduces comparisons per step from nine to two while improving MM by $0.49\%$ relative to full training. Thus, sparse training retains comparable average MM with fewer consistency comparisons. Appendix~\ref{sec:sampling_details} provides all operating points and a controlled sweep varying only the number of sampled chains.

%% file: tables/multiseed_overall.tex
\begin{table}[!t]
\centering
\caption{\textbf{Mean-Missing performance across eight datasets.} Metrics: mean absolute error (MAE), support-weighted F1 (W-F1), accuracy (Acc.), and balanced accuracy (BAcc.). MAE means and Standard Deviations (SDs) are multiplied by 100; F1 and accuracy means are percentages, with SDs in percentage points. Protocol: Appendix~\ref{app:multiseed_protocol}.}
\vspace{-.1in}
\label{tab:multiseed_overall}
\begingroup
\resizebox{1\width}{!}{
\fontsize{8.5}{10}\selectfont
\setlength{\tabcolsep}{0.6pt}
\renewcommand{\arraystretch}{1.07}
\arrayrulecolor{TblInk!65}
\newcommand{\TableOneMethod}[2]{\strut\mbox{#1\if\relax\detokenize{#2}\relax\else\hspace{1pt}{\fontsize{7}{8}\selectfont\color{black!65}#2}\fi}}
\begin{tabular}{@{}>{\raggedright\arraybackslash}p{\dimexpr\linewidth-317.5pt-18\tabcolsep\relax}w{c}{38pt}w{c}{42pt}*{4}{w{c}{34.5pt}}w{c}{37pt}w{c}{34.5pt}>{\columncolor{TblAvg}}w{c}{28pt}@{}}
\toprule
\rowcolor{TblHeader}
\multicolumn{1}{@{}l}{{\color{TblInk}\bfseries Method}} & {\color{TblInk}\bfseries SIMS} & {\color{TblInk}\bfseries MOSI} & {\color{TblInk}\bfseries MELD} & {\color{TblInk}\bfseries URF} & {\color{TblInk}\bfseries Daily} & {\color{TblInk}\bfseries WS} & {\color{TblInk}\bfseries CirCor} & {\color{TblInk}\bfseries ADNI} & {\color{TblInk}\bfseries Avg.} \\
\rowcolor{TblHeader}
 & MAE $\downarrow$ & MAE $\downarrow$ & W-F1 $\uparrow$ & Acc. $\uparrow$ & Acc. $\uparrow$ & Acc. $\uparrow$ & BAcc. $\uparrow$ & BAcc. $\uparrow$ & rank $\downarrow$ \\
\midrule
\rowcolor{TblBlue}
\multicolumn{10}{@{}l@{}}{\rule[-2pt]{0pt}{14pt}\hspace{3pt}\color{TblInk}\bfseries\scshape General collaboration} \\
\addlinespace[3pt]
\TableOneMethod{Mean ens.}{} & \mbox{53.24\hspace{0.4pt}{\fontsize{6}{7}\selectfont\color{black!75}\textpm 0.02}} & \mbox{145.71\hspace{0.4pt}{\fontsize{6}{7}\selectfont\color{black!75}\textpm 0.25}} & \mbox{31.31\hspace{0.4pt}{\fontsize{6}{7}\selectfont\color{black!75}\textpm 0.00}} & \mbox{61.26\hspace{0.4pt}{\fontsize{6}{7}\selectfont\color{black!75}\textpm 0.54}} & \mbox{33.16\hspace{0.4pt}{\fontsize{6}{7}\selectfont\color{black!75}\textpm 0.81}} & \mbox{31.10\hspace{0.4pt}{\fontsize{6}{7}\selectfont\color{black!75}\textpm 0.12}} & \mbox{50.00\hspace{0.4pt}{\fontsize{6}{7}\selectfont\color{black!75}\textpm 0.00}} & \mbox{66.88\hspace{0.4pt}{\fontsize{6}{7}\selectfont\color{black!75}\textpm 6.12}} & $12.06$ \\[2pt]
\TableOneMethod{Logit stack}{\citeyearpar{wolpert1992stacking}} & \mbox{49.62\hspace{0.4pt}{\fontsize{6}{7}\selectfont\color{black!75}\textpm 0.93}} & \mbox{144.12\hspace{0.4pt}{\fontsize{6}{7}\selectfont\color{black!75}\textpm 2.54}} & \mbox{22.70\hspace{0.4pt}{\fontsize{6}{7}\selectfont\color{black!75}\textpm 0.57}} & \mbox{60.92\hspace{0.4pt}{\fontsize{6}{7}\selectfont\color{black!75}\textpm 1.07}} & \mbox{30.11\hspace{0.4pt}{\fontsize{6}{7}\selectfont\color{black!75}\textpm 3.73}} & \mbox{\underline{31.20}\hspace{0.4pt}{\fontsize{6}{7}\selectfont\color{black!75}\textpm 0.15}} & \mbox{51.95\hspace{0.4pt}{\fontsize{6}{7}\selectfont\color{black!75}\textpm 2.20}} & \mbox{71.92\hspace{0.4pt}{\fontsize{6}{7}\selectfont\color{black!75}\textpm 3.20}} & $13.25$ \\[2pt]
\TableOneMethod{Concat MLP}{} & \mbox{42.27\hspace{0.4pt}{\fontsize{6}{7}\selectfont\color{black!75}\textpm 0.73}} & \mbox{130.21\hspace{0.4pt}{\fontsize{6}{7}\selectfont\color{black!75}\textpm 4.21}} & \mbox{44.13\hspace{0.4pt}{\fontsize{6}{7}\selectfont\color{black!75}\textpm 2.98}} & \mbox{64.54\hspace{0.4pt}{\fontsize{6}{7}\selectfont\color{black!75}\textpm 1.37}} & \mbox{31.19\hspace{0.4pt}{\fontsize{6}{7}\selectfont\color{black!75}\textpm 1.85}} & \mbox{30.52\hspace{0.4pt}{\fontsize{6}{7}\selectfont\color{black!75}\textpm 0.65}} & \mbox{74.73\hspace{0.4pt}{\fontsize{6}{7}\selectfont\color{black!75}\textpm 1.89}} & \mbox{\underline{79.71}\hspace{0.4pt}{\fontsize{6}{7}\selectfont\color{black!75}\textpm 3.03}} & $8.00$ \\[2pt]
\TableOneMethod{Attn. pool}{} & \mbox{38.87\hspace{0.4pt}{\fontsize{6}{7}\selectfont\color{black!75}\textpm 0.75}} & \mbox{115.21\hspace{0.4pt}{\fontsize{6}{7}\selectfont\color{black!75}\textpm 6.28}} & \mbox{49.93\hspace{0.4pt}{\fontsize{6}{7}\selectfont\color{black!75}\textpm 0.88}} & \mbox{66.60\hspace{0.4pt}{\fontsize{6}{7}\selectfont\color{black!75}\textpm 0.98}} & \mbox{30.00\hspace{0.4pt}{\fontsize{6}{7}\selectfont\color{black!75}\textpm 1.10}} & \mbox{30.40\hspace{0.4pt}{\fontsize{6}{7}\selectfont\color{black!75}\textpm 0.62}} & \mbox{72.61\hspace{0.4pt}{\fontsize{6}{7}\selectfont\color{black!75}\textpm 4.76}} & \mbox{77.75\hspace{0.4pt}{\fontsize{6}{7}\selectfont\color{black!75}\textpm 3.90}} & $7.62$ \\[2pt]
\TableOneMethod{Soft MoE}{\citeyearpar{jacobs1991mixtures}} & \mbox{58.21\hspace{0.4pt}{\fontsize{6}{7}\selectfont\color{black!75}\textpm 4.88}} & \mbox{133.86\hspace{0.4pt}{\fontsize{6}{7}\selectfont\color{black!75}\textpm 13.86}} & \mbox{47.11\hspace{0.4pt}{\fontsize{6}{7}\selectfont\color{black!75}\textpm 1.33}} & \mbox{66.34\hspace{0.4pt}{\fontsize{6}{7}\selectfont\color{black!75}\textpm 0.49}} & \mbox{29.83\hspace{0.4pt}{\fontsize{6}{7}\selectfont\color{black!75}\textpm 2.02}} & \mbox{30.22\hspace{0.4pt}{\fontsize{6}{7}\selectfont\color{black!75}\textpm 0.99}} & \mbox{70.28\hspace{0.4pt}{\fontsize{6}{7}\selectfont\color{black!75}\textpm 3.54}} & \mbox{73.54\hspace{0.4pt}{\fontsize{6}{7}\selectfont\color{black!75}\textpm 2.07}} & $11.38$ \\[2pt]
\TableOneMethod{FoE}{\citeyearpar{wang2024foe}} & \mbox{53.88\hspace{0.4pt}{\fontsize{6}{7}\selectfont\color{black!75}\textpm 2.27}} & \mbox{148.02\hspace{0.4pt}{\fontsize{6}{7}\selectfont\color{black!75}\textpm 7.96}} & \mbox{31.31\hspace{0.4pt}{\fontsize{6}{7}\selectfont\color{black!75}\textpm 0.00}} & \mbox{52.39\hspace{0.4pt}{\fontsize{6}{7}\selectfont\color{black!75}\textpm 2.28}} & \mbox{31.81\hspace{0.4pt}{\fontsize{6}{7}\selectfont\color{black!75}\textpm 1.92}} & \mbox{29.75\hspace{0.4pt}{\fontsize{6}{7}\selectfont\color{black!75}\textpm 1.81}} & \mbox{49.92\hspace{0.4pt}{\fontsize{6}{7}\selectfont\color{black!75}\textpm 0.19}} & \mbox{64.96\hspace{0.4pt}{\fontsize{6}{7}\selectfont\color{black!75}\textpm 5.84}} & $15.44$ \\[2pt]
\TableOneMethod{Set Trans.}{\citeyearpar{lee2019set}} & \mbox{42.95\hspace{0.4pt}{\fontsize{6}{7}\selectfont\color{black!75}\textpm 0.73}} & \mbox{138.54\hspace{0.4pt}{\fontsize{6}{7}\selectfont\color{black!75}\textpm 12.45}} & \mbox{48.32\hspace{0.4pt}{\fontsize{6}{7}\selectfont\color{black!75}\textpm 1.38}} & \mbox{63.76\hspace{0.4pt}{\fontsize{6}{7}\selectfont\color{black!75}\textpm 0.89}} & \mbox{29.66\hspace{0.4pt}{\fontsize{6}{7}\selectfont\color{black!75}\textpm 2.39}} & \mbox{29.03\hspace{0.4pt}{\fontsize{6}{7}\selectfont\color{black!75}\textpm 2.09}} & \mbox{71.32\hspace{0.4pt}{\fontsize{6}{7}\selectfont\color{black!75}\textpm 3.00}} & \mbox{73.25\hspace{0.4pt}{\fontsize{6}{7}\selectfont\color{black!75}\textpm 5.85}} & $12.12$ \\[2pt]
\TableOneMethod{CLARGA}{\citeyearpar{patapati2026clarga}} & \mbox{39.99\hspace{0.4pt}{\fontsize{6}{7}\selectfont\color{black!75}\textpm 0.34}} & \mbox{122.90\hspace{0.4pt}{\fontsize{6}{7}\selectfont\color{black!75}\textpm 2.63}} & \mbox{52.48\hspace{0.4pt}{\fontsize{6}{7}\selectfont\color{black!75}\textpm 0.81}} & \mbox{66.33\hspace{0.4pt}{\fontsize{6}{7}\selectfont\color{black!75}\textpm 0.67}} & \mbox{29.49\hspace{0.4pt}{\fontsize{6}{7}\selectfont\color{black!75}\textpm 0.95}} & \mbox{30.73\hspace{0.4pt}{\fontsize{6}{7}\selectfont\color{black!75}\textpm 0.27}} & \mbox{73.63\hspace{0.4pt}{\fontsize{6}{7}\selectfont\color{black!75}\textpm 4.83}} & \mbox{77.00\hspace{0.4pt}{\fontsize{6}{7}\selectfont\color{black!75}\textpm 3.38}} & $8.25$ \\[2pt]
\rowcolor{TblTeal}
\multicolumn{10}{@{}l@{}}{\rule[-2pt]{0pt}{14pt}\hspace{3pt}\color{TblInk}\bfseries\scshape Missing-modality methods} \\
\addlinespace[3pt]
\TableOneMethod{GIML}{\citeyearpar{meng2026giml}} & \mbox{68.09\hspace{0.4pt}{\fontsize{6}{7}\selectfont\color{black!75}\textpm 20.73}} & \mbox{181.62\hspace{0.4pt}{\fontsize{6}{7}\selectfont\color{black!75}\textpm 96.15}} & \mbox{33.72\hspace{0.4pt}{\fontsize{6}{7}\selectfont\color{black!75}\textpm 6.49}} & \mbox{53.96\hspace{0.4pt}{\fontsize{6}{7}\selectfont\color{black!75}\textpm 3.06}} & \mbox{32.94\hspace{0.4pt}{\fontsize{6}{7}\selectfont\color{black!75}\textpm 1.29}} & \mbox{30.92\hspace{0.4pt}{\fontsize{6}{7}\selectfont\color{black!75}\textpm 0.18}} & \mbox{74.83\hspace{0.4pt}{\fontsize{6}{7}\selectfont\color{black!75}\textpm 1.26}} & \mbox{72.88\hspace{0.4pt}{\fontsize{6}{7}\selectfont\color{black!75}\textpm 2.67}} & $11.38$ \\[2pt]
\TableOneMethod{MARS}{\citeyearpar{baek2026mars}} & \mbox{40.83\hspace{0.4pt}{\fontsize{6}{7}\selectfont\color{black!75}\textpm 3.79}} & \mbox{132.58\hspace{0.4pt}{\fontsize{6}{7}\selectfont\color{black!75}\textpm 16.83}} & \mbox{51.94\hspace{0.4pt}{\fontsize{6}{7}\selectfont\color{black!75}\textpm 0.30}} & \mbox{67.47\hspace{0.4pt}{\fontsize{6}{7}\selectfont\color{black!75}\textpm 0.50}} & \mbox{31.47\hspace{0.4pt}{\fontsize{6}{7}\selectfont\color{black!75}\textpm 1.61}} & \mbox{30.35\hspace{0.4pt}{\fontsize{6}{7}\selectfont\color{black!75}\textpm 0.58}} & \mbox{63.48\hspace{0.4pt}{\fontsize{6}{7}\selectfont\color{black!75}\textpm 10.72}} & \mbox{76.62\hspace{0.4pt}{\fontsize{6}{7}\selectfont\color{black!75}\textpm 1.58}} & $8.38$ \\[2pt]
\TableOneMethod{LIMSSR}{\citeyearpar{xu2026limssr}} & \mbox{\underline{35.90}\hspace{0.4pt}{\fontsize{6}{7}\selectfont\color{black!75}\textpm 0.21}} & \mbox{102.34\hspace{0.4pt}{\fontsize{6}{7}\selectfont\color{black!75}\textpm 0.56}} & \mbox{\underline{53.98}\hspace{0.4pt}{\fontsize{6}{7}\selectfont\color{black!75}\textpm 0.33}} & \mbox{\underline{68.86}\hspace{0.4pt}{\fontsize{6}{7}\selectfont\color{black!75}\textpm 0.82}} & \mbox{28.59\hspace{0.4pt}{\fontsize{6}{7}\selectfont\color{black!75}\textpm 2.34}} & \mbox{29.91\hspace{0.4pt}{\fontsize{6}{7}\selectfont\color{black!75}\textpm 0.84}} & \mbox{65.98\hspace{0.4pt}{\fontsize{6}{7}\selectfont\color{black!75}\textpm 13.55}} & \mbox{73.58\hspace{0.4pt}{\fontsize{6}{7}\selectfont\color{black!75}\textpm 3.09}} & $7.75$ \\[2pt]
\TableOneMethod{SimMLM}{\citeyearpar{li2025simmlm}} & \mbox{39.64\hspace{0.4pt}{\fontsize{6}{7}\selectfont\color{black!75}\textpm 0.43}} & \mbox{116.10\hspace{0.4pt}{\fontsize{6}{7}\selectfont\color{black!75}\textpm 3.07}} & \mbox{53.09\hspace{0.4pt}{\fontsize{6}{7}\selectfont\color{black!75}\textpm 0.25}} & \mbox{68.05\hspace{0.4pt}{\fontsize{6}{7}\selectfont\color{black!75}\textpm 0.43}} & \mbox{30.17\hspace{0.4pt}{\fontsize{6}{7}\selectfont\color{black!75}\textpm 1.47}} & \mbox{30.93\hspace{0.4pt}{\fontsize{6}{7}\selectfont\color{black!75}\textpm 0.43}} & \mbox{\underline{76.60}\hspace{0.4pt}{\fontsize{6}{7}\selectfont\color{black!75}\textpm 2.30}} & \mbox{79.00\hspace{0.4pt}{\fontsize{6}{7}\selectfont\color{black!75}\textpm 2.93}} & $\underline{5.00}$ \\[2pt]
\TableOneMethod{ShaSpec}{\citeyearpar{wang2023shaspec}} & \mbox{37.39\hspace{0.4pt}{\fontsize{6}{7}\selectfont\color{black!75}\textpm 0.27}} & \mbox{116.82\hspace{0.4pt}{\fontsize{6}{7}\selectfont\color{black!75}\textpm 3.38}} & \mbox{52.92\hspace{0.4pt}{\fontsize{6}{7}\selectfont\color{black!75}\textpm 0.48}} & \mbox{67.25\hspace{0.4pt}{\fontsize{6}{7}\selectfont\color{black!75}\textpm 0.70}} & \mbox{30.79\hspace{0.4pt}{\fontsize{6}{7}\selectfont\color{black!75}\textpm 1.15}} & \mbox{30.95\hspace{0.4pt}{\fontsize{6}{7}\selectfont\color{black!75}\textpm 0.23}} & \mbox{74.59\hspace{0.4pt}{\fontsize{6}{7}\selectfont\color{black!75}\textpm 1.16}} & \mbox{78.08\hspace{0.4pt}{\fontsize{6}{7}\selectfont\color{black!75}\textpm 3.12}} & $5.50$ \\[2pt]
\TableOneMethod{Flex-MoE}{\citeyearpar{yun2024flexmoe}} & \mbox{43.25\hspace{0.4pt}{\fontsize{6}{7}\selectfont\color{black!75}\textpm 7.98}} & \mbox{125.34\hspace{0.4pt}{\fontsize{6}{7}\selectfont\color{black!75}\textpm 7.87}} & \mbox{50.35\hspace{0.4pt}{\fontsize{6}{7}\selectfont\color{black!75}\textpm 1.47}} & \mbox{67.02\hspace{0.4pt}{\fontsize{6}{7}\selectfont\color{black!75}\textpm 1.27}} & \mbox{30.34\hspace{0.4pt}{\fontsize{6}{7}\selectfont\color{black!75}\textpm 2.33}} & \mbox{29.82\hspace{0.4pt}{\fontsize{6}{7}\selectfont\color{black!75}\textpm 0.80}} & \mbox{68.34\hspace{0.4pt}{\fontsize{6}{7}\selectfont\color{black!75}\textpm 4.85}} & \mbox{73.88\hspace{0.4pt}{\fontsize{6}{7}\selectfont\color{black!75}\textpm 2.36}} & $9.62$ \\[2pt]
\rowcolor{TblSand}
\multicolumn{10}{@{}l@{}}{\rule[-2pt]{0pt}{14pt}\hspace{3pt}\color{TblInk}\bfseries\scshape Omni models} \\
\addlinespace[3pt]
\TableOneMethod{Q2.5-O-3B}{\citeyearpar{xu2025qwenomni}} & 171.86 & 161.69 & 22.53 & 55.00 & 35.88 & 28.49 & 61.50 & 52.92 & $15.62$ \\[2pt]
\TableOneMethod{Q2.5-O-7B}{\citeyearpar{xu2025qwenomni}} & 94.56 & \underline{99.67} & 39.58 & 53.91 & \underline{36.72} & 29.71 & 59.80 & 53.75 & $12.00$ \\[2pt]
\TableOneMethod{InterO-8B}{\citeyearpar{tong2025interactiveomni}} & 182.38 & 135.56 & 26.52 & 51.36 & 24.86 & 22.94 & 58.83 & 53.54 & $18.38$ \\[2pt]
\TableOneMethod{OmniV-9B}{\citeyearpar{ye2025omnivinci}} & 106.25 & 104.34 & 34.39 & 56.17 & 31.92 & 25.63 & 75.17 & 50.00 & $12.38$ \\[2pt]
\TableOneMethod{Gemma4-12B}{\citeyearpar{gemmateam2026gemma4}} & 366.55 & 193.88 & 30.53 & 52.93 & 32.49 & 30.06 & 52.19 & 50.83 & $16.50$ \\[2pt]
\TableOneMethod{StreamO-8B}{\citeyearpar{zhang2025streamomni}} & 169.47 & 387.08 & 17.23 & 51.96 & 25.71 & 28.24 & 55.44 & 51.04 & $19.38$ \\[2pt]
\addlinespace[2pt]
\midrule
\rowcolor{TblOurs}
\multicolumn{1}{@{}l}{\color{TblInk}\bfseries\luna (ours)} & \mbox{\textbf{34.24}\hspace{0.4pt}{\fontsize{6}{7}\selectfont\color{black!75}\textpm 0.30}} & \mbox{\textbf{96.71}\hspace{0.4pt}{\fontsize{6}{7}\selectfont\color{black!75}\textpm 0.63}} & \mbox{\textbf{56.10}\hspace{0.4pt}{\fontsize{6}{7}\selectfont\color{black!75}\textpm 0.21}} & \mbox{\textbf{70.62}\hspace{0.4pt}{\fontsize{6}{7}\selectfont\color{black!75}\textpm 0.23}} & \mbox{\textbf{37.29}\hspace{0.4pt}{\fontsize{6}{7}\selectfont\color{black!75}\textpm 0.96}} & \mbox{\textbf{33.90}\hspace{0.4pt}{\fontsize{6}{7}\selectfont\color{black!75}\textpm 0.30}} & \mbox{\textbf{81.76}\hspace{0.4pt}{\fontsize{6}{7}\selectfont\color{black!75}\textpm 0.55}} & \mbox{\textbf{82.21}\hspace{0.4pt}{\fontsize{6}{7}\selectfont\color{black!75}\textpm 1.26}} & $\mathbf{1.00}$ \\[2pt]
\midrule
\rowcolor{TblHeader}
\multicolumn{1}{@{}l}{\color{TblInk}$\Delta_{\mathrm{rel}}$ (\%)} & $+4.61$ & $+2.98$ & $+3.91$ & $+2.56$ & $+1.54$ & $+8.65$ & $+6.73$ & $+3.14$ & $\mathbf{+4.26\%}$ \\
\bottomrule
\end{tabular}
}
\endgroup
\vspace{-.15in}
\end{table}

%% file: tables/supplement_missing_input_main.tex
\begingroup
\fontsize{8.5}{10.1}\selectfont
\tiny
\setlength{\tabcolsep}{1pt}
\renewcommand{\arraystretch}{1}
\arrayrulecolor{TblInk!65}

\begin{tabular*}{\linewidth}{
@{\extracolsep{\fill}}
>{\raggedright\arraybackslash}p{\dimexpr.35\linewidth-8\tabcolsep\relax}
*{3}{>{\raggedleft\arraybackslash}p{.155\linewidth}}
@{}
}
\toprule[.7pt]
\rowcolor{TblHeader}\textbf{Method} & \textbf{SIMS}$\downarrow$ & \textbf{WS}$\uparrow$ & \textbf{ADNI}$\uparrow$ \\
\midrule[.35pt]
Concat MLP & 37.90 & 32.64 & 70.21 \\
Attn. pool & 35.46 & 32.15 & 78.13 \\
Soft MoE & 38.08 & 31.17 & 80.83 \\
FoE & 47.59 & 31.03 & 63.33 \\
Set Trans. & 37.42 & 31.17 & 82.08 \\
CLARGA & 40.18 & 31.17 & 80.63 \\
I$^2$MoE & 37.31 & 30.27 & 77.71 \\
\midrule[.35pt]
Qwen2.5-Omni 3B & 51.55 & 30.51 & 66.88 \\
Qwen2.5-Omni 7B & 48.61 & 30.54 & 69.79 \\
Stream-Omni & 65.71 & 29.92 & 74.17 \\
\midrule[.35pt]
\rowcolor{TblOurs}\textbf{\luna} & \textbf{34.21} & \textbf{33.82} & \textbf{84.17} \\
\bottomrule[.7pt]
\end{tabular*}
\vspace{-.1in}
\endgroup

%% file: tables/refresh_ablation_core.tex
\begin{table}[!t]
\centering
\caption{\textbf{Ablation study.} MM follows the metrics in Table~\ref{tab:multiseed_overall}. The final $\Delta_{\mathrm{rel}}$ row averages relative changes across eight tasks compared with \luna; negative values indicate degradation.}
\label{tab:refresh_ablation_core}
\vspace{-.1in}
\resizebox{.8\textwidth}{!}{
\begingroup\fontsize{8.5}{10}\selectfont\setlength{\tabcolsep}{1.2pt}
\renewcommand{\arraystretch}{1.12}
\begin{tabular}{@{}w{l}{32pt}>{\columncolor{TblOurs}}w{r}{26pt}@{\hspace{6pt}}w{r}{26pt}w{r}{26pt}w{r}{26pt}@{\hspace{6pt}}w{r}{26pt}@{\hspace{6pt}}w{r}{26pt}w{r}{26pt}w{r}{26pt}w{r}{26pt}@{\hspace{6pt}}w{r}{26pt}w{r}{26pt}@{}}
\toprule
 & \multicolumn{1}{c}{} & \multicolumn{3}{c}{\textbf{(a) Paths}} & \textbf{(b) Graph} & \multicolumn{4}{c}{\textbf{(c) PCC}} & \multicolumn{2}{c}{\textbf{(d) Target}} \\
\cmidrule(lr){3-5}\cmidrule(lr){6-6}\cmidrule(lr){7-10}\cmidrule(l){11-12}
\rowcolor{TblHeader}Dataset & \luna & \makebox[26pt][c]{\fontsize{7.5}{8.5}\selectfont \shortstack{w/o\\Private}} & \makebox[26pt][c]{\fontsize{7.5}{8.5}\selectfont \shortstack{w/o\\Common}} & \makebox[26pt][c]{\fontsize{7.5}{8.5}\selectfont \shortstack{Equal\\addition}} & \makebox[26pt][c]{\fontsize{7.5}{8.5}\selectfont \shortstack{Uniform\\Graph}} & \makebox[26pt][c]{\fontsize{7.5}{8.5}\selectfont \shortstack{w/o\\PCC}} & \makebox[26pt][c]{\fontsize{7.5}{8.5}\selectfont \shortstack{Full-to-\\Minimal}} & \makebox[26pt][c]{\fontsize{7.5}{8.5}\selectfont \shortstack{All-\\Pairs}} & \makebox[26pt][c]{\fontsize{7.5}{8.5}\selectfont \shortstack{w/o Stop-\\gradient}} & \makebox[26pt][c]{\fontsize{7.5}{8.5}\selectfont \shortstack{PCC-\\Private}} & \makebox[26pt][c]{\fontsize{7.5}{8.5}\selectfont \shortstack{PCC-\\Both}} \\
\midrule
SIMS$\downarrow$ & \underline{34.21} & 35.35 & 35.52 & \textbf{33.88} & 34.77 & 34.87 & 34.74 & 34.92 & 35.21 & 34.59 & 35.18 \\
MOSI$\downarrow$ & \textbf{98.34} & 101.41 & 101.40 & 100.21 & 100.74 & 99.87 & \underline{99.10} & 99.40 & 103.17 & 99.95 & 99.62 \\
MELD$\uparrow$ & \underline{56.08} & 55.49 & 55.15 & \textbf{56.16} & 55.20 & 55.96 & 54.58 & 55.08 & 55.93 & 55.96 & 55.08 \\
URF$\uparrow$ & \textbf{70.61} & 68.01 & 69.62 & 70.12 & 69.62 & 70.00 & \underline{70.22} & 70.17 & 70.04 & 70.00 & 69.99 \\
\addlinespace[2pt]
Daily$\uparrow$ & \textbf{38.42} & 37.57 & 36.16 & 28.81 & \underline{38.14} & \textbf{38.42} & \textbf{38.42} & \underline{38.14} & 28.53 & \underline{38.14} & 37.57 \\
WS$\uparrow$ & \textbf{33.82} & 32.50 & 32.25 & 33.30 & \underline{33.51} & 32.15 & 33.30 & 32.46 & 32.54 & 32.36 & 32.29 \\
\addlinespace[2pt]
CirCor$\uparrow$ & \textbf{82.56} & 71.88 & 78.12 & 74.89 & 78.71 & 78.71 & 76.92 & 81.93 & 77.42 & \underline{82.26} & 82.16 \\
ADNI$\uparrow$ & \textbf{84.17} & 65.63 & 79.38 & \underline{83.96} & \textbf{84.17} & 80.83 & 81.04 & \underline{83.96} & \textbf{84.17} & 82.50 & 83.33 \\
\midrule
\rowcolor{TblAvg}$\Delta_{\mathrm{rel}}$ (\%) & $\mathbf{0.00}$ & $-6.53$ & $-3.95$ & $-4.70$ & $-1.67$ & $-2.27$ & $-2.20$ & $-1.42$ & $-5.58$ & $-1.40$ & $-1.88$ \\
\bottomrule
\end{tabular}
\endgroup
}
\vspace{-.1in}
\end{table}

%% file: section/conclusion.tex
\section{Conclusion}
We study collaboration among pretrained unimodal models under missing modalities and propose \luna, which combines direct access to modality-specific features with residual common collaboration and progressive cross-state consistency.
Across eight datasets, \luna improves Mean-Missing performance over the strongest evaluated baseline on every task, averaging $4.26\%$ relative improvement.
It also maintains positive mean collaboration gains under minimal inputs, supporting robust collaboration as available evidence decreases.
We discuss limitations in Appendix~\ref{app:limitations}.



%% file: section/appendix.tex
\begingroup
\renewcommand{\contentsname}{Appendix contents}
\setcounter{tocdepth}{2}
\small
\setlength{\parskip}{0pt}
\pdfbookmark[0]{Appendix contents}{appendix.contents}
\tableofcontents
\endgroup
\clearpage
\clearpage
\section{Complete Five-Run Evaluation}
\label{app:multiseed_results}
\subsection{Selection and aggregation protocol}
\label{app:multiseed_protocol}
Table~\ref{tab:multiseed_overall} reports five training runs for each trainable method. Every retrained baseline uses seeds $11,\allowbreak 47,\allowbreak 103,\allowbreak 211,\allowbreak 389$. Each dataset uses one adopted \luna configuration across its five runs; all five runs contribute to the reported mean and standard deviation. Omni baselines use fixed backbone/head checkpoints with deterministic inference; their repeated evaluations reproduce the same score, reported without a training SD.

\paragraph{Mean-Missing and relative improvement.}
Let $\mathcal V_d^{-}$ be the legal, nonempty, non-full availability states of dataset $d$, and let $m_{d,c}(S)$ be the task score for run $c$ in state $S$. Mean-Missing is computed within each run before averaging runs:
\begin{equation}
\operatorname{MM}_{d,c}
=\frac{1}{|\mathcal V_d^{-}|}\sum_{S\in\mathcal V_d^{-}}m_{d,c}(S).
\end{equation}
Each main-table cell reports the corresponding method's MM result. Relative improvement uses the best baseline mean for that dataset. For an individual view, the reference is the strongest baseline mean in that view. Thus, baseline MM is not an average of the different per-view winners. For the displayed five-run mean $\bar m$ and baseline score $B$, relative improvement $\Delta_{\mathrm{rel}}$ (in percent) is
\begin{equation}
\Delta_{\mathrm{rel}}=\begin{cases}
100(B-\bar m)/B, & \text{MAE (lower is better)},\\
100(\bar m-B)/B, & \text{F1 or accuracy (higher is better)}.
\end{cases}
\end{equation}
Positive $\Delta_{\mathrm{rel}}$ indicates improvement; a negative value indicates a shortfall relative to the strongest baseline. In all result tables, MAE means and SDs are multiplied by 100. F1 and accuracy means are percentages, with SDs in percentage points. Every mean and SD uses two decimal places; for MAE this retains the precision of four decimals in native units. The relative improvement reported in the main text averages $\Delta_{\mathrm{rel}}$ equally across the eight datasets, using unrounded values. The last cell of the $\Delta_{\mathrm{rel}}$ row in Table~\ref{tab:multiseed_overall} reports this same mean in percent. The Avg.\ rank column averages each method's rank across the eight datasets. On each dataset, all 21 methods are ranked by their unrounded Mean-Missing means. Rank 1 is best: lower MAE and higher F1 or accuracy receive better ranks. Exact ties receive the average of their occupied positions. We rank the five-run means, not individual runs; every dataset receives equal weight. This summarizes cross-dataset standing without averaging incompatible metrics.

\paragraph{View notation.}
For CH-SIMS v2, CMU-MOSI, MELD, and UR-FUNNY v2, T/A/V denote text/\allowbreak audio/\allowbreak vision. Daily-Omni uses A/V for audio/video. WorldSense uses V/A/S for visual/audio/subtitles. ADNI uses M/P/C for MRI/PET/clinical inputs, including states without clinical inputs. CirCor-4Site uses AV/PV/TV/MV for four auscultation sites; these are site identifiers, not audio/visual modality abbreviations. Required question-and-choice context is retained for the scene-understanding tasks while the listed modalities vary.

\input{tables/multiseed_details}
\section{Main Experimental Setup}
\label{app:experimental_setup}

\subsection{Tasks, splits, and availability states}
SIMS and MOSI are sentiment-regression tasks; MELD and URF are emotion and humor classification tasks. Daily and WS use supervised, video-disjoint QA splits derived from the public benchmarks. CirCor predicts murmur presence from four auscultation sites, and ADNI predicts MCI-to-dementia conversion using MRI, PET, and clinical inputs. The medical tasks use patient- or subject-disjoint splits. The task-specific paragraphs below specify outcome windows, eligibility rules, and preprocessing.

Table~\ref{tab:dataset_protocol} specifies the evaluation population, split unit, and task metric. The same test examples are evaluated under every legal nonempty modality subset. Full inputs are reported separately; Mean-Missing equally weights the remaining states within each run before the five-run aggregation in Appendix~\ref{app:multiseed_protocol}. Accuracy is abbreviated Acc.; balanced accuracy (BAcc.) is the mean of the two class recalls. MELD uses support-weighted F1 over its seven emotion classes.

\begin{table}[!htb]
\centering
\caption{\textbf{Dataset protocols.} Entries count utterances (SIMS/MOSI/MELD), humor examples (URF), question--answer items (Daily/WS), visits (CirCor), or labeled subjects (ADNI). Splits are train / validation / test; state counts include Full. Input abbreviations follow Appendix~\ref{app:multiseed_protocol}.}
\label{tab:dataset_protocol}
\begingroup
\small
\setlength{\tabcolsep}{4pt}
\renewcommand{\arraystretch}{1.13}
\begin{tabular}{@{}lllrll@{}}
\toprule
\rowcolor{TblHeader}
\textbf{Dataset} & \textbf{Inputs} & \textbf{Train / val. / test} & \textbf{States} &
\textbf{Split unit} & \textbf{Task / metric} \\
\midrule
SIMS & T/A/V & 2,722 / 647 / 1,034 & 7 & Source video & Sentiment / MAE \\
MOSI & T/A/V & 1,284 / 229 / 686 & 7 & Source video & Sentiment / MAE \\
MELD & T/A/V & 9,988 / 1,108 / 2,608 & 7 & Dialogue & Emotion / W-F1 \\
URF & T/A/V & 7,614 / 980 / 994 & 7 & Source video & Humor / Acc. \\
\addlinespace[3pt]
Daily & A/V & 836 / 183 / 177 & 3 & Video & Four-choice QA / Acc. \\
WS & V/A/S & 2,214 / 480 / 478 & 7 & Video & 3/4-choice QA / Acc. \\
\addlinespace[3pt]
CirCor & AV/PV/TV/MV & 409 / 88 / 88 & 15 & Patient group & Murmur / BAcc. \\
ADNI & M/P/C & 202 / 32 / 50 & 7 & Subject & MCI conversion / BAcc. \\
\bottomrule
\end{tabular}
\endgroup
\end{table}

\paragraph{Affective tasks.}
SIMS and MOSI predict continuous utterance-level sentiment in $[-1,1]$ and $[-3,3]$, respectively. MELD predicts anger, disgust, fear, joy, neutral, sadness, or surprise; URF predicts funny versus not funny. URF v2 refers to the fixed preprocessing protocol used here. Each task varies text, audio, and vision, giving one Full state, three two-modality states, and three one-modality states.

\paragraph{Derived audio--visual question answering.}
Daily and WS use supervised, video-disjoint derived splits of their public QA resources. Daily contains 479 / 103 / 102 videos; WS contains 1,164 / 249 / 249. The question and candidate answers remain visible as task context in every state. Daily has four choices and two evidence modalities, audio and video. Its locally generated ASR text depends on audio and disappears with it. WS has three or four choices and varies video, raw audio, and official subtitles independently; padded option positions are masked before prediction. Candidate-order prefixes are removed before encoding. Generated captions and answer-bearing annotations are not evidence inputs.

\paragraph{CirCor.}
The four-site derived task retains 585 visits from 556 patient groups with all four auscultation sites and a Present/Absent murmur label; Unknown labels are excluded. Patient groups connect records sharing Patient ID or Additional ID. The train / validation / test splits contain 387 / 84 / 85 groups with zero overlap, and 84 / 18 / 18 positive visits. AV, PV, TV, and MV denote the aortic, pulmonary, tricuspid, and mitral sites. Removing a site removes its waveform and every representation derived from that waveform. Age, sex, height, weight, and pregnancy status remain available as context; diagnostic annotations, identifiers, murmur descriptors, and segmentation labels are excluded from model inputs. Mean-Missing averages all 14 incomplete site combinations.

\paragraph{ADNI.}
The task predicts MCI-to-dementia conversion with subsequent diagnostic confirmation. The frozen subject partition contains 384 / 72 / 101 subjects before outcome eligibility filtering; the labeled subset in Table~\ref{tab:dataset_protocol} contains 41 / 10 / 10 converters. MRI, amyloid PET, ADAS, FAQ, and CDR measurements are matched within a 90-day window. The prediction index is the latest date among these selected measurements, so the outcome clock starts after all input measurements. A converter has a first dementia diagnosis after the index and no later than 36 months plus 90 days (approximately 39 months), followed by a confirming dementia diagnosis on a distinct later date. This cutoff applies to the first diagnosis, not to the confirmation visit, which may occur beyond that window. Stable MCI requires no dementia through the first-diagnosis window and an explicit MCI assessment at or after 36 months minus 90 days. Feature filtering, imputation, standardization, and PCA are fitted only on the 384 training subjects; task heads and compatibility estimates use the training subset eligible under this outcome definition. MRI, PET, and the clinical instruments form three independently removable modalities, including clinical-absent states.

\subsection{Frozen inputs and baseline interfaces}
\paragraph{Main model pools.}
\label{app:expert_pools}
Each task uses a pool of nine frozen model representations or task-specific feature sets. SIMS, MOSI, MELD, URF, WS, Daily, and CirCor share the nine encoders in Table~\ref{tab:graph_expert_pools}(a), with task-specific inputs summarized in Table~\ref{tab:pool_input_mapping}. ADNI uses the MRI and PET embeddings, imaging measurements, and clinical assessment features listed in Table~\ref{tab:graph_expert_pools}(b). Each pool entry has its own adapter.

\begin{table}[!htb]
\centering
\caption{\textbf{Frozen model pools and parameter counts.} (a) Shared encoders for SIMS, MOSI, MELD, URF, WS, Daily, and CirCor. Counts are approximate billions of frozen parameters in the components used for feature extraction under complete inputs, excluding trainable collaboration modules. (b) ADNI inputs. Entry indices are preserved for both pools; Table~\ref{tab:pool_input_mapping} specifies task-specific input mappings.}
\label{tab:graph_expert_pools}
\begingroup
\small
\setlength{\tabcolsep}{4pt}
\renewcommand{\arraystretch}{1.12}
\begin{tabular}{@{}>{\centering\arraybackslash}p{.07\textwidth}>{\centering\arraybackslash}p{.12\textwidth}p{\dimexpr.81\textwidth-6\tabcolsep-50pt\relax}>{\raggedleft\arraybackslash}p{50pt}@{}}
\multicolumn{4}{@{}l}{\textbf{(a) Shared text/audio/vision model pool}} \\[4pt]
\toprule
\rowcolor{TblHeader}
\textbf{Entry} & \textbf{Modality} & \textbf{Model / component used} & \textbf{Params (B)} \\
\midrule
1 & \multirow{3}{*}{Text} & Qwen3-1.7B & 1.721 \\
2 & & SmolLM2-1.7B-Instruct & 1.711 \\
3 & & XLM-RoBERTa-large & 0.560 \\ \midrule
\addlinespace[4pt]
4 & \multirow{3}{*}{Audio} & Whisper-large-v3 (encoder) & 0.637 \\
5 & & WavLM-base-plus & 0.094 \\
6 & & HuBERT-large-ls960-ft & 0.315 \\ \midrule
\addlinespace[4pt]
7 & \multirow{3}{*}{Vision} & DINOv2-giant & 1.136 \\
8 & & SigLIP2-giant-opt-patch16-256 (vision tower) & 1.163 \\
9 & & VideoMAE-base-finetuned-kinetics & 0.087 \\
\midrule
\rowcolor{TblAvg}
\multicolumn{3}{l}{\textbf{Total (9 models)}} & $\approx 7.43$ \\
\bottomrule
\end{tabular}
\par\vspace{8pt}
\begin{tabular}{@{}>{\centering\arraybackslash}p{.07\textwidth}p{\dimexpr.245\textwidth-4\tabcolsep\relax}>{\centering\arraybackslash}p{.07\textwidth}p{\dimexpr.245\textwidth-4\tabcolsep\relax}>{\centering\arraybackslash}p{.07\textwidth}p{\dimexpr.30\textwidth-2\tabcolsep\relax}@{}}
\multicolumn{6}{@{}l}{\textbf{(b) ADNI model representations and measurements}} \\[4pt]
\toprule
\rowcolor{TblHeader}
\textbf{Entry} & \textbf{MRI} & \textbf{Entry} & \textbf{PET} & \textbf{Entry} & \textbf{Clinical} \\
\midrule
1 & DINOv2-giant & 4 & DINOv2-giant & 7 & ADAS-Cog features \\
2 & SigLIP2-giant & 5 & SigLIP2-giant & 8 & FAQ features \\
3 & FS morphometry & 6 & SUVR summaries & 9 & CDR features \\
\bottomrule
\end{tabular}
\endgroup
\end{table}

\begin{table}[!htb]
\centering
\caption{\textbf{Input mapping across all eight datasets.} Entry numbers refer to Table~\ref{tab:graph_expert_pools}. Availability follows the originating input, with task-specific masking described below.}
\label{tab:pool_input_mapping}
\begingroup
\small
\setlength{\tabcolsep}{3pt}
\renewcommand{\arraystretch}{1.15}
\begin{tabular}{@{}p{.21\linewidth}p{.245\linewidth}p{.245\linewidth}p{.245\linewidth}@{}}
\toprule
\rowcolor{TblHeader}
\textbf{Datasets} & \textbf{Inputs to 1--3} & \textbf{Inputs to 4--6} & \textbf{Inputs to 7--9} \\
\midrule
SIMS, MOSI,\newline MELD, URF & Utterance text & Speech audio & Video frames \\
\addlinespace[3pt]
WS & Subtitles & Audio & Video frames \\
\addlinespace[3pt]
Daily & ASR transcripts & Audio & Video frames \\
\addlinespace[3pt]
CirCor & Demographic text & PCG waveforms & PCG log-spectrograms \\
\addlinespace[3pt]
ADNI & MRI images and\newline morphometry & PET images and\newline SUVR measurements & Clinical assessments \\
\bottomrule
\end{tabular}
\endgroup
\end{table}
\FloatBarrier

\paragraph{Input availability.}
SIMS, MOSI, MELD, URF, and WS independently mask their text/subtitle, audio, and visual inputs. Daily's ASR transcripts depend on audio, so its three text and three acoustic streams disappear together when audio is removed; the other three streams depend on video. ADNI groups its streams into MRI, PET, and clinical inputs.

CirCor varies the four auscultation sites while retaining the permitted demographic template. The audio encoders process PCG waveforms and the visual encoders process deterministic PCG log-spectrograms. Each signal stream preserves four site slots; all six signal encoders lose the same slot when that site is removed. Each signal adapter projects each site separately, applies the site mask after projection and site conditioning, and concatenates the four ordered slots with the four-bit mask. This yields one node per encoder. A signal node stays available while any site survives, and all three demographic nodes remain available; every legal CirCor state therefore retains nine graph nodes while their signal contents change.

\paragraph{ADNI feature definitions.}
FS denotes FreeSurfer morphometry, and SUVR denotes standardized uptake value ratio. ADAS, FAQ, and CDR denote ADAS-Cog item/total scores, Functional Activities Questionnaire measures, and Clinical Dementia Rating features.

\paragraph{Graph node indexing.}
For the three case studies in Figure~\ref{fig:analysis_model_graphs}, entries 1--9 of Table~\ref{tab:graph_expert_pools} give the node order: MELD and WS use panel (a), and ADNI uses panel (b). The graph case studies visualize model relations under each task's graph-weighting and pruning settings. Graph weights are fitted on training-partition predictions.

\paragraph{Feature preparation and adapted methods.}
For the affective tasks, text is encoded in nonoverlapping chunks of up to 256 tokens and pooled with token-count weights; 16-kHz audio is encoded in 20-second chunks and pooled with duration weights. Visual encoders use 16 uniformly sampled frames across the clip. The QA pipeline retains eight ordered temporal bins and uses question-and-option-conditioned scoring of the frozen evidence. Missing inputs zero the corresponding feature slots and are accompanied by an explicit availability mask. Same-pool baselines share these cached inputs and boundaries of independently maskable inputs. Their task outputs are adapted to regression, classification, or masked candidate scoring. Each baseline retains its defining collaboration or missing-modality mechanism within this common frozen-feature interface. LIMSSR uses Qwen2.5-0.5B with rank-4 LoRA and task adapters. Frozen unimodal encoders are excluded from trainable-parameter counts.

\paragraph{Baseline training and checkpoint selection.}
Each baseline retains its training procedure after adaptation to the shared frozen-feature interface. Table~\ref{tab:baseline_training_protocol} summarizes the availability states used by the general collaboration and staged missing-modality baselines. For each trainable baseline, we select one checkpoint by validation Mean-Missing and evaluate it across all test states. Mean ens.\ averages the outputs of fixed unimodal task heads and has no trainable collaboration layer.

\begin{table}[!htb]
\centering
\caption{\textbf{Baseline training views.}}
\label{tab:baseline_training_protocol}
\begingroup
\small
\setlength{\tabcolsep}{5pt}
\renewcommand{\arraystretch}{1.12}
\begin{tabular}{@{}p{.34\linewidth}p{\dimexpr.66\linewidth-2\tabcolsep\relax}@{}}
\toprule
\rowcolor{TblHeader}
\textbf{Methods} & \textbf{Training availability states} \\
\midrule
Logit stack, Concat MLP, Attn.\ pool, Soft MoE, FoE, Set Trans. & Complete inputs. \\
\addlinespace[3pt]
CLARGA, LIMSSR, ShaSpec & One uniformly sampled legal nonempty state per batch. \\
\addlinespace[3pt]
SimMLM & Five epochs of complete-input expert training, followed by sampled pairs in which one state's modalities strictly contain the other's. \\
\addlinespace[3pt]
Flex-MoE & Five complete-input generalization epochs, followed by specialization on uniformly sampled incomplete states. \\
\bottomrule
\end{tabular}
\endgroup
\end{table}

GIML retains feature-degradation training, whereas MARS uses adaptive partial-view sampling.

\paragraph{Omni models.}
Each Omni comparator receives the currently available raw evidence through its own input processor, independently of the nine-model feature cache. Unavailable modalities are physically removed before processing. The backbone remains frozen; a final prompt-token representation is standardized and fitted with one task head using complete-input \emph{training examples only}. Classification uses balanced logistic regression with inverse regularization strength 1; regression uses ridge regression with regularization strength 10. The same head evaluates every test state, with no missing-input head training, state-specific heads, or test-time adaptation. Questions and choices remain in QA prompts, with invalid candidate positions masked. CirCor retains permitted demographic context; ADNI inputs use the available image montages and clinical fields. These are frozen-backbone task-head evaluations, not zero-shot answer generation. Repeated evaluation is deterministic and therefore is reported without a training SD.

\paragraph{Graph case-study details.}

To examine how model interactions are organized across modalities, we visualize the fixed model graphs of MELD, WS, and ADNI (Figure~\ref{fig:analysis_model_graphs}). All three graphs use identical node coordinates. MELD and WS also share model identities at each position, with text interpreted as subtitles in WS; ADNI uses the same three-group layout for MRI, PET, and clinical models. Edge weights encode dataset-level prediction compatibility estimated on the training partition. Positive undirected edges share one width scale; self-loops are excluded and weights precede incoming normalization, without additional thresholding. The matrices average between-group weights including zeros, using shared logarithmic colors with zero shown in white. MELD annotations are scaled by $10^3$.

MELD has nine within-modality links and two weak cross-modality bridges: Qwen3--Whisper and WavLM--VideoMAE. These bridges connect the three modality groups into one component. WS has nineteen links, twelve of which cross modalities, with stronger mean audio--visual coupling ($0.77$) than subtitle--audio ($0.21$) or subtitle--visual ($0.15$) coupling. ADNI retains all thirty-six possible links, with cross-modality means of $0.76$--$0.80$. These patterns show dataset-dependent interaction structures, ranging from sparse modality bridges to dense cross-modality connections. The private path remains available to every active model irrespective of these connections. Appendix~\ref{app:refresh_sampling_modality} analyzes modality contributions and utility interactions of the resulting predictor, complementing the prediction compatibility used here to construct graph weights.
\FloatBarrier

\subsection{Main-run optimization settings}
Table~\ref{tab:main_training_settings} summarizes the main training settings. Within each task, the five runs share all settings except the random seeds for initialization, training, and chain sampling. AdamW updates the adapters and collaboration layers while keeping the unimodal encoders fixed. Training uses task-specific availability-state weights, whereas evaluation averages all non-full states equally. Task-specific configurations specify branch-wise learning-rate and regularization multipliers, state weights, and loss schedules. For each run, one checkpoint is selected by validation Mean-Missing and used across all availability states.

\begin{table}[!htb]
\centering
\caption{\textbf{\luna main-run settings.} Private/common widths are $d_{\mathrm{priv}}/d_{\mathrm{com}}$; $L$ is graph depth and $K$ the chain budget per step. Learning rates and loss coefficients are rounded to three significant digits. Patience is in epochs.}
\label{tab:main_training_settings}
\begingroup
\small
\setlength{\tabcolsep}{4pt}
\renewcommand{\arraystretch}{1.12}
\begin{tabular}{@{}p{\dimexpr.26\linewidth-12\tabcolsep\relax}>{\raggedleft\arraybackslash}p{.11\linewidth}>{\raggedleft\arraybackslash}p{.11\linewidth}>{\raggedleft\arraybackslash}p{.06\linewidth}>{\raggedleft\arraybackslash}p{.14\linewidth}>{\raggedleft\arraybackslash}p{.22\linewidth}>{\raggedleft\arraybackslash}p{.10\linewidth}@{}}
\multicolumn{7}{@{}l}{\textbf{(a) Architecture and training objectives}} \\[3pt]
\toprule
\rowcolor{TblHeader}
\textbf{Dataset} & $d_{\mathrm{priv}}$ & $d_{\mathrm{com}}$ & $L$ & $\lambda_{\mathrm{com}}$ & $\lambda_{\mathrm{PCC}}$ & $K$ \\
\midrule
SIMS & 64 & 192 & 2 & $0.300$ & $1.58\!\times\!10^{-2}$ & 2 \\
MOSI & 144 & 576 & 2 & $0.278$ & $9.60\!\times\!10^{-2}$ & 2 \\
MELD & 128 & 384 & 2 & $0.0860$ & $6.25\!\times\!10^{-4}$ & 1 \\
URF & 80 & 160 & 1 & $0.0125$ & $1.00\!\times\!10^{-2}$ & 1 \\
Daily & 128 & 256 & 1 & $0.0500$ & $1.81\!\times\!10^{-3}$ & 1 \\
WS & 128 & 768 & 1 & $0.0500$ & $1.00\!\times\!10^{-3}$ & 2 \\
CirCor & 320 & 640 & 1 & $0.200$ & $1.00\!\times\!10^{-1}$ & 2 \\
ADNI & 96 & 384 & 1 & $0.0625$ & $9.23\!\times\!10^{-3}$ & 2 \\
\bottomrule
\end{tabular}
\par\vspace{7pt}
\begin{tabular}{@{}p{\dimexpr.23\linewidth-10\tabcolsep\relax}>{\raggedleft\arraybackslash}p{.09\linewidth}>{\raggedleft\arraybackslash}p{.23\linewidth}>{\raggedleft\arraybackslash}p{.20\linewidth}>{\raggedleft\arraybackslash}p{.14\linewidth}>{\raggedleft\arraybackslash}p{.11\linewidth}@{}}
\multicolumn{6}{@{}l}{\textbf{(b) Optimization}} \\[3pt]
\toprule
\rowcolor{TblHeader}
\textbf{Dataset} & \textbf{Batch} & \textbf{Learning rate} & \textbf{Weight decay} & \textbf{Max epochs} & \textbf{Patience} \\
\midrule
SIMS & 128 & $3.26\!\times\!10^{-4}$ & $10^{-4}$ & 105 & 25 \\
MOSI & 192 & $3.22\!\times\!10^{-4}$ & $10^{-3}$ & 60 & 15 \\
MELD & 200 & $8.88\!\times\!10^{-5}$ & $2.5\!\times\!10^{-5}$ & 45 & 10 \\
URF & 384 & $7.50\!\times\!10^{-5}$ & $10^{-3}$ & 180 & 25 \\
Daily & 64 & $4.00\!\times\!10^{-4}$ & $10^{-3}$ & 150 & 30 \\
WS & 64 & $2.03\!\times\!10^{-4}$ & $0.05$ & 240 & 64 \\
CirCor & 6 & $7.50\!\times\!10^{-5}$ & $10^{-3}$ & 24 & 8 \\
ADNI & 32 & $1.05\!\times\!10^{-3}$ & $5\!\times\!10^{-4}$ & 420 & 80 \\
\bottomrule
\end{tabular}
\endgroup
\end{table}
\FloatBarrier
\FloatBarrier
\subsection{Missing-input training controls}
\label{app:missing_input_controls}
This experiment evaluates whether the comparison changes when the baselines receive missing-input views during training. We use SIMS, WS, and ADNI to cover the three domains, with the same splits and test sets as the main evaluation. Complete training examples are augmented with complete and nonempty missing-input views.

\paragraph{Collaboration controls.}
Concat MLP, Attn.\ pool, Soft MoE, FoE, Set Transformer, CLARGA, and I$^2$MoE use the same frozen model pool and modality-level masking. The first six controls use the corresponding \luna configuration's task-view weights and batch construction; FoE, Set Transformer, and CLARGA also use its sampled-chain view schedule. FoE freezes its unimodal task heads before fitting the fuser, while CLARGA retains its contrastive objective. The I$^2$MoE adaptation receives missing-input views through its own training procedure. Each control retains its defining fusion mechanism and uses one validation-MM-selected checkpoint across all test views.

\paragraph{Omni controls.}
Qwen2.5-Omni 3B/7B and Stream-Omni use their frozen backbones. For each of the seven nonempty views, we remove the unavailable raw modalities before feature extraction. One shared standardizer and one shared linear task head are fitted on the training representations from all seven views and used unchanged across test states. Regression uses ridge regression with regularization strength $10$; classification uses balanced logistic regression with inverse regularization strength $1$. View weights follow the task-view objective and sum to one per original training example. In the column order of Tables~\ref{tab:supplement_missing_views_sims}--\ref{tab:supplement_missing_views_adni}, the normalized weights are $(0.042553,\allowbreak 0.191489,\allowbreak 0.191489,\allowbreak 0.191489,\allowbreak 0.127660,\allowbreak 0.127660,\allowbreak 0.127660)$ for SIMS, $(0.058824,\allowbreak 0.078431,\allowbreak 0.078431,\allowbreak 0.078431,\allowbreak 0.235294,\allowbreak 0.235294,\allowbreak 0.235294)$ for WS, and $(0.102564,\allowbreak 0.076923,\allowbreak 0.205128,\allowbreak 0.205128,\allowbreak 0.153846,\allowbreak 0.102564,\allowbreak 0.153846)$ for ADNI. This matches the support and relative weighting of missing-view supervision.

Tables~\ref{tab:supplement_missing_views_sims}--\ref{tab:supplement_missing_views_adni} report every test state for the ten controls. MM averages the six non-full states equally, irrespective of their training weights.
\input{tables/supplement_missing_views_sims}
\input{tables/supplement_missing_views_ws}
\input{tables/supplement_missing_views_adni}
\FloatBarrier

\section{Additional Ablation Results}
\label{sec:refresh_ablation_details}
Table~\ref{tab:refresh_ablation_core} compares four ablation groups. The tables below detail these variants and give every sampled-pair and L9 setting in Section~\ref{sec:refresh_ablation_pairs}.
The final $\Delta_{\mathrm{rel}}$ row in Table~\ref{tab:refresh_ablation_core} averages each variant's percentage change from \luna equally across the eight tasks, reversing the sign for MAE so that negative values always indicate degradation. Bold and underlining identify the best and second-best values in each dataset row.

\paragraph{Prediction-path controls.}
The w/o Common variant predicts using the private head alone and disables the common auxiliary loss and PCC. The w/o Private variant predicts using the unscaled common head alone. The remaining training settings are retained within each task.

\paragraph{Common/private combination.}
We retain both branches and compare the learned residual scale in Eq.~(\ref{eq:prediction}) with fixed coefficients $\alpha\in\{0.25,0.5,1,2\}$ multiplying the common prediction. Equal addition is the case $\alpha=1$. Concat instead feeds the concatenated private features, post-interaction common consensus, and availability indicators to an MLP; its hidden width and dropout follow the private head. All variants retain PCC and common auxiliary supervision. Table~\ref{tab:supplement_residual_combination} reports the complete sweep. Concatenation changes the prediction head as well as its parameter count. The learned residual rule improves over equal addition on six of eight tasks, while the preferred fixed coefficient varies across tasks.

\input{tables/supplement_residual_combination}

\paragraph{Graph interaction controls.}
The graph panel in Table~\ref{tab:refresh_ablation_core} changes only the interaction mode. \luna uses the fixed compatibility graph. Uniform Graph assigns equal weights to the complete active-node graph, including self-loops, followed by incoming normalization; it therefore changes connectivity as well as weights when the compatibility graph is sparse.

\paragraph{Consistency-pair selection and budget.}
These experiments test which state relationships should be regularized and how many are needed per step. With two, three, and four independently maskable inputs, adjacent PCC has $2$, $9$, and $28$ eligible pairs; Full-to-Minimal has $2$, $3$, and $4$; All-Pairs has $3$, $21$, and $105$. These are candidate-set sizes, not per-step comparison counts. All-Pairs includes nonnested pairs and uses the endpoint with more modalities as the stop-gradient target, breaking ties by a fixed view order. For two modalities, Full-to-Minimal and adjacent PCC select the same pair set. The separate PCC pair-budget experiment selects up to $Q$ unique single-modality-removal pairs per step while retaining all task-view forwards. Daily's $Q=4,9$ cells reuse $Q=2$, giving thirty distinct budget runs. CirCor has twenty-eight cover edges, so $Q=9$ covers $32.1\%$ per step.

\paragraph{PCC direction.}
The w/o Stop-gradient variant removes detachment from the richer state's normalized common representation. It does not add a second loss or an explicit midpoint target: the forward cosine-distance value is unchanged for fixed representations, but both endpoints receive its gradient. The task losses continue to supervise both states. This variant retains the graph, adjacent-pair construction, common-state target, and PCC coefficient and schedule. Table~\ref{tab:refresh_ablation_core}c reports its MM on all eight tasks.

\paragraph{Consistency target.}
Table~\ref{tab:refresh_ablation_core}d retains both prediction paths and varies the representation targeted by PCC. PCC-Common corresponds to \luna, applying consistency to the common representation. PCC-Private instead targets the hidden representation immediately before the private head's final linear layer. PCC-Both averages the cosine losses on the separately normalized common and private representations, retaining the same overall PCC coefficient. The comparison examines whether cross-state consistency is most effective when focused on the common representation.

\begin{figure}[!htb]
\centering
\includegraphics[width=\textwidth]{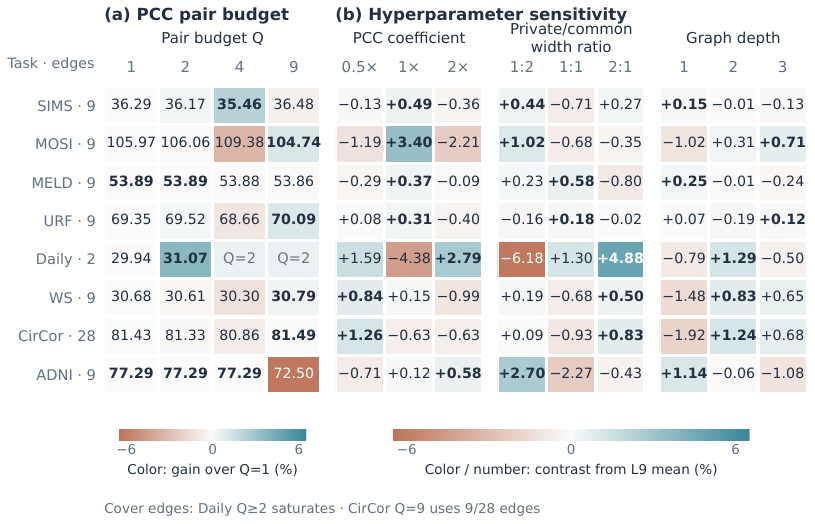}
\caption{\textbf{PCC pair budget and hyperparameter sensitivity.} (a) Numbers: test MM (MAE $\times100$ for SIMS/MOSI; percentages otherwise). Colors: oriented relative change from $Q=1$. Task labels give the numbers of single-modality-removal pairs; gray Daily cells reuse $Q=2$. (b) Numbers/colors: oriented relative contrasts from each task's nine-configuration L9 mean, averaging three distinct configurations per factor level; bold marks each task/factor's best level. Positive means better in both panels.}
\label{fig:refresh_ablation_pairs}
\label{fig:refresh_ablation_l9}
\end{figure}

\input{tables/refresh_ablation_pairs}
\FloatBarrier

\paragraph{Hyperparameter sensitivity: L9 design.}
The L9 design examines sensitivity to consistency strength, private/common width ratio, and interaction depth. The PCC coefficient is $0.5\times$, $1\times$, or $2\times$ its task-specific base value, and graph depth is one, two, or three layers. For fixed private width $d_{\mathrm{priv}}$, the common width $d_{\mathrm{com}}$ is $2d_{\mathrm{priv}}$, $d_{\mathrm{priv}}$, or $\lfloor d_{\mathrm{priv}}/2\rfloor$ at ratios $1{:}2$, $1{:}1$, and $2{:}1$; width and depth changes also change model capacity. For each task, we average native MM across the three configurations at each factor level. Figure~\ref{fig:refresh_ablation_l9}b subtracts the nine-configuration mean from each factor-level mean, divides by the absolute nine-configuration mean, and multiplies by 100. We reverse the sign for MAE so that positive values indicate improvement. Each cell averages three different configurations, not seed replicates. L9 balances factor levels but does not separately estimate their interactions.

\input{tables/refresh_ablation_l9}
\FloatBarrier
\section{Additional Analyses}
\label{app:additional_analyses}
\subsection{Model-Pool and Parameter Analyses}
\label{app:refresh_resources}
\paragraph{Individual pool-entry performance.}
Table~\ref{tab:individual_expert_performance} reports each entry's test performance using the inputs specified in Table~\ref{tab:pool_input_mapping}. Outside CirCor, an entry's prediction is unchanged whenever its required input remains available. Table~\ref{tab:individual_expert_circor_views} gives the CirCor results as auscultation sites are removed, with demographic inputs retained.
\input{tables/individual_expert_performance}
\input{tables/individual_expert_circor_views}
\FloatBarrier

\paragraph{Effect of model-pool size.}
We examine how \luna's predictive performance changes with the number of available models. Each dataset uses nested pools of three, six, and nine frozen models. Table~\ref{tab:model_pool_members} lists the retained models. Figure~\ref{fig:model_pool_scaling_main}(b) summarizes relative MM changes across all eight tasks; Figure~\ref{fig:refresh_resource_pool} gives the task scores at each pool size.

Table~\ref{tab:refresh_resource_pool} reports MM and percentage changes from the three-model pool: the score difference divided by the reference score, with its sign reversed for MAE so that positive means better. Nine models give the best observed MM on all eight tasks. Performance improves at both expansions on six tasks, whereas CirCor and ADNI decline at six models before improving at nine. Daily increases from 35.59\% to 37.85\% and 38.42\%; ADNI changes from 82.71\% to 76.67\% and 84.17\%. Thus, the nine-model pool provides the strongest performance among the evaluated sizes, while the benefit of an intermediate expansion remains task-dependent.
\input{tables/model_pool_members}
\begin{figure}[!htb]
\centering
\includegraphics[width=\textwidth]{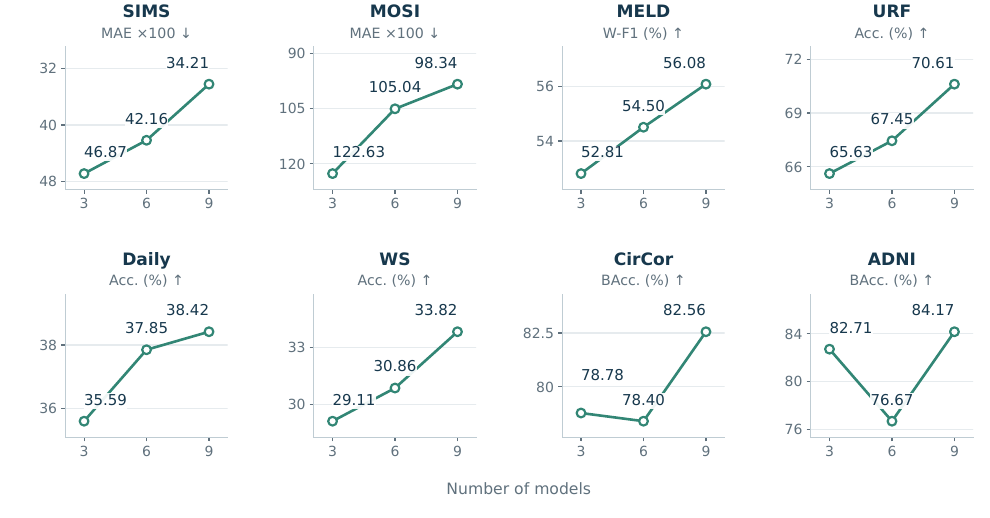}
\caption{\textbf{\luna at different pool sizes.} MM versus the number of models; point labels give scores. Higher positions indicate better performance, with inverted MAE axes. Numerical results: Table~\ref{tab:refresh_resource_pool}.}
\label{fig:refresh_resource_pool}
\end{figure}
\input{tables/refresh_resource_pool}
\FloatBarrier
\paragraph{Sensitivity to model-pool composition.}
The scaling experiment changes pool size; here we keep nine models and vary their composition (Figure~\ref{fig:refresh_resource_replacement}). Starting from the original pool, we replace one designated text model, then an audio model, then a visual model; ADNI uses clinical, MRI, and PET models in that order. Each replacement pool is retrained, its graph is refitted on training data, and its hyperparameters are retuned. The experiment characterizes performance under changes in model composition after re-optimization.

For original metric $m_0$ and metric $m_r$ after $r$ cumulative model replacements, relative change is $100(m_r-m_0)/m_0$ for higher-better tasks and its negative for MAE. Replacing one, two, and three models gives mean direction-corrected relative changes of $-1.19\%$, $-4.65\%$, and $-4.27\%$ from the original pool, respectively. Sensitivity varies across tasks: larger replacements produce the largest changes on SIMS and MOSI, while the effect does not grow uniformly with the number of replaced models. Pool cardinality alone therefore does not determine performance; the representations supplied by the selected models also matter. Together with the scaling experiment, these results identify model composition as a task-dependent design choice at a fixed pool size.

\begin{figure}[!htb]
\centering
\includegraphics[width=\textwidth]{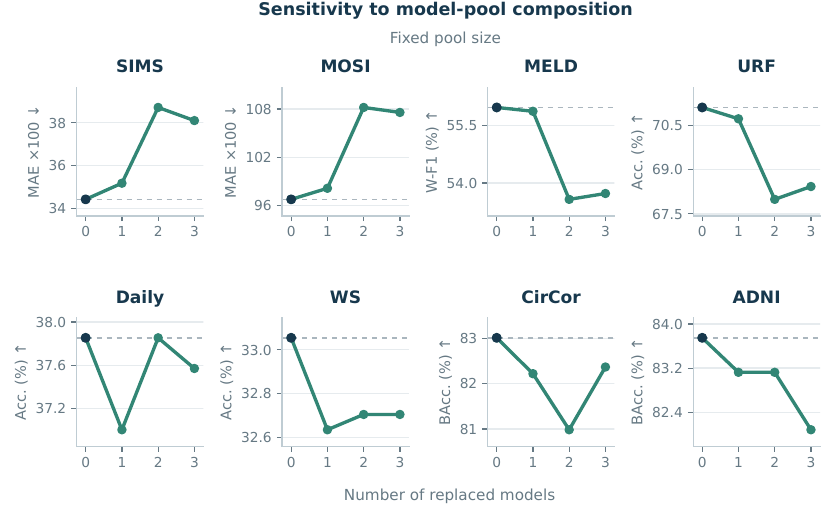}
\caption{\textbf{Sensitivity to model-pool composition.} Zero replacements and dashed lines mark the original nine-model pool. SIMS/MOSI report MAE $\times100$; other scores are percentages.}
\label{fig:refresh_resource_replacement}
\end{figure}
\input{tables/refresh_resource_replacement}
\FloatBarrier

\paragraph{Parameter counts.}
The count includes parameters optimized during collaboration training and excludes frozen backbones, adjacency buffers, and separately fitted frozen task heads. Thus Mean ens.\ has zero optimized collaboration parameters, and LIMSSR's frozen language backbone is excluded.
Figure~\ref{fig:resource_parameters_all} extends the UR-FUNNY view in Figure~\ref{fig:refresh_resource_parameters}(c) to all eight tasks.

\begin{figure}[!htb]
\centering
\includegraphics[width=\textwidth]{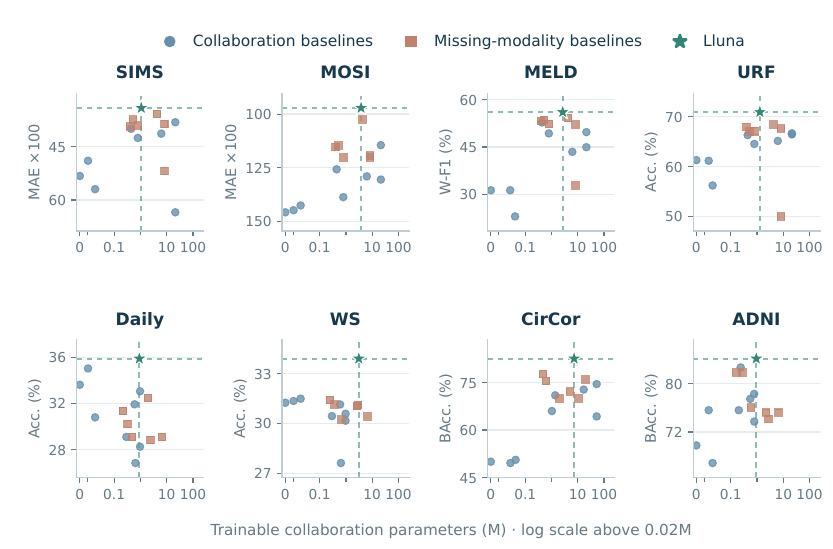}
\caption{\textbf{Trainable parameters versus MM.} Stars: \luna; circles/squares: collaboration/missing-modality baselines; Omni omitted. SIMS/MOSI use MAE $\times100$; other scores are percentages. Inverted MAE axes make higher positions better throughout. Dashed lines mark \luna. The parameter axis is logarithmic above 0.02M and linear near zero. Counts: Table~\ref{tab:refresh_resource_parameters}.}
\label{fig:resource_parameters_all}
\end{figure}
\input{tables/refresh_resource_parameters}
\FloatBarrier
\FloatBarrier
\subsection{Modality-Utility Results}
\label{app:refresh_sampling_modality}
\label{sec:modality_analysis}
We examine modality contributions and pair interactions using exact Shapley values over all modality subsets of a fixed \luna model. Utility is negative absolute error on SIMS/MOSI and the log probability of the true class or valid answer on the other tasks. The empty-input reference predicts zero for regression and uses a uniform class/candidate distribution otherwise.

For the modality set $\mathcal{M}$ with $M$ modalities, let $v(S)$ denote one example's scalar predictive utility under available subset $S$, including the empty-input reference. The Shapley contribution of modality $M_i$ is
\begin{equation}
\operatorname{Shapley}(M_i)=\sum_{S\subseteq\mathcal{M}\setminus\{M_i\}}
\frac{|S|!(M-|S|-1)!}{M!}
\big[v(S\cup\{M_i\})-v(S)\big].
\end{equation}
Pair interactions measure departures from additive predictive utility. For distinct modalities $M_i,M_j$, the interaction averages the utility second difference $v(S\cup\{M_i,M_j\})-v(S\cup\{M_i\})-v(S\cup\{M_j\})+v(S)$ over $S\subseteq\mathcal{M}\setminus\{M_i,M_j\}$ with weight $|S|!(M-|S|-2)!/(M-1)!$, without an additional factor of $1/2$. For three modalities, this equally weights the third modality being absent and present. The reported median and positive/largest-contribution frequencies describe variation across test examples, not uncertainty across training seeds. The attribution uses a fixed model across all modality subsets; it measures that model's modality use rather than the effect of retraining without a modality.

\begin{figure}[!htbp]
\centering
\captionsetup{skip=4pt}
\includegraphics[width=.94\textwidth]{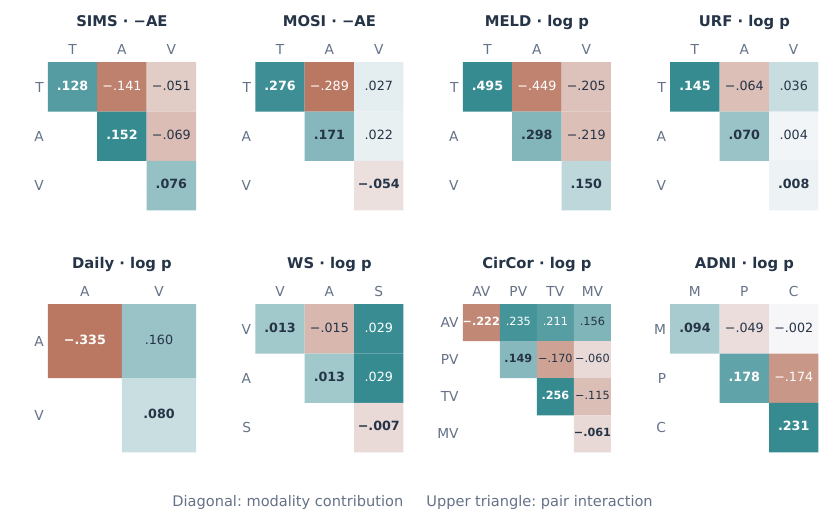}
\caption{\textbf{All-task modality contributions and interactions.} Diagonal cells show mean Shapley values; upper-triangle cells show mean pair interactions. Teal/brown denote positive/negative values, with color ranges scaled within each task. In panel labels, AE is absolute error and $p$ is the probability of the true class or answer. Units are native negative absolute error on SIMS/MOSI and nats elsewhere. Modality abbreviations follow Appendix~\ref{app:multiseed_protocol}.}
\label{fig:analysis_modality}
\end{figure}

The leading modality differs across tasks (Figure~\ref{fig:analysis_modality}): audio on SIMS, text on MOSI/MELD/URF, vision on Daily, the TV auscultation site on CirCor, and clinical inputs on ADNI. On WS, audio and vision contribute similarly; subtitles have a mean contribution of $-0.007$ nats but positive interactions with both audio and vision (about $+0.029$ each). ADNI clinical and PET inputs both contribute positively, while their $-0.174$-nat interaction indicates a joint gain below the additive reference. Individual contributions and pair interactions thus reveal different aspects of evidence use. Section~\ref{sec:graph_cases} examines the prediction compatibility used to construct the model graph; Tables~\ref{tab:refresh_modality_sources} and~\ref{tab:refresh_modality_pairs} give complete modality-utility statistics.

\input{tables/refresh_modality_complete}
\FloatBarrier
\subsection{Chain-Sampling Results}
\label{sec:sampling_details}
To test how chain coverage affects MM, we vary the number of sampled maximal chains while keeping the remaining settings fixed within each task. Table~\ref{tab:sampling_budget} gives every budget in Figure~\ref{fig:refresh_sampling_all}. Repeated pairs within a training step are counted once, so their mean count can be noninteger. The counts measure PCC comparisons rather than wall-clock time.

\begin{figure}[!htb]
\centering
\includegraphics[width=\textwidth]{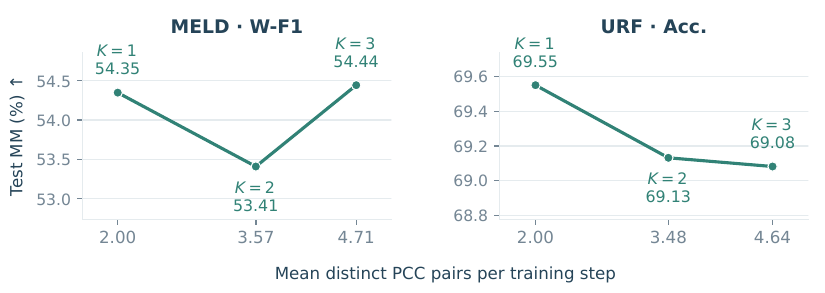}
\caption{\textbf{Chain-sampling budget and predictive performance.} MM versus the mean number of distinct PCC pairs per step. Labels give the number of sampled chains $K$ and the MM score. Only $K$ changes within each task. Numeric results appear in Appendix~\ref{sec:sampling_details}.}
\label{fig:refresh_sampling_all}
\end{figure}

\input{tables/sampling_budget}
As $K$ increases from one to three, MELD MM changes from $54.35\%$ to $53.41\%$ and $54.44\%$, while URF changes from $69.55\%$ to $69.13\%$ and $69.08\%$. More sampled chains therefore do not uniformly improve MM.
\FloatBarrier

\paragraph{Task-specific sparse configurations.}
To assess whether sparse training retains full-lattice predictive performance across tasks, Table~\ref{tab:refresh_sparse_operating} compares task-specific sparse operating points with their full-lattice references. SIMS, MELD, URF, Daily, and CirCor use one independently sampled maximal chain per step; MOSI uses two. WS uses four cyclically scheduled chains, and ADNI uses three. The equal-task mean relative MM change is $-0.28\%$, ranging from $+1.40\%$ on SIMS to $-2.85\%$ on CirCor. At these task-specific configurations, sparse training achieves comparable average missing-modality performance with fewer PCC comparisons per training step. The complementary fixed-setting comparison in Figure~\ref{fig:refresh_sampling_all} varies only $K$ and selects checkpoints by validation Mean-Missing.

\input{tables/refresh_sampling_operating}
\FloatBarrier

\paragraph{PCC comparison counts.}
The last column of Table~\ref{tab:refresh_sparse_operating} gives the mean number of distinct PCC comparisons per training step as full-lattice $\to$ sparse. Pairs shared by sampled chains are counted once; averaging over training steps can yield noninteger counts. Full-lattice training evaluates $M(2^{M-1}-1)$ cover pairs per step for $M$ maskable modalities. For example, sparse training reduces comparisons from nine to two on MELD and from twenty-eight to three on CirCor. These counts describe training-time PCC computation; test evaluation still includes every nonempty availability state.
\section{Motivation Diagnostics}
\label{sec:diagnosis}
\label{app:motivation_details}
\label{sec:diagnostic_audit}
This appendix gives the protocols and complete results underlying Figure~\ref{fig:motivation} and the fixed-anchor analysis in Section~\ref{sec:refresh_diag_cooperation}: collaboration gains under modality removal, recoverability of modality-specific features, and changes in collaborative representations.

\subsection{Fixed-anchor collaboration gains}
\label{app:fixed_anchor_lluna}
\label{app:fixed_anchor_transitions}
Let $\mathcal M$ denote the independently maskable inputs and $S\subseteq\mathcal M$ the available subset. For each modality $M_i\in S$, we select a unimodal anchor $f_i^{\mathrm{uni}}$ by validation performance and keep its predictions fixed across all views retaining that modality. Each collaborative predictor $f_S$ also uses one checkpoint across views. Collaboration gain is
\begin{equation}
\Delta_{\mathrm{collab}}(S;M_i)
= \operatorname{Acc}(f_S)-\operatorname{Acc}(f_i^{\mathrm{uni}}).
\label{eq:collab_gain}
\end{equation}
A negative gain means that collaboration is less accurate than the fixed surviving anchor. \emph{Harmed} is the share of all test examples on which the anchor is correct and collaboration is wrong; \emph{Rescued} is the reverse. Their difference, Rescued minus Harmed, equals $\Delta_{\mathrm{collab}}$ in percentage points.

We average legal anchor/view comparisons within each removal level and dataset. The conventional result additionally averages Logit stack, Concat MLP, Attn. pool, and Soft MoE equally. \luna uses the same evaluation inputs and fixed anchors. SIMS and MOSI use auxiliary binary sentiment accuracy, Acc-2(has0); the other tasks use ordinary classification accuracy for this diagnostic. Required task context is retained, and CirCor's maskable inputs are its four auscultation sites.
For CirCor, we fit a \texttt{StandardScaler} and balanced logistic-regression head ($C=1$) on training examples only for each of the six signal experts at each site; validation balanced accuracy selects one anchor per site, whose head remains fixed across states retaining that site.

\paragraph{Seven-baseline summary.}
Figure~\ref{fig:refresh_diag_cooperation} combines the four conventional methods above with I$^2$MoE, CLARGA, and Set Transformer trained with missing-input views. All use the same fixed anchors, test examples, and legal availability states for this diagnostic. We first average legal anchor/view comparisons within each method and removal level, then give each of the seven methods equal weight. Thus, the two subgroups contribute four-sevenths and three-sevenths, respectively. Full uses all modalities; Minimal retains one independently maskable input. The combined mean gain declines from Full to Minimal on all eight tasks and is negative at Minimal on six; \luna's gain remains positive and exceeds this mean at both endpoints on every task. Table~\ref{tab:refresh_diag_cooperation} gives the corresponding unnormalized gains and test counts.

Figure~\ref{fig:refresh_diag_cooperation} uses task-normalized radii, while outer labels give raw gains in percentage points, in Baselines / \luna order. For each dataset, the shared scale is the largest absolute gain among the baseline mean and \luna at Full and Minimal. Dividing all four gains by this scale preserves their signs and relative magnitudes within that dataset; a common radial offset accommodates negative gains. Both panels use the same normalization, and the zero-gain circle separates beneficial from harmful collaboration.
\begin{table}[!htb]
\centering
\caption{\textbf{Fixed-anchor gains (pp).} Min.\ retains one maskable input; Baselines averages seven methods equally. $N_{\mathrm{test}}$: test examples.}
\label{tab:refresh_diag_cooperation}
\input{tables/fixed_anchor_combined}
\end{table}

\paragraph{Four-method subgroup and \luna transitions.}
Table~\ref{tab:refresh_diag_transitions} reports the complete harmed/rescued rates and net gains for the original four-method conventional subgroup and \luna. At Minimal, \luna retains positive mean gains on all eight datasets, whereas this subgroup's mean is negative on all eight. Minimal retains one maskable input, which may still supply features from several models. CirCor uses site-specific anchors and retains demographic context across availability states. These are dataset-level averages: individual anchor/view comparisons can remain negative.

\input{tables/refresh_diag_transitions}
\FloatBarrier

\paragraph{Additional controls trained with missing-input views.}
We repeat the fixed-anchor evaluation for I$^2$MoE, CLARGA, and Set Transformer after missing-input training, using the same anchors, test examples, legal views, and gain definition. We first average over legal anchor/view comparisons at each removal level within each method, then average the three methods equally. This subgroup contributes to the seven-baseline mean in Figure~\ref{fig:refresh_diag_cooperation}; its results are also summarized separately below. Table~\ref{tab:supplement_anchor_summary} gives Full and Minimal gains; Tables~\ref{tab:supplement_anchor_methods_1}--\ref{tab:supplement_anchor_methods_4} give harmed/rescued rates and net gains for every method and removal level. The eight-task evaluation comprises 324 method/anchor/view comparisons, summarized into 72 method/level rows. The subgroup's mean gain declines from Full to Minimal on all tasks and is negative at Minimal on URF, Daily, WS, and CirCor; this does not imply that every intermediate transition or individual method behaves monotonically.

\input{tables/supplement_anchor_summary}
\input{tables/supplement_anchor_methods}
\FloatBarrier

\paragraph{Sensitivity to baseline composition.}
Excluding Set Transformer on all eight tasks, the equally weighted mean of the remaining six baselines still declines from Full to Minimal on every task and is negative at Minimal on seven; \luna's gains remain positive and exceed this mean at both endpoints on all eight. On CirCor, I$^2$MoE and CLARGA alone average $+2.56$ pp at Full and $+1.28$ pp at Minimal, compared with $+5.97$ and $+1.70$ pp for \luna. Omitting each of the seven baselines in turn also preserves the decline in mean gain and \luna's advantage over the remaining mean at both endpoints on every task.

\paragraph{Exact-removal domain view.}
Figure~\ref{fig:motivation}a compares Full input and removal of exactly one or two modalities on the same seven eligible datasets. Affective computing contains SIMS, MOSI, MELD, and URF; scene understanding contains WS; medical prediction contains CirCor and ADNI. Daily is excluded at all three states because removing both of its modalities would leave no input for a surviving anchor. CirCor's Remove 2 endpoint averages two-site states; the other eligible tasks have three modalities, so Remove 2 coincides with Minimal. Datasets receive equal weight within each domain.

Figure annotations report the absolute decline in domain-mean collaboration gain from Full to Remove~2, in percentage points. Table~\ref{tab:motivation_remove2} reports every plotted endpoint. Conventional gains are negative after two removals in all three domains, whereas \luna retains gains of $+6.96$, $+2.86$, and $+4.17$ pp, respectively.

Absolute declines decrease from $17.51$ to $10.04$ pp in affective computing and from $3.95$ to $0.84$ pp in scene understanding. They are nearly unchanged in medical prediction ($6.51$ versus $6.47$ pp). Averaging the three absolute declines equally gives $9.32$ pp for conventional collaboration and $5.79$ pp for \luna, a $37.9\%$ reduction. This is a domain-averaged comparison, not a percentage decrease in accuracy.

\input{tables/motivation_remove2}
\FloatBarrier

Figure~\ref{fig:collaboration_detail} gives individual conventional-method curves and prediction transitions for MELD, WS, and ADNI.
\begin{figure}[!htb]
  \centering
  \includegraphics[width=\textwidth]{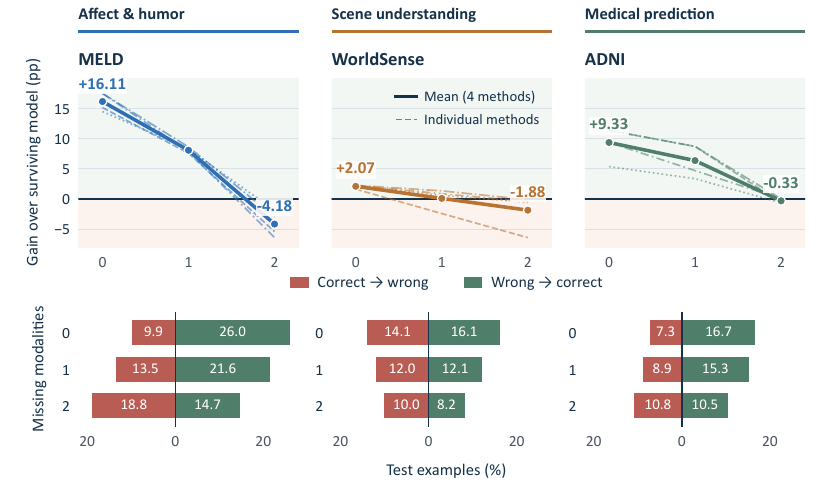}
  \caption{\textbf{Method-level gains and prediction transitions.} Solid curves show four-method mean gains; thin curves show individual methods. Bars report Harmed/Rescued shares of all test examples; Rescued minus Harmed equals ordinary-accuracy gain. Counts 0--2 denote removed modalities. These are individual datasets; anchors and collaboration checkpoints remain fixed across states.}
  \label{fig:collaboration_detail}
\end{figure}
\FloatBarrier

\subsection{Private-signal recoverability}
\label{app:probe_protocol}
We assess how accessible modality-specific features remain after collaboration. For each modality, cross-fitted ridge regression predicts its features from the other modalities; the residual serves as an operational private-signal target. A linear probe then decodes this same target from the modality-alone and full-collaboration representations. The score measures linear recoverability of that residual, rather than task accuracy.

\paragraph{Paired probe in Figure~\ref{fig:motivation}b.}
The paired Attn. pool diagnostic uses 64-dimensional representations and three-fold cross-validation. Scores use the coefficient of determination $R^2$, clipped below at zero and expressed as percentages: zero indicates no improvement over predicting the target mean, and 100 indicates perfect linear recovery. We average over maskable modalities within each dataset, then weight the eight datasets equally. The mean score is $58.52\%$ for a modality alone and $33.47\%$ after full collaboration; recovery decreases on every dataset. Table~\ref{tab:motivation_representation} gives the endpoints and Full/modality-alone retention ratios. These observations motivate direct access to modality features; the prediction-path and consistency-target ablations in Table~\ref{tab:refresh_ablation_core} evaluate the corresponding design choices through task performance.

\paragraph{Across collaboration methods.}
\label{app:probe_multimethod}
Table~\ref{tab:private_probe_multimethod} extends the diagnostic to Logit stack, Concat MLP, Attn. pool, and Soft MoE under Full-only and All-views training. For each modality, Minimal retains that input and any always-available context; the same residual target is decoded from Minimal and Full representations. Scores are variance-weighted, out-of-fold $R^2$, clipped below at zero and expressed as percentages. The residual is defined for independently maskable inputs, not individual pretrained models. For candidate-answer tasks, the correct-option representation is selected as the post-hoc probe readout. This supplementary analysis uses separate folds for Minimal and Full decoding; its endpoints are therefore reported separately from the paired probe in Figure~\ref{fig:motivation}b.

\input{tables/private_probe_multimethod}
All four methods have lower Full recovery under both training policies. The decrease holds across all 192 modality-level comparisons: 24 maskable modalities, four methods, and two policies. Logit stack starts from lower recovery because it operates on modality logits; its smaller absolute decrease does not indicate better preservation of modality features.
\FloatBarrier

\subsection{Representation changes under modality removal}
\label{app:state_geometry_motivation}
\label{sec:refresh_diagnostics}
For the same example under states $S_1$ and $S_2$, let $\mathbf z_1^{\mathrm{att}}$ and $\mathbf z_2^{\mathrm{att}}$ be the corresponding Attn. pool shared representations. We measure their cosine distance:
\[
1-\frac{\langle \mathbf z_1^{\mathrm{att}},\mathbf z_2^{\mathrm{att}}\rangle}
{\|\mathbf z_1^{\mathrm{att}}\|_2\|\mathbf z_2^{\mathrm{att}}\|_2}.
\]
Figure~\ref{fig:motivation}c compares adjacent nested states differing by one modality with Full-to-Minimal states. It equally weights the six three-modality tasks (SIMS, MOSI, MELD, URF, WS, and ADNI), giving mean distances of $0.279$ and $0.456$, respectively. The adjacent category averages all single-removal edges, while Full-to-Minimal removes two modalities from Full. Daily has only one removal to Minimal and CirCor has three, so neither enters this exact-count comparison.

Table~\ref{tab:motivation_representation} reports all eight datasets and additionally includes equal-count states with different modality identities, such as $TA/TV$. Adjacent distances are smaller than Full-to-Minimal distances on every task with a multi-removal comparison; the categories coincide on Daily. These observations characterize representation changes as evidence is removed, with larger changes accompanying multi-step removals. \pcc uses the connections between availability states to organize local guidance (Section~\ref{sec:method}); the task-level effect of this choice is evaluated by the pair-selection ablation in Table~\ref{tab:refresh_ablation_core}c.
\begin{table}[!htb]
\centering
\setlength{\tabcolsep}{4pt}
\caption{\textbf{Complete Attn. pool representation diagnostics.} Probe scores and retention are percentages. Drift is mean cosine distance for one-modality removal, Full to Minimal, and equal-count states with different modality identities. Summary rows weight datasets equally, with mean retention computed from per-dataset Full/modality-alone ratios.}
\label{tab:motivation_representation}
\small
\begin{tabular}{lrrrrrr}
\toprule
 & \multicolumn{3}{c}{Modality-specific probe (\%)} & \multicolumn{3}{c}{Cosine distance} \\
\cmidrule(lr){2-4}\cmidrule(lr){5-7}
Dataset & \shortstack{Modality\\alone} & Full & Retention & \shortstack{One\\removed} & \shortstack{Full to\\Minimal} & \shortstack{Same\\count} \\
\midrule
SIMS & 69.40 & 40.41 & 58.22 & 0.271 & 0.442 & 0.758 \\
MOSI & 68.33 & 39.12 & 57.25 & 0.273 & 0.447 & 0.764 \\
MELD & 68.62 & 37.74 & 55.01 & 0.276 & 0.451 & 0.763 \\
URF & 66.81 & 34.21 & 51.21 & 0.302 & 0.489 & 0.778 \\
Daily & 57.25 & 45.23 & 79.01 & 0.329 & 0.329 & 1.039 \\
WS & 62.10 & 41.36 & 66.61 & 0.297 & 0.483 & 0.800 \\
CirCor & 30.30 & 12.95 & 42.73 & 0.155 & 0.381 & 0.386 \\
ADNI & 45.34 & 16.76 & 36.97 & 0.257 & 0.421 & 0.729 \\
\midrule
Dataset mean & 58.52 & 33.47 & 55.88 & 0.270 & 0.430 & 0.752 \\
Three-modality mean & -- & -- & -- & 0.279 & 0.456 & 0.765 \\
\bottomrule
\end{tabular}
\end{table}

\FloatBarrier

\subsection{Prediction changes across availability states}
\label{app:prediction_shift}
As a complementary output diagnostic, we evaluate fixed collaboration models on the diagnostic cohorts reported in Table~\ref{tab:refresh_diag_jsd}. For the same example under states $S_1$ and $S_2$, we compare the model's class-probability distributions $\mathbf q_1$ and $\mathbf q_2$ using normalized Jensen--Shannon divergence, $\operatorname{JSD}(\mathbf q_1,\mathbf q_2)/\ln 2\in[0,1]$. Zero means identical probability distributions. The table gives diagnostic-cohort sizes; the main evaluation split sizes appear in Table~\ref{tab:dataset_protocol}. This diagnostic applies to the six classification tasks; SIMS and MOSI have scalar regression outputs. Figure~\ref{fig:refresh_diag_attention} presents these output differences alongside the separate Attn. pool representation analysis.

Adjacent pairs differ by one removed modality, and Full-to-Minimal pairs connect Full to each minimal state. Same-level JSD compares states at the largest incomplete level; the representation summary includes equal-cardinality pairs at all nonempty levels. Full-to-Minimal JSD exceeds Adjacent JSD on all five classification tasks with a multi-removal contrast; Daily's pair sets coincide. Table~\ref{tab:refresh_diag_jsd} provides the exact means, and Figure~\ref{fig:refresh_diag_attention}b shows variation across individual state pairs.
\begin{figure}[!htb]
\centering
\includegraphics[width=\textwidth]{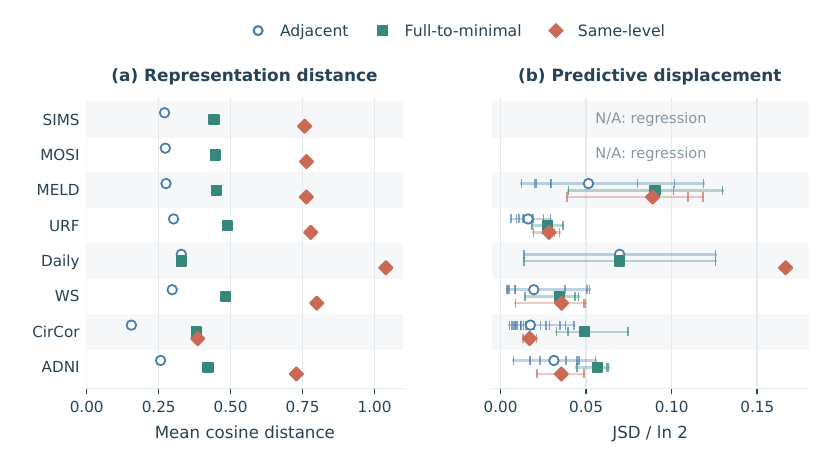}
\caption{\textbf{Complementary representation and prediction diagnostics.} (a) Attn. pool cosine distance on eight tasks. (b) Normalized predictive JSD on the diagnostic cohorts of six classification tasks; regression JSD is N/A. Large markers show category means; ticks and spans show individual state-pair means and their range. Pair definitions appear above; Table~\ref{tab:refresh_diag_jsd} reports exact JSD values and cohort sizes.}
\label{fig:refresh_diag_attention}
\end{figure}

\input{tables/refresh_diag_jsd}
\FloatBarrier
\section{Implementation and Reproducibility}
\label{app:implementation}
Appendix~\ref{app:experimental_setup} specifies the task splits, feature inputs, baseline interfaces, and main-run optimization settings. The released code will include preprocessing scripts, fixed split files, unimodal encoders, feature caches where redistribution is permitted, all modality-availability state definitions, hyperparameter search spaces, random seeds, and commands that reproduce every table. Modality features are projected and standardized before collaboration. Task-level optimization, regularization, early stopping, and availability-state sampling follow the corresponding per-dataset configuration.

\paragraph{Compatibility-graph construction.}
\label{app:graph_configuration}
Task heads are fitted to frozen model features on the training partition, and their predictions on the training examples used for graph construction determine the compatibility weights described in Section~\ref{sec:compatibility}. For classification, we use $\operatorname{sim}(\mathbf q_i,\mathbf q_j)=1-\tfrac{1}{2}\|\mathbf q_i-\mathbf q_j\|_1$, where $\mathbf q_i$ and $\mathbf q_j$ are class-probability vectors from models $i$ and $j$ on the same training example. For regression, we use $\operatorname{sim}(r_i,r_j)=\exp(-|r_i-r_j|)$, where $r_i,r_j$ are the corresponding standardized scalar predictions. Using the absolute rather than squared gap yields a more gradual decay for large standardized prediction differences. The configured weight floor and symmetrization are applied. When graph-temperature scaling is enabled, it transforms an agreement weight $A_{ij}$ to $\exp((\min(A_{ij},1)-1)/\tau)$, where $\tau>0$ is the graph temperature, before union top-$k$ pruning. Here $k$ counts the highest-weight neighbors selected per model; the union retains an edge if either endpoint selects it, and positive self-loops are then retained. The graph temperature acts on agreement weights, not on per-example predictive probabilities. These weights remain fixed throughout collaboration training and inference. For each availability state, the graph is restricted to available models and incoming weights are normalized with a numerical stabilizer.

\subsection{PCC weighting and discrepancy configuration}
\label{app:pcc_configuration}
The base discrepancy underlying Equation~\eqref{eq:pcc} is the cosine distance between normalized common consensuses of the same example at two availability states:
\begin{equation}
d(S_k,S_{k+1})=1-\bar{\mathbf z}_k^\top\bar{\mathbf z}_{k+1}.
\label{eq:edge_drift}
\end{equation}
Configurations may shape a base distance $d$ as $\max(0,d-\delta)^{\rho}$ with a dead zone $\delta\geq0$ and power $\rho>0$; the base objective has $\delta=0$ and $\rho=1$. Edge weights normalize contributions within each child missingness level and may redistribute that level's weight across removed modalities. The chain budget $K$, level weights, $\delta$, $\rho$, and loss coefficients belong to the per-task configuration; the code specifies the chain-sampling implementation. The stability theorem in Appendix~\ref{app:derivations} uses these base distances along a path.

\paragraph{Prediction interfaces.}
The examples used for compatibility estimation come from the training partition; validation and test examples are excluded. Task heads fitted on training-partition frozen features supply the predictions used to construct the graph. These predictions are not additional per-example inputs to the collaboration model. Model features are masked before and after adaptation; zero slots represent unavailable models and are accompanied by an explicit availability mask. Both prediction heads use that mask to adapt their outputs to the available modalities, while the private path retains direct access to adapted modality features. \pcc acts directly on the normalized common consensus; the shared adapters also receive its gradients through the common branch. For each edge, the richer endpoint is detached and the poorer endpoint receives the direct alignment gradient. Parameters are shared across states, so these updates also affect subsequent predictions at richer states; an intermediate state can itself be the poorer endpoint of another sampled edge. Numerical normalization uses an epsilon near zero.

The attention pooling in group readout uses a learned scoring vector $\mathbf w$ and a normalized attention weight $\alpha_i$ for each available model $i$:
\begin{equation}
\mathbf z=\sum_i\alpha_i\mathbf c_i^{(L)},\qquad
\alpha_i=\frac{\exp(\mathbf w^\top\mathbf c_i^{(L)})}
{\sum_j\exp(\mathbf w^\top\mathbf c_j^{(L)})}.
\end{equation}
Both sums run only over available models; the state index is omitted, as in the common-branch equations in the main text.
\section{Additional Derivations}
\label{app:derivations}

\subsection{Paths between arbitrary modality subsets}
For any two distinct nonempty endpoint states $S_0,S_D\subseteq\mathcal M$, add the modalities in $S_D\setminus S_0$ one at a time to reach $S_0\cup S_D$, then remove the modalities in $S_0\setminus S_D$ one at a time to reach $S_D$. Every intermediate state is nonempty: the addition segment contains $S_0$, and the removal segment contains $S_D$. Consecutive states differ by exactly one modality, so they form a path on the undirected version of the cover graph. Its length is
\begin{equation}
D=|S_D\setminus S_0|+|S_0\setminus S_D|.
\end{equation}
Each modality in the symmetric difference must change availability at least once along any path from $S_0$ to $S_D$; this construction therefore has the minimum number of edges. For example, $\mathrm{TA}\rightarrow\mathrm{TAV}\rightarrow\mathrm{TV}$ connects equal-size states with different modality identities. For identical endpoints, the representations coincide and the stability statements hold trivially with an empty path.

The cosine distance in Equation~\eqref{eq:edge_drift} is symmetric, even though \pcc trains on each edge in the richer-to-poorer direction. Thus an addition step in the path uses the same distance as the corresponding removal edge. Stop-gradient directs the alignment update toward the poorer endpoint while preserving this geometric distance.

\subsection{Local-to-global representation stability}
Throughout this subsection, model parameters are fixed and inference is deterministic.

\noindent\textbf{Theorem~\ref{thm:representation_stability}
(Local-to-global representation stability).}
Fix an input and a path $S_0,S_1,\ldots,S_D$ of $D\geq1$ single-modality
transitions. Let $d_k=1-\bar{\mathbf z}_{k-1}^{\top}\bar{\mathbf z}_k$
denote the cosine discrepancy between consecutive normalized common
representations. If the common consensus is nonzero at every state on the path,
then
\begin{equation}
\|\bar{\mathbf z}_0-\bar{\mathbf z}_D\|_2
\leq
\sum_{k=1}^{D}\sqrt{2d_k}
\leq
\sqrt{2D\sum_{k=1}^{D}d_k}.
\end{equation}

\begin{proof}[Proof of Theorem~\ref{thm:representation_stability}]
Fix an input and a path $S_0,S_1,\ldots,S_D$ with $D\geq1$. Assume $\mathbf z_k\neq\mathbf 0$ for all $k\in\{0,\ldots,D\}$, so that $\bar{\mathbf z}_k=\mathbf z_k/\|\mathbf z_k\|_2$ is defined and has unit norm. For each consecutive pair, the definition $d_k=d(S_{k-1},S_k)=1-\bar{\mathbf z}_{k-1}^\top\bar{\mathbf z}_k$ gives $0\leq d_k\leq2$ and
\begin{align}
\|\bar{\mathbf z}_{k-1}-\bar{\mathbf z}_k\|_2^2
&=\|\bar{\mathbf z}_{k-1}\|_2^2+\|\bar{\mathbf z}_k\|_2^2-2\bar{\mathbf z}_{k-1}^\top\bar{\mathbf z}_k \\
&=2d_k.
\end{align}
Telescoping the endpoint difference and applying the triangle inequality therefore yields
\begin{align}
\|\bar{\mathbf z}_0-\bar{\mathbf z}_D\|_2
&=\left\|\sum_{k=1}^D\left(\bar{\mathbf z}_{k-1}-\bar{\mathbf z}_k\right)\right\|_2 \\
&\leq\sum_{k=1}^D\|\bar{\mathbf z}_{k-1}-\bar{\mathbf z}_k\|_2 \\
&=\sum_{k=1}^D\sqrt{2d_k}.
\end{align}
Applying Cauchy--Schwarz to the vectors $(1,\ldots,1)$ and $(\sqrt{2d_1},\ldots,\sqrt{2d_D})$ gives
\begin{equation}
\left(\sum_{k=1}^D\sqrt{2d_k}\right)^2
\leq\left(\sum_{k=1}^D1^2\right)\left(\sum_{k=1}^D2d_k\right)
=2D\sum_{k=1}^Dd_k.
\end{equation}
\end{proof}

The result holds pointwise for every path satisfying the nonzero-consensus assumption, using the base cosine distances $d_k$ along that path. Taking the minimum over admissible paths gives the tightest pathwise bound. Both path length and accumulated local drift contribute to the bound, so a longer path with smaller discrepancies can provide tighter control.

\paragraph{Connection to the training objective.}
The bounds use the base cosine distances along the chosen path. For a fixed sample, a weighted PCC sum with positive weights bounds the sum of base distances on any included simple path after division by the smallest weight on that path. Under random chain sampling, the corresponding expected sum provides this control when every required edge has positive expected weight, using those expected weights in the conversion. The margin and power configuration in Appendix~\ref{app:pcc_configuration} is related to the base distances by
\begin{equation}
\sum_{k=1}^{D}d_k
\leq D\delta+
\sum_{k=1}^{D}\left[\max(0,d_k-\delta)^{\rho}\right]^{1/\rho}.
\end{equation}
Thus weighted or transformed sampled objectives can be related to the same geometric bound through their weights, sampling coverage, and discrepancy settings. Stop-gradient preserves the forward distances while directing the alignment update toward the poorer state. The proof uses exact unit normalization; for an epsilon-stabilized implementation, it applies when each path consensus norm is at least the normalization floor.

\subsection{Supplementary output-drift bound}
Fix the model parameters and one input, and use deterministic inference. At state $S_k$, write the complete prediction from Equation~\eqref{eq:prediction} as $\widehat{\mathbf y}_k=f(\mathbf p(S_k),\mathbf z_k)$. Here $\mathbf p(S_k)$ collects independently processed direct-evidence slots in a fixed order, zero-fills unavailable slots, and includes every availability indicator used by the prediction mapping. Any subsequent mixing on the private path belongs to the prediction mapping $f$. Because each retained evidence slot depends only on its own input, its value is unchanged when other slots become unavailable. Thus the difference $\mathbf p(S_0)-\mathbf p(S_D)$ is confined to changed evidence slots and availability coordinates; the private prediction itself remains free to adapt to these changes.

For direct-input vectors $\mathbf p,\mathbf p'$ and raw common consensuses $\mathbf z,\mathbf z'$ in a product domain containing the compared inputs, assume the complete mapping is Lipschitz in each block:
\begin{align}
\|f(\mathbf p,\mathbf z)-f(\mathbf p',\mathbf z)\|_2
&\leq L_{\mathrm{priv}}\|\mathbf p-\mathbf p'\|_2,\\
\|f(\mathbf p,\mathbf z)-f(\mathbf p,\mathbf z')\|_2
&\leq L_{\mathrm{com}}\|\mathbf z-\mathbf z'\|_2.
\end{align}
The constant $L_{\mathrm{priv}}$ accounts for the private path and both heads' explicit availability conditioning while the common consensus is held fixed. The constant $L_{\mathrm{com}}$ includes the fixed residual scale in Equation~\eqref{eq:prediction}. This formulation accounts separately for changes in the direct inputs and in the raw common consensus used by the prediction head.

\begin{corollary}[Output drift of the complete predictor]
\label{cor:prediction_stability}
Under the assumptions of Theorem~\ref{thm:representation_stability} and the block-Lipschitz conditions above, the complete predictions satisfy
\begin{equation}
\begin{aligned}
\|\widehat{\mathbf y}_0-\widehat{\mathbf y}_D\|_2
\leq{}& L_{\mathrm{priv}}\|\mathbf p(S_0)-\mathbf p(S_D)\|_2 \\
&+L_{\mathrm{com}}\Bigl[\bigl(\|\mathbf z_0\|_2-\|\mathbf z_D\|_2\bigr)^2
+2D\|\mathbf z_0\|_2\|\mathbf z_D\|_2\sum_{k=1}^{D}d_k\Bigr]^{1/2}.
\end{aligned}
\label{eq:prediction_bound}
\end{equation}
\end{corollary}

\begin{proof}[Proof of Corollary~\ref{cor:prediction_stability}]
Insert the intermediate prediction $f(\mathbf p(S_D),\mathbf z_0)$ and apply the block-Lipschitz conditions:
\begin{align}
\|\widehat{\mathbf y}_0-\widehat{\mathbf y}_D\|_2
&\leq\|f(\mathbf p(S_0),\mathbf z_0)-f(\mathbf p(S_D),\mathbf z_0)\|_2\notag\\
&\quad+\|f(\mathbf p(S_D),\mathbf z_0)-f(\mathbf p(S_D),\mathbf z_D)\|_2\notag\\
&\leq L_{\mathrm{priv}}\|\mathbf p(S_0)-\mathbf p(S_D)\|_2
+L_{\mathrm{com}}\|\mathbf z_0-\mathbf z_D\|_2.
\end{align}
For nonzero endpoint consensuses, their raw Euclidean distance decomposes exactly into magnitude and direction changes:
\begin{align}
\|\mathbf z_0-\mathbf z_D\|_2^2
&=(\|\mathbf z_0\|_2-\|\mathbf z_D\|_2)^2\notag\\
&\quad+\|\mathbf z_0\|_2\|\mathbf z_D\|_2
\|\bar{\mathbf z}_0-\bar{\mathbf z}_D\|_2^2\notag\\
&\leq(\|\mathbf z_0\|_2-\|\mathbf z_D\|_2)^2
+2D\|\mathbf z_0\|_2\|\mathbf z_D\|_2\sum_{k=1}^D d_k,
\end{align}
where the last inequality follows from Theorem~\ref{thm:representation_stability}. Taking the square root and substituting into the preceding prediction bound proves Equation~\eqref{eq:prediction_bound}.
\end{proof}

The bound separates changes in direct evidence from changes in the magnitude and direction of the common consensus. The pathwise \pcc distances enter through the directional component.

\section{Limitations}
\label{app:limitations}
\paragraph{Missingness conditions.}
Our evaluation focuses on controlled removal of entire modalities with known availability states. This setting isolates the effect of missing inputs across modality combinations. Further evaluation is needed for partial corruption, temporal misalignment, and naturally occurring missingness associated with sample characteristics.

\paragraph{Model-pool updates.}
\luna is trained for a specified pool of frozen models. Our pool-size and replacement experiments evaluate alternative pools after training the corresponding collaboration modules. Supporting changes to the pool without retraining the full collaboration module is a direction for future work.

\section{Broader Impact}
More reliable behavior under missing modalities is valuable in healthcare, sensing, and human-centered applications, but it does not by itself ensure calibration, fairness, or safety. Performance should be audited separately across missingness mechanisms and demographic groups. Predictions made with fewer modalities should also expose the available-modality state to downstream users rather than presenting the same confidence semantics as a complete observation.

%% file: tables/multiseed_details.tex
\clearpage
\subsection{All methods under every availability state}
\label{app:multiseed_views}
Tables~\ref{tab:multiseed_chsims}--\ref{tab:multiseed_adni} give the complete
method-by-state results for all eight datasets. Bold marks the best mean and
underlining the second-best mean in each column, using unrounded values and
preserving ties. These marks indicate rank rather than statistical significance.
\begin{table}[!htbp]
\centering
\caption{\textbf{SIMS: complete results.} \citep{liu2022chsims} MAE $\times 100$ $\downarrow$. Trained models: five-run mean $\pm$ sample SD; fixed models: single score. $\mathrm{F}$ denotes full input; MM averages non-full states.}
\label{tab:multiseed_chsims}
\fontsize{8.5}{10}\selectfont
\setlength{\tabcolsep}{1.2pt}
\renewcommand{\arraystretch}{1.1}
\arrayrulecolor{TblInk!55}
\begin{tabular}{@{}w{l}{68pt}w{c}{\dimexpr(\linewidth-68pt-44pt-16\tabcolsep)/7\relax}w{c}{\dimexpr(\linewidth-68pt-44pt-16\tabcolsep)/7\relax}w{c}{\dimexpr(\linewidth-68pt-44pt-16\tabcolsep)/7\relax}w{c}{\dimexpr(\linewidth-68pt-44pt-16\tabcolsep)/7\relax}w{c}{\dimexpr(\linewidth-68pt-44pt-16\tabcolsep)/7\relax}w{c}{\dimexpr(\linewidth-68pt-44pt-16\tabcolsep)/7\relax}w{c}{44pt}w{c}{\dimexpr(\linewidth-68pt-44pt-16\tabcolsep)/7\relax}@{}}
\toprule
\rowcolor{TblHeader}
{\color{TblInk}\bfseries Method} & {\color{TblInk}\bfseries TAV$^{\mathrm{F}}$} & {\color{TblInk}\bfseries TA} & {\color{TblInk}\bfseries TV} & {\color{TblInk}\bfseries AV} & {\color{TblInk}\bfseries T} & {\color{TblInk}\bfseries A} & {\color{TblInk}\bfseries V} & {\color{TblInk}\bfseries \textbf{MM}} \\
\midrule
\rowcolor{TblBlue}
\multicolumn{9}{@{}l}{\color{TblInk}\bfseries\scshape General collaboration} \\
\addlinespace[2pt]
Mean ens. & \mbox{53.24\hspace{0.4pt}{\fontsize{6}{7}\selectfont\color{black!75}\textpm 0.02}} & \mbox{53.52\hspace{0.4pt}{\fontsize{6}{7}\selectfont\color{black!75}\textpm 0.02}} & \mbox{53.20\hspace{0.4pt}{\fontsize{6}{7}\selectfont\color{black!75}\textpm 0.02}} & \mbox{53.00\hspace{0.4pt}{\fontsize{6}{7}\selectfont\color{black!75}\textpm 0.02}} & \mbox{53.72\hspace{0.4pt}{\fontsize{6}{7}\selectfont\color{black!75}\textpm 0.02}} & \mbox{53.31\hspace{0.4pt}{\fontsize{6}{7}\selectfont\color{black!75}\textpm 0.03}} & \mbox{52.71\hspace{0.4pt}{\fontsize{6}{7}\selectfont\color{black!75}\textpm 0.03}} & \mbox{53.24\hspace{0.4pt}{\fontsize{6}{7}\selectfont\color{black!75}\textpm 0.02}} \\
\bottomrule
\end{tabular}
\end{table}
\FloatBarrier
\clearpage
\begin{table}[!htbp]
\centering
\caption{\textbf{MOSI: complete results.} \citep{zadeh2016mosi} MAE $\times 100$ $\downarrow$. Trained models: five-run mean $\pm$ sample SD; fixed models: single score. $\mathrm{F}$ denotes full input; MM averages non-full states.}
\label{tab:multiseed_mosi}
\fontsize{8}{10}\selectfont
\setlength{\tabcolsep}{0.6pt}
\renewcommand{\arraystretch}{1.1}
\arrayrulecolor{TblInk!55}
\begin{tabular}{@{}w{l}{\dimexpr\linewidth-329pt-16\tabcolsep\relax}w{c}{40pt}w{c}{40pt}w{c}{40pt}w{c}{43pt}w{c}{40pt}w{c}{43pt}w{c}{43pt}w{c}{40pt}@{}}
\toprule
\rowcolor{TblHeader}
{\color{TblInk}\bfseries Method} & {\color{TblInk}\bfseries TAV$^{\mathrm{F}}$} & {\color{TblInk}\bfseries TA} & {\color{TblInk}\bfseries TV} & {\color{TblInk}\bfseries AV} & {\color{TblInk}\bfseries T} & {\color{TblInk}\bfseries A} & {\color{TblInk}\bfseries V} & {\color{TblInk}\bfseries \textbf{MM}} \\
\midrule
\rowcolor{TblBlue}
\multicolumn{9}{@{}l}{\color{TblInk}\bfseries\scshape General collaboration} \\
\addlinespace[2pt]
Mean ens. & \mbox{145.70\hspace{0.4pt}{\fontsize{6}{7}\selectfont\color{black!75}\textpm 0.25}} & \mbox{145.35\hspace{0.4pt}{\fontsize{6}{7}\selectfont\color{black!75}\textpm 0.23}} & \mbox{146.12\hspace{0.4pt}{\fontsize{6}{7}\selectfont\color{black!75}\textpm 0.25}} & \mbox{145.62\hspace{0.4pt}{\fontsize{6}{7}\selectfont\color{black!75}\textpm 0.26}} & \mbox{145.85\hspace{0.4pt}{\fontsize{6}{7}\selectfont\color{black!75}\textpm 0.23}} & \mbox{144.88\hspace{0.4pt}{\fontsize{6}{7}\selectfont\color{black!75}\textpm 0.24}} & \mbox{146.41\hspace{0.4pt}{\fontsize{6}{7}\selectfont\color{black!75}\textpm 0.27}} & \mbox{145.71\hspace{0.4pt}{\fontsize{6}{7}\selectfont\color{black!75}\textpm 0.25}} \\
\bottomrule
\end{tabular}
\end{table}
\FloatBarrier
\clearpage
\begin{table}[!htbp]
\centering
\caption{\textbf{MELD: complete results.} \citep{poria2019meld} Weighted F1 (\%) $\uparrow$. Trained models: five-run mean $\pm$ sample SD; fixed models: single score. $\mathrm{F}$ denotes full input; MM averages non-full states.}
\label{tab:multiseed_meld}
\fontsize{8.5}{10}\selectfont
\setlength{\tabcolsep}{1.2pt}
\renewcommand{\arraystretch}{1.1}
\arrayrulecolor{TblInk!55}
\begin{tabular}{@{}w{l}{68pt}*{8}{w{c}{\dimexpr(\linewidth-68pt-16\tabcolsep)/8\relax}}@{}}
\toprule
\rowcolor{TblHeader}
{\color{TblInk}\bfseries Method} & {\color{TblInk}\bfseries TAV$^{\mathrm{F}}$} & {\color{TblInk}\bfseries TA} & {\color{TblInk}\bfseries TV} & {\color{TblInk}\bfseries AV} & {\color{TblInk}\bfseries T} & {\color{TblInk}\bfseries A} & {\color{TblInk}\bfseries V} & {\color{TblInk}\bfseries \textbf{MM}} \\
\midrule
\rowcolor{TblBlue}
\multicolumn{9}{@{}l}{\color{TblInk}\bfseries\scshape General collaboration} \\
\addlinespace[2pt]
Mean ens. & \mbox{31.31\hspace{0.4pt}{\fontsize{6}{7}\selectfont\color{black!75}\textpm 0.00}} & \mbox{31.31\hspace{0.4pt}{\fontsize{6}{7}\selectfont\color{black!75}\textpm 0.00}} & \mbox{31.31\hspace{0.4pt}{\fontsize{6}{7}\selectfont\color{black!75}\textpm 0.00}} & \mbox{31.31\hspace{0.4pt}{\fontsize{6}{7}\selectfont\color{black!75}\textpm 0.00}} & \mbox{31.31\hspace{0.4pt}{\fontsize{6}{7}\selectfont\color{black!75}\textpm 0.00}} & \mbox{31.31\hspace{0.4pt}{\fontsize{6}{7}\selectfont\color{black!75}\textpm 0.00}} & \mbox{31.31\hspace{0.4pt}{\fontsize{6}{7}\selectfont\color{black!75}\textpm 0.00}} & \mbox{31.31\hspace{0.4pt}{\fontsize{6}{7}\selectfont\color{black!75}\textpm 0.00}} \\[2pt]
Logit stack & \mbox{52.23\hspace{0.4pt}{\fontsize{6}{7}\selectfont\color{black!75}\textpm 0.10}} & \mbox{43.76\hspace{0.4pt}{\fontsize{6}{7}\selectfont\color{black!75}\textpm 1.63}} & \mbox{12.71\hspace{0.4pt}{\fontsize{6}{7}\selectfont\color{black!75}\textpm 1.67}} & \mbox{38.24\hspace{0.4pt}{\fontsize{6}{7}\selectfont\color{black!75}\textpm 1.45}} & \mbox{7.83\hspace{0.4pt}{\fontsize{6}{7}\selectfont\color{black!75}\textpm 3.22}} & \mbox{31.59\hspace{0.4pt}{\fontsize{6}{7}\selectfont\color{black!75}\textpm 0.28}} & \mbox{2.10\hspace{0.4pt}{\fontsize{6}{7}\selectfont\color{black!75}\textpm 0.00}} & \mbox{22.70\hspace{0.4pt}{\fontsize{6}{7}\selectfont\color{black!75}\textpm 0.57}} \\
\bottomrule
\end{tabular}
\end{table}
\FloatBarrier
\clearpage
\begin{table}[!htbp]
\centering
\caption{\textbf{URF: complete results.} \citep{hasan2019urfunny} Accuracy (\%) $\uparrow$. Trained models: five-run mean $\pm$ sample SD; fixed models: single score. $\mathrm{F}$ denotes full input; MM averages non-full states.}
\label{tab:multiseed_urfunny}
\fontsize{8.5}{10}\selectfont
\setlength{\tabcolsep}{1.2pt}
\renewcommand{\arraystretch}{1.1}
\arrayrulecolor{TblInk!55}
\begin{tabular}{@{}w{l}{68pt}*{8}{w{c}{\dimexpr(\linewidth-68pt-16\tabcolsep)/8\relax}}@{}}
\toprule
\rowcolor{TblHeader}
{\color{TblInk}\bfseries Method} & {\color{TblInk}\bfseries TAV$^{\mathrm{F}}$} & {\color{TblInk}\bfseries TA} & {\color{TblInk}\bfseries TV} & {\color{TblInk}\bfseries AV} & {\color{TblInk}\bfseries T} & {\color{TblInk}\bfseries A} & {\color{TblInk}\bfseries V} & {\color{TblInk}\bfseries \textbf{MM}} \\
\midrule
\rowcolor{TblBlue}
\multicolumn{9}{@{}l}{\color{TblInk}\bfseries\scshape General collaboration} \\
\addlinespace[2pt]
Mean ens. & \mbox{62.58\hspace{0.4pt}{\fontsize{6}{7}\selectfont\color{black!75}\textpm 2.13}} & \mbox{63.82\hspace{0.4pt}{\fontsize{6}{7}\selectfont\color{black!75}\textpm 1.69}} & \mbox{63.32\hspace{0.4pt}{\fontsize{6}{7}\selectfont\color{black!75}\textpm 0.87}} & \mbox{59.34\hspace{0.4pt}{\fontsize{6}{7}\selectfont\color{black!75}\textpm 0.74}} & \mbox{67.04\hspace{0.4pt}{\fontsize{6}{7}\selectfont\color{black!75}\textpm 0.45}} & \mbox{59.40\hspace{0.4pt}{\fontsize{6}{7}\selectfont\color{black!75}\textpm 0.60}} & \mbox{54.63\hspace{0.4pt}{\fontsize{6}{7}\selectfont\color{black!75}\textpm 0.61}} & \mbox{61.26\hspace{0.4pt}{\fontsize{6}{7}\selectfont\color{black!75}\textpm 0.54}} \\[2pt]
Logit stack & \mbox{67.16\hspace{0.4pt}{\fontsize{6}{7}\selectfont\color{black!75}\textpm 0.35}} & \mbox{67.75\hspace{0.4pt}{\fontsize{6}{7}\selectfont\color{black!75}\textpm 0.83}} & \mbox{66.80\hspace{0.4pt}{\fontsize{6}{7}\selectfont\color{black!75}\textpm 0.44}} & \mbox{56.00\hspace{0.4pt}{\fontsize{6}{7}\selectfont\color{black!75}\textpm 4.00}} & \mbox{66.80\hspace{0.4pt}{\fontsize{6}{7}\selectfont\color{black!75}\textpm 0.83}} & \mbox{55.17\hspace{0.4pt}{\fontsize{6}{7}\selectfont\color{black!75}\textpm 3.39}} & \mbox{53.00\hspace{0.4pt}{\fontsize{6}{7}\selectfont\color{black!75}\textpm 1.44}} & \mbox{60.92\hspace{0.4pt}{\fontsize{6}{7}\selectfont\color{black!75}\textpm 1.07}} \\
\bottomrule
\end{tabular}
\end{table}
\FloatBarrier
\clearpage
\begin{table}[!htbp]
\centering
\caption{\textbf{Daily: complete results.} \citep{zhou2025dailyomni} Accuracy (\%) $\uparrow$. Trained models: five-run mean $\pm$ sample SD; fixed models: single score. $\mathrm{F}$ denotes full input; MM averages non-full states.}
\label{tab:multiseed_daily}
\fontsize{8.5}{10}\selectfont
\setlength{\tabcolsep}{1.2pt}
\renewcommand{\arraystretch}{1.1}
\arrayrulecolor{TblInk!55}
\begin{tabular}{@{}w{l}{68pt}*{4}{w{c}{\dimexpr(\linewidth-68pt-8\tabcolsep)/4\relax}}@{}}
\toprule
\rowcolor{TblHeader}
{\color{TblInk}\bfseries Method} & {\color{TblInk}\bfseries AV$^{\mathrm{F}}$} & {\color{TblInk}\bfseries A} & {\color{TblInk}\bfseries V} & {\color{TblInk}\bfseries \textbf{MM}} \\
\midrule
\rowcolor{TblBlue}
\multicolumn{5}{@{}l}{\color{TblInk}\bfseries\scshape General collaboration} \\
\addlinespace[2pt]
Mean ens. & \mbox{34.12\hspace{0.4pt}{\fontsize{6}{7}\selectfont\color{black!75}\textpm 0.31}} & \mbox{34.46\hspace{0.4pt}{\fontsize{6}{7}\selectfont\color{black!75}\textpm 0.56}} & \mbox{31.86\hspace{0.4pt}{\fontsize{6}{7}\selectfont\color{black!75}\textpm 1.36}} & \mbox{33.16\hspace{0.4pt}{\fontsize{6}{7}\selectfont\color{black!75}\textpm 0.81}} \\[2pt]
Logit stack & \mbox{33.90\hspace{0.4pt}{\fontsize{6}{7}\selectfont\color{black!75}\textpm 0.69}} & \mbox{34.12\hspace{0.4pt}{\fontsize{6}{7}\selectfont\color{black!75}\textpm 1.17}} & \mbox{26.10\hspace{0.4pt}{\fontsize{6}{7}\selectfont\color{black!75}\textpm 6.89}} & \mbox{30.11\hspace{0.4pt}{\fontsize{6}{7}\selectfont\color{black!75}\textpm 3.73}} \\[2pt]
Concat MLP & \mbox{33.22\hspace{0.4pt}{\fontsize{6}{7}\selectfont\color{black!75}\textpm 1.09}} & \mbox{33.22\hspace{0.4pt}{\fontsize{6}{7}\selectfont\color{black!75}\textpm 1.01}} & \mbox{29.15\hspace{0.4pt}{\fontsize{6}{7}\selectfont\color{black!75}\textpm 3.16}} & \mbox{31.19\hspace{0.4pt}{\fontsize{6}{7}\selectfont\color{black!75}\textpm 1.85}} \\[2pt]
Attn. pool & \mbox{33.11\hspace{0.4pt}{\fontsize{6}{7}\selectfont\color{black!75}\textpm 1.90}} & \mbox{32.88\hspace{0.4pt}{\fontsize{6}{7}\selectfont\color{black!75}\textpm 2.27}} & \mbox{27.12\hspace{0.4pt}{\fontsize{6}{7}\selectfont\color{black!75}\textpm 2.00}} & \mbox{30.00\hspace{0.4pt}{\fontsize{6}{7}\selectfont\color{black!75}\textpm 1.10}} \\
\bottomrule
\end{tabular}
\end{table}
\FloatBarrier
\clearpage
\begin{table}[!htbp]
\centering
\caption{\textbf{WS: complete results.} \citep{hong2025worldsense} Accuracy (\%) $\uparrow$. Trained models: five-run mean $\pm$ sample SD; fixed models: single score. $\mathrm{F}$ denotes full input; MM averages non-full states.}
\label{tab:multiseed_worldsense}
\fontsize{8.5}{10}\selectfont
\setlength{\tabcolsep}{1.2pt}
\renewcommand{\arraystretch}{1.1}
\arrayrulecolor{TblInk!55}
\begin{tabular}{@{}w{l}{68pt}*{8}{w{c}{\dimexpr(\linewidth-68pt-16\tabcolsep)/8\relax}}@{}}
\toprule
\rowcolor{TblHeader}
{\color{TblInk}\bfseries Method} & {\color{TblInk}\bfseries VAS$^{\mathrm{F}}$} & {\color{TblInk}\bfseries VA} & {\color{TblInk}\bfseries VS} & {\color{TblInk}\bfseries AS} & {\color{TblInk}\bfseries V} & {\color{TblInk}\bfseries A} & {\color{TblInk}\bfseries S} & {\color{TblInk}\bfseries \textbf{MM}} \\
\midrule
\rowcolor{TblBlue}
\multicolumn{9}{@{}l}{\color{TblInk}\bfseries\scshape General collaboration} \\
\addlinespace[2pt]
Mean ens. & \mbox{31.97\hspace{0.4pt}{\fontsize{6}{7}\selectfont\color{black!75}\textpm 0.27}} & \mbox{29.08\hspace{0.4pt}{\fontsize{6}{7}\selectfont\color{black!75}\textpm 0.44}} & \mbox{32.55\hspace{0.4pt}{\fontsize{6}{7}\selectfont\color{black!75}\textpm 0.19}} & \mbox{32.34\hspace{0.4pt}{\fontsize{6}{7}\selectfont\color{black!75}\textpm 0.24}} & \mbox{29.00\hspace{0.4pt}{\fontsize{6}{7}\selectfont\color{black!75}\textpm 0.41}} & \mbox{29.04\hspace{0.4pt}{\fontsize{6}{7}\selectfont\color{black!75}\textpm 0.31}} & \mbox{\textbf{34.60}\hspace{0.4pt}{\fontsize{6}{7}\selectfont\color{black!75}\textpm 0.32}} & \mbox{31.10\hspace{0.4pt}{\fontsize{6}{7}\selectfont\color{black!75}\textpm 0.12}} \\[2pt]
Logit stack & \mbox{33.85\hspace{0.4pt}{\fontsize{6}{7}\selectfont\color{black!75}\textpm 0.62}} & \mbox{28.95\hspace{0.4pt}{\fontsize{6}{7}\selectfont\color{black!75}\textpm 0.48}} & \mbox{33.35\hspace{0.4pt}{\fontsize{6}{7}\selectfont\color{black!75}\textpm 0.64}} & \mbox{\textbf{33.81}\hspace{0.4pt}{\fontsize{6}{7}\selectfont\color{black!75}\textpm 0.50}} & \mbox{29.16\hspace{0.4pt}{\fontsize{6}{7}\selectfont\color{black!75}\textpm 0.94}} & \mbox{28.74\hspace{0.4pt}{\fontsize{6}{7}\selectfont\color{black!75}\textpm 0.60}} & \mbox{\underline{33.18}\hspace{0.4pt}{\fontsize{6}{7}\selectfont\color{black!75}\textpm 0.70}} & \mbox{\underline{31.20}\hspace{0.4pt}{\fontsize{6}{7}\selectfont\color{black!75}\textpm 0.15}} \\
\bottomrule
\end{tabular}
\end{table}
\FloatBarrier
\clearpage
\begin{table}[!htbp]
\centering
\caption{\textbf{CirCor: complete results.} \citep{oliveira2022circor} Balanced accuracy (\%) $\uparrow$. Panel 1/2. Trained models: five-run mean $\pm$ sample SD; fixed models: single score. $\mathrm{F}$ denotes full input; MM averages non-full states. AV/PV/TV/MV denote auscultation sites.}
\label{tab:multiseed_circor}
\fontsize{8.5}{10}\selectfont
\setlength{\tabcolsep}{1.2pt}
\renewcommand{\arraystretch}{1.1}
\arrayrulecolor{TblInk!55}
\begin{tabular}{@{}w{l}{68pt}*{8}{w{c}{\dimexpr(\linewidth-68pt-16\tabcolsep)/8\relax}}@{}}
\toprule
\rowcolor{TblHeader}
{\color{TblInk}\bfseries Method} & {\color{TblInk}\bfseries \shortstack[c]{AV+PV+\\TV+MV$^{\mathrm{F}}$}} & {\color{TblInk}\bfseries \shortstack[c]{AV+PV+\\TV}} & {\color{TblInk}\bfseries \shortstack[c]{AV+PV+\\MV}} & {\color{TblInk}\bfseries \shortstack[c]{AV+TV+\\MV}} & {\color{TblInk}\bfseries \shortstack[c]{PV+TV+\\MV}} & {\color{TblInk}\bfseries AV+PV} & {\color{TblInk}\bfseries AV+TV} & {\color{TblInk}\bfseries AV+MV} \\
\midrule
\rowcolor{TblBlue}
\multicolumn{9}{@{}l}{\color{TblInk}\bfseries\scshape General collaboration} \\
\addlinespace[2pt]
Mean ens. & \mbox{50.00\hspace{0.4pt}{\fontsize{6}{7}\selectfont\color{black!75}\textpm 0.00}} & \mbox{50.00\hspace{0.4pt}{\fontsize{6}{7}\selectfont\color{black!75}\textpm 0.00}} & \mbox{50.00\hspace{0.4pt}{\fontsize{6}{7}\selectfont\color{black!75}\textpm 0.00}} & \mbox{50.00\hspace{0.4pt}{\fontsize{6}{7}\selectfont\color{black!75}\textpm 0.00}} & \mbox{50.00\hspace{0.4pt}{\fontsize{6}{7}\selectfont\color{black!75}\textpm 0.00}} & \mbox{50.00\hspace{0.4pt}{\fontsize{6}{7}\selectfont\color{black!75}\textpm 0.00}} & \mbox{50.00\hspace{0.4pt}{\fontsize{6}{7}\selectfont\color{black!75}\textpm 0.00}} & \mbox{50.00\hspace{0.4pt}{\fontsize{6}{7}\selectfont\color{black!75}\textpm 0.00}} \\
\bottomrule
\end{tabular}
\end{table}
\FloatBarrier
\clearpage
\begin{table}[!htbp]
\centering
\begin{minipage}{\linewidth}
Table~\ref{tab:multiseed_circor} (continued): \textbf{CirCor: complete results.} \citep{oliveira2022circor} Balanced accuracy (\%) $\uparrow$. Panel 2/2. Trained models: five-run mean $\pm$ sample SD; fixed models: single score. AV/PV/TV/MV denote auscultation sites.
\end{minipage}\par
\vspace{7pt}
\fontsize{8.5}{10}\selectfont
\setlength{\tabcolsep}{1.2pt}
\renewcommand{\arraystretch}{1.1}
\arrayrulecolor{TblInk!55}
\begin{tabular}{@{}w{l}{68pt}*{8}{w{c}{\dimexpr(\linewidth-68pt-16\tabcolsep)/8\relax}}@{}}
\toprule
\rowcolor{TblHeader}
{\color{TblInk}\bfseries Method} & {\color{TblInk}\bfseries PV+TV} & {\color{TblInk}\bfseries PV+MV} & {\color{TblInk}\bfseries TV+MV} & {\color{TblInk}\bfseries AV} & {\color{TblInk}\bfseries PV} & {\color{TblInk}\bfseries TV} & {\color{TblInk}\bfseries MV} & {\color{TblInk}\bfseries \textbf{MM}} \\
\midrule
\rowcolor{TblBlue}
\multicolumn{9}{@{}l}{\color{TblInk}\bfseries\scshape General collaboration} \\
\addlinespace[2pt]
Mean ens. & \mbox{50.00\hspace{0.4pt}{\fontsize{6}{7}\selectfont\color{black!75}\textpm 0.00}} & \mbox{50.00\hspace{0.4pt}{\fontsize{6}{7}\selectfont\color{black!75}\textpm 0.00}} & \mbox{50.00\hspace{0.4pt}{\fontsize{6}{7}\selectfont\color{black!75}\textpm 0.00}} & \mbox{50.00\hspace{0.4pt}{\fontsize{6}{7}\selectfont\color{black!75}\textpm 0.00}} & \mbox{50.00\hspace{0.4pt}{\fontsize{6}{7}\selectfont\color{black!75}\textpm 0.00}} & \mbox{50.00\hspace{0.4pt}{\fontsize{6}{7}\selectfont\color{black!75}\textpm 0.00}} & \mbox{50.00\hspace{0.4pt}{\fontsize{6}{7}\selectfont\color{black!75}\textpm 0.00}} & \mbox{50.00\hspace{0.4pt}{\fontsize{6}{7}\selectfont\color{black!75}\textpm 0.00}} \\[2pt]
Logit stack & \mbox{51.11\hspace{0.4pt}{\fontsize{6}{7}\selectfont\color{black!75}\textpm 1.52}} & \mbox{50.56\hspace{0.4pt}{\fontsize{6}{7}\selectfont\color{black!75}\textpm 1.24}} & \mbox{51.67\hspace{0.4pt}{\fontsize{6}{7}\selectfont\color{black!75}\textpm 3.73}} & \mbox{50.00\hspace{0.4pt}{\fontsize{6}{7}\selectfont\color{black!75}\textpm 0.00}} & \mbox{50.00\hspace{0.4pt}{\fontsize{6}{7}\selectfont\color{black!75}\textpm 0.00}} & \mbox{56.16\hspace{0.4pt}{\fontsize{6}{7}\selectfont\color{black!75}\textpm 9.72}} & \mbox{50.00\hspace{0.4pt}{\fontsize{6}{7}\selectfont\color{black!75}\textpm 0.00}} & \mbox{51.95\hspace{0.4pt}{\fontsize{6}{7}\selectfont\color{black!75}\textpm 2.20}} \\
\bottomrule
\end{tabular}
\end{table}
\FloatBarrier
\clearpage
\begin{table}[!htbp]
\centering
\caption{\textbf{ADNI: complete results.} \citep{mueller2005adni} Balanced accuracy (\%) $\uparrow$. Trained models: five-run mean $\pm$ sample SD; fixed models: single score. $\mathrm{F}$ denotes full input; MM averages non-full states.}
\label{tab:multiseed_adni}
\fontsize{8.5}{10}\selectfont
\setlength{\tabcolsep}{1.2pt}
\renewcommand{\arraystretch}{1.1}
\arrayrulecolor{TblInk!55}
\begin{tabular}{@{}w{l}{68pt}*{8}{w{c}{\dimexpr(\linewidth-68pt-16\tabcolsep)/8\relax}}@{}}
\toprule
\rowcolor{TblHeader}
{\color{TblInk}\bfseries Method} & {\color{TblInk}\bfseries MPC$^{\mathrm{F}}$} & {\color{TblInk}\bfseries MP} & {\color{TblInk}\bfseries MC} & {\color{TblInk}\bfseries PC} & {\color{TblInk}\bfseries M} & {\color{TblInk}\bfseries P} & {\color{TblInk}\bfseries C} & {\color{TblInk}\bfseries \textbf{MM}} \\
\midrule
\rowcolor{TblBlue}
\multicolumn{9}{@{}l}{\color{TblInk}\bfseries\scshape General collaboration} \\
\addlinespace[2pt]
Mean ens. & \mbox{73.25\hspace{0.4pt}{\fontsize{6}{7}\selectfont\color{black!75}\textpm 8.18}} & \mbox{64.00\hspace{0.4pt}{\fontsize{6}{7}\selectfont\color{black!75}\textpm 11.71}} & \mbox{64.50\hspace{0.4pt}{\fontsize{6}{7}\selectfont\color{black!75}\textpm 9.42}} & \mbox{76.75\hspace{0.4pt}{\fontsize{6}{7}\selectfont\color{black!75}\textpm 6.99}} & \mbox{48.75\hspace{0.4pt}{\fontsize{6}{7}\selectfont\color{black!75}\textpm 12.02}} & \mbox{70.25\hspace{0.4pt}{\fontsize{6}{7}\selectfont\color{black!75}\textpm 6.40}} & \mbox{77.00\hspace{0.4pt}{\fontsize{6}{7}\selectfont\color{black!75}\textpm 4.56}} & \mbox{66.88\hspace{0.4pt}{\fontsize{6}{7}\selectfont\color{black!75}\textpm 6.12}} \\[2pt]
Logit stack & \mbox{78.00\hspace{0.4pt}{\fontsize{6}{7}\selectfont\color{black!75}\textpm 3.60}} & \mbox{74.25\hspace{0.4pt}{\fontsize{6}{7}\selectfont\color{black!75}\textpm 6.71}} & \mbox{76.25\hspace{0.4pt}{\fontsize{6}{7}\selectfont\color{black!75}\textpm 4.59}} & \mbox{79.75\hspace{0.4pt}{\fontsize{6}{7}\selectfont\color{black!75}\textpm 3.79}} & \mbox{53.75\hspace{0.4pt}{\fontsize{6}{7}\selectfont\color{black!75}\textpm 7.40}} & \mbox{74.25\hspace{0.4pt}{\fontsize{6}{7}\selectfont\color{black!75}\textpm 9.54}} & \mbox{73.25\hspace{0.4pt}{\fontsize{6}{7}\selectfont\color{black!75}\textpm 2.09}} & \mbox{71.92\hspace{0.4pt}{\fontsize{6}{7}\selectfont\color{black!75}\textpm 3.20}} \\
\bottomrule
\end{tabular}
\end{table}
\FloatBarrier

%% file: tables/supplement_missing_views_sims.tex
\begin{table}[!htbp]
\centering
\captionsetup{skip=4pt}
\caption{\textbf{Missing-input training on CH-SIMS v2.} Results for every test view; column labels specify the retained modalities. MAE is multiplied by 100. MM averages the six non-full views. T/A/V denote text, audio, and vision.}
\label{tab:supplement_missing_views_sims}
\begingroup\fontsize{8.5}{10.2}\selectfont
\setlength{\tabcolsep}{3.5pt}
\renewcommand{\arraystretch}{1.13}
\begin{tabular}{@{}l*{8}{r}@{}}
\toprule
\rowcolor{TblHeader}Method & TAV & TA & TV & AV & T & A & V & MM \\
\midrule
Concat MLP & 32.68 & 34.46 & 35.68 & 36.12 & 39.78 & 37.63 & 43.74 & 37.90 \\
Attn. pool & 29.24 & 31.47 & 33.30 & 31.72 & 37.13 & 35.03 & 44.10 & 35.46 \\
Soft MoE & 35.53 & 33.87 & 38.16 & 37.70 & 38.72 & 36.33 & 43.72 & 38.08 \\
FoE & 44.93 & 44.84 & 47.62 & 45.90 & 49.06 & 47.65 & 50.48 & 47.59 \\
Set Trans. & 30.89 & 33.00 & 36.45 & 33.16 & 40.17 & 37.04 & 44.66 & 37.42 \\
CLARGA & 34.20 & 36.91 & 38.61 & 36.90 & 41.74 & 40.88 & 46.06 & 40.18 \\
I$^2$MoE & 30.71 & 33.87 & 36.06 & 32.31 & 40.52 & 37.56 & 43.51 & 37.31 \\
\midrule
Qwen2.5-Omni 3B & 48.69 & 53.76 & 50.51 & 46.14 & 57.34 & 53.04 & 48.51 & 51.55 \\
Qwen2.5-Omni 7B & 51.18 & 48.11 & 49.39 & 45.97 & 58.04 & 44.12 & 46.03 & 48.61 \\
Stream-Omni & 59.27 & 61.41 & 61.34 & 65.72 & 63.95 & 80.02 & 61.83 & 65.71 \\
\bottomrule
\end{tabular}
\endgroup
\end{table}

%% file: tables/supplement_missing_views_ws.tex
\begin{table}[!htbp]
\centering
\captionsetup{skip=4pt}
\caption{\textbf{Missing-input training on WorldSense.} Results for every test view; column labels specify the retained modalities. Accuracy is reported as a percentage. MM averages the six non-full views. V/A/S denote vision, audio, and subtitles.}
\label{tab:supplement_missing_views_ws}
\begingroup\fontsize{8.5}{10.2}\selectfont
\setlength{\tabcolsep}{3.5pt}
\renewcommand{\arraystretch}{1.13}
\begin{tabular}{@{}l*{8}{r}@{}}
\toprule
\rowcolor{TblHeader}Method & VAS & VA & VS & AS & V & A & S & MM \\
\midrule
Concat MLP & 32.64 & 34.31 & 31.59 & 31.59 & 35.15 & 31.59 & 31.59 & 32.64 \\
Attn. pool & 34.10 & 31.80 & 30.13 & 30.13 & 32.01 & 32.22 & 36.61 & 32.15 \\
Soft MoE & 29.92 & 33.47 & 29.92 & 29.50 & 33.26 & 31.38 & 29.50 & 31.17 \\
FoE & 32.64 & 28.87 & 33.68 & 32.85 & 28.87 & 28.87 & 33.05 & 31.03 \\
Set Trans. & 32.01 & 29.71 & 33.05 & 32.01 & 30.96 & 29.29 & 32.01 & 31.17 \\
CLARGA & 33.68 & 28.87 & 34.73 & 32.85 & 29.29 & 28.03 & 33.26 & 31.17 \\
I$^2$MoE & 32.22 & 28.66 & 32.43 & 32.22 & 26.36 & 29.92 & 32.01 & 30.27 \\
\midrule
Qwen2.5-Omni 3B & 29.29 & 38.08 & 28.24 & 28.45 & 32.01 & 31.59 & 24.69 & 30.51 \\
Qwen2.5-Omni 7B & 34.94 & 32.85 & 31.80 & 32.22 & 29.71 & 30.13 & 26.57 & 30.54 \\
Stream-Omni & 30.96 & 30.54 & 29.71 & 28.66 & 31.17 & 29.08 & 30.34 & 29.92 \\
\bottomrule
\end{tabular}
\endgroup
\end{table}

%% file: tables/supplement_missing_views_adni.tex
\begin{table}[!htbp]
\centering
\captionsetup{skip=4pt}
\caption{\textbf{Missing-input training on ADNI.} Results for every test view; column labels specify the retained modalities. Balanced accuracy is reported as a percentage. MM averages the six non-full views. M/P/C denote MRI, PET, and clinical information.}
\label{tab:supplement_missing_views_adni}
\begingroup\fontsize{8.5}{10.2}\selectfont
\setlength{\tabcolsep}{3.5pt}
\renewcommand{\arraystretch}{1.13}
\begin{tabular}{@{}l*{8}{r}@{}}
\toprule
\rowcolor{TblHeader}Method & MPC & MP & MC & PC & M & P & C & MM \\
\midrule
Concat MLP & 87.50 & 76.25 & 67.50 & 81.25 & 46.25 & 76.25 & 73.75 & 70.21 \\
Attn. pool & 87.50 & 91.25 & 68.75 & 76.25 & 80.00 & 82.50 & 70.00 & 78.13 \\
Soft MoE & 75.00 & 83.75 & 80.00 & 80.00 & 81.25 & 80.00 & 80.00 & 80.83 \\
FoE & 76.25 & 77.50 & 62.50 & 57.50 & 70.00 & 52.50 & 60.00 & 63.33 \\
Set Trans. & 95.00 & 87.50 & 82.50 & 88.75 & 73.75 & 81.25 & 78.75 & 82.08 \\
CLARGA & 88.75 & 78.75 & 80.00 & 86.25 & 77.50 & 82.50 & 78.75 & 80.63 \\
I$^2$MoE & 73.75 & 86.25 & 75.00 & 83.75 & 58.75 & 86.25 & 76.25 & 77.71 \\
\midrule
Qwen2.5-Omni 3B & 67.50 & 60.00 & 71.25 & 85.00 & 38.75 & 61.25 & 85.00 & 66.88 \\
Qwen2.5-Omni 7B & 81.25 & 56.25 & 85.00 & 83.75 & 50.00 & 56.25 & 87.50 & 69.79 \\
Stream-Omni & 81.25 & 75.00 & 62.50 & 76.25 & 67.50 & 86.25 & 77.50 & 74.17 \\
\bottomrule
\end{tabular}
\endgroup
\end{table}

%% file: tables/supplement_residual_combination.tex
\begin{table}[!htbp]
\centering
\captionsetup{skip=4pt}
\caption{\textbf{Common/private combination variants.} MM for feature concatenation (Concat) and fixed residual scales. Equal addition corresponds to $\alpha=1$. \luna uses the learned residual scale. Metrics follow Table~\ref{tab:multiseed_overall}; bold identifies the best displayed value in each row.}
\label{tab:supplement_residual_combination}
\begingroup\fontsize{8.5}{10.2}\selectfont
\setlength{\tabcolsep}{7pt}
\renewcommand{\arraystretch}{1.13}
\begin{tabular}{@{}l>{\columncolor{TblOurs}}r*{5}{r}@{}}
\toprule
\rowcolor{TblHeader}Dataset & \luna & Concat & $\alpha=0.25$ & $\alpha=0.5$ & $\alpha=1$ & $\alpha=2$ \\
\midrule
SIMS$\downarrow$ & 34.21 & 34.79 & 34.52 & 34.37 & \textbf{33.88} & 34.69 \\
MOSI$\downarrow$ & \textbf{98.34} & 103.19 & 99.55 & 105.66 & 100.21 & 99.36 \\
MELD$\uparrow$ & 56.08 & 55.59 & 55.65 & 55.91 & \textbf{56.16} & 56.05 \\
URF$\uparrow$ & \textbf{70.61} & 69.06 & 70.47 & 70.09 & 70.12 & 69.90 \\
\addlinespace[3pt]
Daily$\uparrow$ & \textbf{38.42} & 31.92 & 29.10 & 29.66 & 28.81 & 32.20 \\
WS$\uparrow$ & \textbf{33.82} & 31.10 & 33.33 & 33.54 & 33.30 & 33.17 \\
\addlinespace[3pt]
CirCor$\uparrow$ & \textbf{82.56} & 81.58 & 78.91 & 78.27 & 74.89 & 77.69 \\
ADNI$\uparrow$ & 84.17 & 83.00 & 83.79 & \textbf{84.21} & 83.96 & 80.63 \\
\bottomrule
\end{tabular}
\endgroup
\end{table}

%% file: tables/refresh_ablation_pairs.tex
\begin{table}[!htb]\centering
\caption{\textbf{PCC pair budgets.} MM; metrics and units follow Table~\ref{tab:multiseed_overall}. $Q$ counts cover pairs selected per step. Daily's $Q=4,9$ cells reuse $Q=2$; CirCor samples at most 9 of 28 edges. Bold marks each task's best budget.}
\label{tab:refresh_ablation_pairs}
\begingroup\fontsize{9}{10.5}\selectfont\setlength{\tabcolsep}{3.3pt}
\renewcommand{\arraystretch}{1.13}
\begin{tabular}{@{}lrrrrrrrr@{}}\toprule
Setting & SIMS$\downarrow$ & MOSI$\downarrow$ & MELD$\uparrow$ & URF$\uparrow$ & Daily$\uparrow$ & WS$\uparrow$ & CirCor$\uparrow$ & ADNI$\uparrow$ \\
\midrule
$Q=1$ & 36.29 & 105.97 & \textbf{53.89} & 69.35 & 29.94 & 30.68 & 81.43 & \textbf{77.29} \\
$Q=2$ & 36.17 & 106.06 & \textbf{53.89} & 69.52 & \textbf{31.07} & 30.61 & 81.33 & \textbf{77.29} \\
$Q=4$ & \textbf{35.46} & 109.38 & 53.88 & 68.66 & \textbf{31.07} & 30.30 & 80.86 & \textbf{77.29} \\
$Q=9$ & 36.48 & \textbf{104.74} & 53.86 & \textbf{70.09} & \textbf{31.07} & \textbf{30.79} & \textbf{81.49} & 72.50 \\
\bottomrule\end{tabular}\endgroup\end{table}

%% file: tables/refresh_ablation_l9.tex
\begin{table}[!htb]\centering
\caption{\textbf{Nine L9 configurations.} MM; metrics and units follow Table~\ref{tab:multiseed_overall}. Weights multiply each task's base PCC coefficient. Bold marks the best configuration per task.}
\label{tab:refresh_ablation_l9}
\begingroup\fontsize{9}{10.5}\selectfont\setlength{\tabcolsep}{3.3pt}
\renewcommand{\arraystretch}{1.13}
\begin{tabular}{@{}lrrrrrrrr@{}}\toprule
Setting & SIMS$\downarrow$ & MOSI$\downarrow$ & MELD$\uparrow$ & URF$\uparrow$ & Daily$\uparrow$ & WS$\uparrow$ & CirCor$\uparrow$ & ADNI$\uparrow$ \\
\midrule
L9-1 & 35.70 & 109.18 & 54.03 & 69.42 & 28.81 & 29.92 & 81.18 & 76.46 \\
L9-2 & 36.20 & 111.13 & 54.51 & 69.47 & \textbf{33.90} & \textbf{31.21} & 81.47 & 74.17 \\
L9-3 & 36.13 & 112.45 & 53.53 & 69.50 & 33.33 & \textbf{31.21} & 82.49 & 73.75 \\
L9-4 & 35.80 & 106.35 & 54.57 & 69.35 & 28.53 & 30.96 & 80.81 & 77.50 \\
L9-5 & 35.99 & \textbf{104.11} & 54.31 & \textbf{69.85} & 29.38 & 30.09 & 80.93 & 71.67 \\
L9-6 & \textbf{35.57} & 107.21 & 54.26 & 69.67 & 32.49 & 30.65 & 78.82 & 77.08 \\
L9-7 & 35.91 & 109.95 & 54.30 & 69.11 & 31.36 & 30.86 & 80.32 & \textbf{78.12} \\
L9-8 & 36.46 & 115.83 & \textbf{54.65} & 69.28 & 32.49 & 29.64 & 77.45 & 75.00 \\
L9-9 & 35.90 & 110.34 & 53.44 & 69.00 & 33.33 & 30.16 & \textbf{82.80} & 74.17 \\
\midrule
\multicolumn{9}{l}{\textit{L9 row: (weight multiplier, private/common width ratio, graph depth)}} \\
\multicolumn{3}{l}{L9-1: (0.5$\times$, 1:2, 1)} & \multicolumn{3}{l}{L9-2: (0.5$\times$, 1:1, 2)} & \multicolumn{3}{l}{L9-3: (0.5$\times$, 2:1, 3)} \\
\multicolumn{3}{l}{L9-4: (1$\times$, 1:2, 2)} & \multicolumn{3}{l}{L9-5: (1$\times$, 1:1, 3)} & \multicolumn{3}{l}{L9-6: (1$\times$, 2:1, 1)} \\
\multicolumn{3}{l}{L9-7: (2$\times$, 1:2, 3)} & \multicolumn{3}{l}{L9-8: (2$\times$, 1:1, 1)} & \multicolumn{3}{l}{L9-9: (2$\times$, 2:1, 2)} \\
\bottomrule\end{tabular}\endgroup\end{table}

%% file: tables/individual_expert_performance.tex
\begin{table}[!htb]
\centering
\caption{\textbf{Individual pool-entry performance.} Full-input test scores; each entry uses its own features and task head. Metrics and units follow Table~\ref{tab:multiseed_overall}.}
\label{tab:individual_expert_performance}
\begingroup
\small
\setlength{\tabcolsep}{3pt}
\renewcommand{\arraystretch}{1.15}
\begin{tabular*}{\linewidth}{@{\extracolsep{\fill}}l*{7}{r}@{}}
\multicolumn{8}{@{}l}{\textbf{(a) Shared encoder pool}} \\[4pt]
\toprule
Model & SIMS & MOSI & MELD & URF & Daily & WS & CirCor \\
 & MAE $\downarrow$ & MAE $\downarrow$ & W-F1 $\uparrow$ & Acc. $\uparrow$ & Acc. $\uparrow$ & Acc. $\uparrow$ & BAcc. $\uparrow$ \\
\midrule
Qwen3-1.7B & 64.43 & 144.78 & 50.99 & 66.60 & 28.25 & 33.47 & 46.90 \\
SmolLM2-1.7B & 65.84 & 125.47 & 51.93 & 68.01 & 35.03 & 32.22 & 52.62 \\
XLM-RoBERTa-large & 51.14 & 128.34 & 54.93 & 60.16 & 32.20 & 28.66 & 60.16 \\
\addlinespace[4pt]
Whisper-large-v3 & 46.91 & 142.29 & 44.99 & 69.32 & 32.20 & 29.92 & 79.05 \\
WavLM-base-plus & 50.42 & 133.07 & 45.58 & 63.28 & 26.55 & 28.87 & 80.32 \\
HuBERT-large & 54.27 & 142.42 & 43.50 & 60.06 & 34.46 & 27.82 & 81.11 \\
\addlinespace[4pt]
DINOv2-giant & 59.37 & 250.63 & 29.98 & 52.82 & 30.51 & 28.45 & 87.38 \\
SigLIP2-giant & 47.10 & 358.27 & 35.95 & 56.84 & 32.77 & 27.82 & 81.83 \\
VideoMAE-base & 54.61 & 227.03 & 33.34 & 53.82 & 31.64 & 29.08 & 82.54 \\
\bottomrule
\end{tabular*}
\par\vspace{10pt}
\begin{tabular*}{\linewidth}{@{\extracolsep{\fill}}lr@{\hspace{20pt}}lr@{\hspace{20pt}}lr@{}}
\multicolumn{6}{@{}l}{\textbf{(b) ADNI pool: balanced accuracy (\%)}} \\[4pt]
\toprule
MRI & Score & PET & Score & Clinical & Score \\
\midrule
DINOv2-giant & 53.75 & DINOv2-giant & 70.00 & ADAS & 67.50 \\
SigLIP2-giant & 67.50 & SigLIP2-giant & 62.50 & FAQ & 70.00 \\
FreeSurfer & 66.25 & SUVR & 78.75 & CDR & 73.75 \\
\bottomrule
\end{tabular*}
\endgroup
\end{table}

%% file: tables/individual_expert_circor_views.tex
\begin{table}[!htb]
\centering
\caption{\textbf{Individual CirCor entries across available-site subsets.} Balanced accuracy (\%); MM averages the 14 non-full states. Site abbreviations follow Table~\ref{tab:dataset_protocol}.}
\label{tab:individual_expert_circor_views}
\begingroup
\small
\setlength{\tabcolsep}{3pt}
\renewcommand{\arraystretch}{1.13}
\begin{tabular*}{\linewidth}{@{\extracolsep{\fill}}l*{5}{r}@{}}
\multicolumn{6}{@{}l}{\textbf{(a) Four or three available sites}} \\[4pt]
\toprule
Model & Full & \shortstack{AV+PV\\+TV} & \shortstack{AV+PV\\+MV} & \shortstack{AV+TV\\+MV} & \shortstack{PV+TV\\+MV} \\
\midrule
Qwen3-1.7B & 46.90 & 46.90 & 46.90 & 46.90 & 46.90 \\
SmolLM2-1.7B & 52.62 & 52.62 & 52.62 & 52.62 & 52.62 \\
XLM-RoBERTa-large & 60.16 & 60.16 & 60.16 & 60.16 & 60.16 \\
\addlinespace[3pt]
Whisper-large-v3 & 79.05 & 76.11 & 50.00 & 62.06 & 74.21 \\
WavLM-base-plus & 80.32 & 55.79 & 50.00 & 50.00 & 58.73 \\
HuBERT-large & 81.11 & 50.00 & 55.71 & 61.11 & 50.00 \\
\addlinespace[3pt]
DINOv2-giant & 87.38 & 70.00 & 60.71 & 57.14 & 55.00 \\
SigLIP2-giant & 81.83 & 50.71 & 50.00 & 50.00 & 77.78 \\
VideoMAE-base & 82.54 & 77.30 & 50.00 & 73.41 & 50.71 \\
\bottomrule
\end{tabular*}
\par\vspace{9pt}
\begin{tabular*}{\linewidth}{@{\extracolsep{\fill}}l*{6}{r}@{}}
\multicolumn{7}{@{}l}{\textbf{(b) Two available sites}} \\[4pt]
\toprule
Model & AV+PV & AV+TV & AV+MV & PV+TV & PV+MV & TV+MV \\
\midrule
Qwen3-1.7B & 46.90 & 46.90 & 46.90 & 46.90 & 46.90 & 46.90 \\
SmolLM2-1.7B & 52.62 & 52.62 & 52.62 & 52.62 & 52.62 & 52.62 \\
XLM-RoBERTa-large & 60.16 & 60.16 & 60.16 & 60.16 & 60.16 & 60.16 \\
\addlinespace[3pt]
Whisper-large-v3 & 50.00 & 61.11 & 50.00 & 79.76 & 50.00 & 72.06 \\
WavLM-base-plus & 50.00 & 50.00 & 50.00 & 50.71 & 50.00 & 50.00 \\
HuBERT-large & 50.00 & 72.54 & 60.40 & 50.00 & 50.00 & 50.00 \\
\addlinespace[3pt]
DINOv2-giant & 50.00 & 50.00 & 50.00 & 50.00 & 50.00 & 50.00 \\
SigLIP2-giant & 50.00 & 50.00 & 50.00 & 52.14 & 50.00 & 50.00 \\
VideoMAE-base & 50.00 & 65.87 & 50.00 & 50.00 & 50.00 & 50.00 \\
\bottomrule
\end{tabular*}
\par\vspace{9pt}
\begin{tabular*}{\linewidth}{@{\extracolsep{\fill}}l*{5}{r}@{}}
\multicolumn{6}{@{}l}{\textbf{(c) One available site and Mean-Missing}} \\[4pt]
\toprule
Model & AV & PV & TV & MV & MM \\
\midrule
Qwen3-1.7B & 46.90 & 46.90 & 46.90 & 46.90 & 46.90 \\
SmolLM2-1.7B & 52.62 & 52.62 & 52.62 & 52.62 & 52.62 \\
XLM-RoBERTa-large & 60.16 & 60.16 & 60.16 & 60.16 & 60.16 \\
\addlinespace[3pt]
Whisper-large-v3 & 50.00 & 50.00 & 71.98 & 50.00 & 60.52 \\
WavLM-base-plus & 50.00 & 50.00 & 50.00 & 50.00 & 51.09 \\
HuBERT-large & 50.00 & 50.00 & 50.00 & 50.00 & 53.55 \\
\addlinespace[3pt]
DINOv2-giant & 50.00 & 50.00 & 50.00 & 50.00 & 53.06 \\
SigLIP2-giant & 50.00 & 50.00 & 50.00 & 50.00 & 52.19 \\
VideoMAE-base & 50.00 & 50.00 & 50.00 & 50.00 & 54.81 \\
\bottomrule
\end{tabular*}
\endgroup
\end{table}

%% file: tables/model_pool_members.tex
\begin{table}[!htb]
\centering
\caption{\textbf{Nested \luna pools.} Entries refer to Table~\ref{tab:graph_expert_pools}: panel (a) for the first seven datasets and panel (b) for ADNI. Entries are listed in inventory order; each nine-model pool contains entries 1--9.}
\label{tab:model_pool_members}
\begingroup
\small
\setlength{\tabcolsep}{13pt}
\renewcommand{\arraystretch}{1.10}
\begin{tabular}{lll}
\toprule
\rowcolor{TblHeader}
\textbf{Dataset} & \textbf{3-member pool} & \textbf{6-member pool} \\
\midrule
SIMS & 3, 5, 9 & 2, 3, 5, 6, 8, 9 \\
MOSI & 1, 5, 8 & 1, 2, 4, 5, 8, 9 \\
MELD & 2, 5, 7 & 1, 2, 4, 5, 7, 9 \\
URF & 2, 6, 7 & 2, 3, 5, 6, 7, 9 \\
\addlinespace[3pt]
Daily & 1, 6, 7 & 1, 2, 6, 7, 8, 9 \\
WS & 3, 6, 9 & 1, 3, 5, 6, 7, 9 \\
\addlinespace[3pt]
CirCor & 1, 4, 6 & 1, 3, 4, 5, 6, 8 \\
ADNI & 3, 4, 9 & 2, 3, 4, 5, 7, 9 \\
\bottomrule
\end{tabular}
\endgroup
\end{table}

%% file: tables/refresh_resource_pool.tex
\begin{table}[!htb]
\centering
\caption{\textbf{\luna at different pool sizes.} MM; units follow Table~\ref{tab:multiseed_overall}. Bold marks the best count within each dataset. Relative changes use the three-model pool as reference, with positive values indicating improvement.}
\label{tab:refresh_resource_pool}
\begingroup
\fontsize{8.1}{10}\selectfont
\setlength{\tabcolsep}{8pt}
\renewcommand{\arraystretch}{1.12}
\begin{tabular}{@{}lrrrrr@{}}
\toprule
\rowcolor{TblHeader}
Dataset & 3 models & 6 models & 9 models & $3\!\to\!6$ (\%) & $3\!\to\!9$ (\%) \\
\midrule
SIMS$\downarrow$ & 46.87 & 42.16 & \textbf{34.21} & +10.06 & +27.02 \\
MOSI$\downarrow$ & 122.63 & 105.04 & \textbf{98.34} & +14.34 & +19.81 \\
MELD$\uparrow$ & 52.81 & 54.50 & \textbf{56.08} & +3.21 & +6.20 \\
URF$\uparrow$ & 65.63 & 67.45 & \textbf{70.61} & +2.78 & +7.59 \\
Daily$\uparrow$ & 35.59 & 37.85 & \textbf{38.42} & +6.35 & +7.95 \\
WS$\uparrow$ & 29.11 & 30.86 & \textbf{33.82} & +5.99 & +16.16 \\
CirCor$\uparrow$ & 78.78 & 78.40 & \textbf{82.56} & -0.48 & +4.80 \\
ADNI$\uparrow$ & 82.71 & 76.67 & \textbf{84.17} & -7.30 & +1.77 \\
\bottomrule
\end{tabular}
\endgroup
\end{table}

%% file: tables/refresh_resource_replacement.tex
\begin{table}[!htb]
\centering
\caption{\textbf{Nine-model pool replacements.} Replace 1--3 cumulatively changes model membership, followed by retraining and task-specific retuning. Scores are MM with Table~\ref{tab:multiseed_overall} metrics and units. $\Delta_r$ is the direction-corrected relative MM change from the original pool after $r$ cumulative replacements.}
\label{tab:refresh_resource_replacement}
\begingroup
\fontsize{8.1}{10}\selectfont
\setlength{\tabcolsep}{4pt}
\renewcommand{\arraystretch}{1.12}
\begin{tabular}{@{}lrrrrrrr@{}}
\toprule
\rowcolor{TblHeader}
Dataset & Original & Replace 1 & Replace 2 & Replace 3 & $\Delta_1$ (\%) & $\Delta_2$ (\%) & $\Delta_3$ (\%) \\
\midrule
SIMS & 34.43 & 35.19 & 38.72 & 38.10 & -2.19 & -12.44 & -10.66 \\
MOSI & 96.76 & 98.14 & 108.19 & 107.57 & -1.42 & -11.81 & -11.17 \\
MELD & 55.98 & 55.88 & 53.57 & 53.73 & -0.18 & -4.30 & -4.02 \\
URF & 71.11 & 70.72 & 67.99 & 68.43 & -0.54 & -4.39 & -3.77 \\
Daily & 37.85 & 37.01 & 37.85 & 37.57 & -2.24 & +0.00 & -0.75 \\
WS & 33.05 & 32.64 & 32.71 & 32.71 & -1.27 & -1.05 & -1.05 \\
CirCor & 83.02 & 82.22 & 80.98 & 82.37 & -0.96 & -2.45 & -0.78 \\
ADNI & 83.75 & 83.12 & 83.12 & 82.08 & -0.75 & -0.75 & -1.99 \\
\bottomrule
\end{tabular}
\endgroup
\end{table}

%% file: tables/refresh_resource_parameters.tex
\begin{table}[!htb]
\centering
\caption{\textbf{Trainable collaboration parameters (millions).} Counts match Figure~\ref{fig:resource_parameters_all} and exclude frozen backbones, feature encoders, and separately fitted frozen task heads. Mean ens.\ optimizes no collaboration parameters; Omni models are omitted.}
\label{tab:refresh_resource_parameters}
\begingroup
\fontsize{7.7}{10}\selectfont
\setlength{\tabcolsep}{3pt}
\renewcommand{\arraystretch}{1.12}
\begin{tabular}{@{}lrrrrrrrr@{}}
\toprule
\rowcolor{TblHeader}
Method & SIMS & MOSI & MELD & URF & Daily & WS & CirCor & ADNI \\
\midrule
Mean ens. & 0 & 0 & 0 & 0 & 0 & 0 & 0 & 0 \\
Logit stack & 0.0102 & 0.0102 & 0.0410 & 0.0154 & 0.0102 & 0.0102 & 0.0430 & 0.0154 \\
Concat MLP & 6.1901 & 6.1901 & 6.1932 & 6.1906 & 0.5978 & 0.5978 & 16.8582 & 0.2380 \\
Attn. pool & 21.2199 & 21.2199 & 21.2214 & 21.2201 & 0.9608 & 0.9608 & 52.3718 & 0.7802 \\
Soft MoE & 21.2219 & 21.2219 & 21.2358 & 21.2242 & 0.9628 & 0.9628 & 52.3851 & 0.7843 \\
FoE & 0.0191 & 0.0191 & 0.0268 & 0.0204 & 0.0191 & 0.0191 & 0.0273 & 0.0204 \\
Set Trans. & 0.7892 & 0.7892 & 0.7900 & 0.7894 & 0.6418 & 0.6418 & 1.3861 & 0.5522 \\
CLARGA & 0.4391 & 0.4391 & 0.4399 & 0.4392 & 0.2916 & 0.2916 & 1.0360 & 0.2021 \\
GIML & 8.1434 & 8.1434 & 8.1526 & 8.1449 & 1.9118 & 2.5728 & 19.6885 & 2.2154 \\
MARS & 8.0437 & 8.0437 & 8.0467 & 8.0442 & 6.6179 & 6.6510 & 10.7647 & 6.5618 \\
LIMSSR & 4.2492 & 4.2492 & 4.2508 & 4.2495 & 2.4397 & 2.8584 & 5.0976 & 2.7725 \\
SimMLM & 0.3983 & 0.3983 & 0.4007 & 0.3987 & 0.2176 & 0.2509 & 0.4656 & 0.1615 \\
ShaSpec & 0.5146 & 0.5146 & 0.5153 & 0.5147 & 0.3172 & 0.3671 & 0.6145 & 0.2775 \\
Flex-MoE & 0.8157 & 0.8157 & 0.8165 & 0.8158 & 0.4673 & 0.6682 & 2.0268 & 0.5786 \\
\rowcolor{TblOurs}
Lluna & 1.0743 & 3.7523 & 2.6690 & 1.2852 & 0.9035 & 3.0790 & 7.2514 & 0.9374 \\
\bottomrule
\end{tabular}
\endgroup
\end{table}

%% file: tables/refresh_modality_complete.tex
\begin{table}[!htbp]
\centering
\begin{minipage}[t]{.54\linewidth}
\vspace{0pt}
\centering
\captionsetup{skip=5pt}
\caption{\textbf{Modality contributions.} Single-checkpoint Shapley utilities across test examples. Mean/median use negative-absolute-error units for SIMS/MOSI and nats otherwise; Positive/Largest are percentages of test examples.}
\label{tab:refresh_modality_sources}
\begingroup
\fontsize{8}{9}\selectfont
\setlength{\tabcolsep}{3pt}
\renewcommand{\arraystretch}{1.12}
\begin{tabular}{@{}llrrrr@{}}
\toprule
\rowcolor{TblHeader}
Task & Modality & Mean & Median & Positive (\%) & Largest (\%) \\
\midrule
SIMS & A & $+0.152$ & $+0.187$ & $72.24$ & $38.88$ \\
 & T & $+0.128$ & $+0.161$ & $67.50$ & $32.40$ \\
 & V & $+0.076$ & $+0.085$ & $63.15$ & $28.72$ \\
\addlinespace[3pt]
MOSI & A & $+0.171$ & $+0.168$ & $68.37$ & $33.97$ \\
 & T & $+0.276$ & $+0.313$ & $73.91$ & $53.35$ \\
 & V & $-0.054$ & $-0.037$ & $40.52$ & $12.68$ \\
\addlinespace[3pt]
MELD & A & $+0.298$ & $+0.479$ & $75.00$ & $25.23$ \\
 & T & $+0.495$ & $+0.714$ & $79.33$ & $60.47$ \\
 & V & $+0.150$ & $+0.264$ & $70.97$ & $14.30$ \\
\addlinespace[3pt]
URF & A & $+0.070$ & $+0.113$ & $67.10$ & $22.03$ \\
 & T & $+0.145$ & $+0.335$ & $74.55$ & $62.37$ \\
 & V & $+0.008$ & $+0.001$ & $50.50$ & $15.59$ \\
\addlinespace[3pt]
Daily & A & $-0.335$ & $-0.102$ & $46.89$ & $45.76$ \\
 & V & $+0.080$ & $+0.033$ & $59.89$ & $54.24$ \\
\addlinespace[3pt]
WS & A & $+0.013$ & $+0.002$ & $53.35$ & $25.73$ \\
 & S & $-0.007$ & $+0.035$ & $53.56$ & $53.14$ \\
 & V & $+0.013$ & $+0.002$ & $52.93$ & $21.13$ \\
\addlinespace[3pt]
CirCor & AV & $-0.222$ & $+0.075$ & $60.23$ & $7.95$ \\
 & MV & $-0.061$ & $+0.168$ & $70.45$ & $6.82$ \\
 & PV & $+0.149$ & $+0.217$ & $84.09$ & $23.86$ \\
 & TV & $+0.256$ & $+0.267$ & $87.50$ & $61.36$ \\
\addlinespace[3pt]
ADNI & C & $+0.231$ & $+0.291$ & $90.00$ & $70.00$ \\
 & M & $+0.094$ & $+0.157$ & $72.00$ & $12.00$ \\
 & P & $+0.178$ & $+0.256$ & $84.00$ & $18.00$ \\
\bottomrule
\end{tabular}
\endgroup
\end{minipage}
\hfill
\begin{minipage}[t]{.44\linewidth}
\vspace{0pt}
\centering
\captionsetup{skip=5pt}
\caption{\textbf{Complete pair interactions.} Positive interactions increase joint utility above the additive reference; negative values are subadditive. Units follow Table~\ref{tab:refresh_modality_sources}; synergy is the share of examples with positive interaction.}
\label{tab:refresh_modality_pairs}
\begingroup
\fontsize{8}{9}\selectfont
\setlength{\tabcolsep}{3pt}
\renewcommand{\arraystretch}{1.12}
\begin{tabular}{@{}llrrr@{}}
\toprule
\rowcolor{TblHeader}
Task & Pair & Mean & Median & Synergy (\%) \\
\midrule
SIMS & A+T & $-0.141$ & $-0.149$ & $33.46$ \\
 & A+V & $-0.069$ & $-0.078$ & $39.26$ \\
 & T+V & $-0.051$ & $-0.065$ & $41.30$ \\
\addlinespace[3pt]
MOSI & A+T & $-0.289$ & $-0.324$ & $26.97$ \\
 & A+V & $+0.022$ & $+0.006$ & $51.31$ \\
 & T+V & $+0.027$ & $+0.021$ & $56.27$ \\
\addlinespace[3pt]
MELD & A+T & $-0.449$ & $-0.576$ & $22.24$ \\
 & A+V & $-0.219$ & $-0.312$ & $34.28$ \\
 & T+V & $-0.205$ & $-0.339$ & $34.47$ \\
\addlinespace[3pt]
URF & A+T & $-0.064$ & $-0.081$ & $39.34$ \\
 & A+V & $+0.004$ & $-0.030$ & $40.14$ \\
 & T+V & $+0.036$ & $-0.013$ & $45.98$ \\
\addlinespace[3pt]
Daily & A+V & $+0.160$ & $+0.061$ & $59.89$ \\
\addlinespace[3pt]
WS & A+S & $+0.029$ & $+0.003$ & $52.30$ \\
 & A+V & $-0.015$ & $-0.002$ & $46.03$ \\
 & S+V & $+0.029$ & $+0.003$ & $52.30$ \\
\addlinespace[3pt]
CirCor & AV+MV & $+0.156$ & $-0.044$ & $47.73$ \\
 & AV+PV & $+0.235$ & $-0.113$ & $44.32$ \\
 & AV+TV & $+0.211$ & $-0.083$ & $42.05$ \\
 & MV+PV & $-0.060$ & $-0.210$ & $32.95$ \\
 & MV+TV & $-0.115$ & $-0.231$ & $25.00$ \\
 & PV+TV & $-0.170$ & $-0.249$ & $20.45$ \\
\addlinespace[3pt]
ADNI & C+M & $-0.002$ & $-0.206$ & $38.00$ \\
 & C+P & $-0.174$ & $-0.369$ & $22.00$ \\
 & M+P & $-0.049$ & $-0.195$ & $38.00$ \\
\bottomrule
\end{tabular}
\endgroup
\end{minipage}
\end{table}
\FloatBarrier

%% file: tables/sampling_budget.tex
\begin{table}[!htb]
\centering
\caption{\textbf{Complete chain-budget results.} MM is weighted F1 for MELD and accuracy for URF, in percent. $K$ counts sampled maximal chains; pairs are the mean distinct PCC comparisons per training step.}
\label{tab:sampling_budget}
\begingroup\small
\setlength{\tabcolsep}{8pt}
\renewcommand{\arraystretch}{1.10}
\begin{tabular}{lrrr}
\toprule
\rowcolor{TblHeader} \textbf{Dataset} & \textbf{$K$} & \textbf{PCC pairs} & \textbf{MM (\%)} \\
\midrule
MELD & 1 & 2.00 & 54.35 \\
 & 2 & 3.57 & 53.41 \\
 & 3 & 4.71 & 54.44 \\
\addlinespace[4pt]
URF & 1 & 2.00 & 69.55 \\
 & 2 & 3.48 & 69.13 \\
 & 3 & 4.64 & 69.08 \\
\bottomrule
\end{tabular}
\endgroup
\end{table}

%% file: tables/refresh_sampling_operating.tex
\begin{table}[!htb]
\captionsetup{skip=4pt}
\centering
\caption{\textbf{Sparse versus full training.} Full denotes full-lattice training. Scores: MM in Table~\ref{tab:multiseed_overall} units; bold marks better. $\Delta$: direction-corrected relative MM change (positive is better). \# PCC: mean distinct comparisons per step (Full $\to$ Sparse). All nonempty states are tested.}
\label{tab:refresh_sparse_operating}
\begingroup\small
\setlength{\tabcolsep}{5pt}
\renewcommand{\arraystretch}{1.10}
\begin{tabular}{@{}lrrrc@{}}
\toprule
\rowcolor{TblHeader}
\textbf{Dataset} & \textbf{Full} & \textbf{Sparse} & \textbf{$\Delta$ (\%)} & \textbf{\# PCC} \\
\midrule
SIMS & $35.87$ & $\mathbf{35.37}$ & $+1.40$ & $9 \to 2$ \\
MOSI & $100.99$ & $\mathbf{100.08}$ & $+0.91$ & $9 \to 3.48$ \\
MELD & $55.28$ & $\mathbf{55.55}$ & $+0.49$ & $9 \to 2$ \\
URF & $\mathbf{70.82}$ & $70.09$ & $-1.04$ & $9 \to 2$ \\
\addlinespace[1.5pt]
Daily & $\mathbf{36.44}$ & $36.16$ & $-0.78$ & $2 \to 1$ \\
WS & $\mathbf{33.05}$ & $33.02$ & $-0.11$ & $9 \to 7$ \\
\addlinespace[1.5pt]
CirCor & $\mathbf{82.32}$ & $79.97$ & $-2.85$ & $28 \to 3$ \\
ADNI & $\mathbf{83.75}$ & $83.54$ & $-0.25$ & $9 \to 5$ \\
\bottomrule
\end{tabular}
\endgroup
\end{table}

%% file: tables/fixed_anchor_combined.tex
\centering
\begingroup\fontsize{9}{10.8}\selectfont
\setlength{\tabcolsep}{7pt}
\renewcommand{\arraystretch}{1.22}
\begin{tabular}{lr*{2}{r}*{2}{>{\columncolor{TblOurs}}r}}
\toprule
\rowcolor{TblHeader} & & \multicolumn{2}{c}{Baselines} & \multicolumn{2}{c}{Lluna} \\
\rowcolor{TblHeader} Dataset & $N_{\mathrm{test}}$ & Full & Min. & Full & Min. \\
SIMS & 1034 & +10.07 & -0.28 & +10.93 & +4.48 \\
MOSI & 686 & +16.65 & -0.08 & +23.23 & +11.13 \\
MELD & 2608 & +16.29 & +1.05 & +19.79 & +8.93 \\
URF & 994 & +9.31 & -2.36 & +14.05 & +3.29 \\
Daily & 177 & +0.77 & -2.10 & +4.24 & +5.93 \\
WS & 478 & +2.14 & -1.20 & +3.70 & +2.86 \\
CirCor & 88 & -17.09 & -21.59 & +5.97 & +1.70 \\
ADNI & 50 & +9.05 & +1.24 & +15.33 & +3.33 \\
\bottomrule
\end{tabular}
\endgroup

%% file: tables/refresh_diag_transitions.tex
\begin{table}[!htb]
\centering
\caption{\textbf{Fixed-anchor prediction transitions.} H/R are harmed/rescued percentages of all test examples; Net $=\mathrm{R}-\mathrm{H}$ (pp). Independent rounding can yield 0.01-pp discrepancies. All tasks use ordinary accuracy, including auxiliary Acc-2(has0) for SIMS/MOSI. Conventional averages four methods equally.}
\label{tab:refresh_diag_transitions}
\begingroup\fontsize{8.5}{10.2}\selectfont
\setlength{\tabcolsep}{3.4pt}
\renewcommand{\arraystretch}{1.13}
\begin{tabular}{@{}lr*{6}{r}@{}}
\toprule
\rowcolor{TblHeader} & & \multicolumn{3}{c}{Conventional mean} & \multicolumn{3}{c}{Lluna} \\
\rowcolor{TblHeader} Dataset & Removed & H (\%) & R (\%) & Net (pp) & H (\%) & R (\%) & Net (pp) \\
SIMS & 0 & 8.09 & 18.05 & +9.96 & 7.03 & 17.96 & +10.93 \\
 & 1 & 11.67 & 16.55 & +4.88 & 7.41 & 16.51 & +9.09 \\
 & 2 & 16.53 & 14.35 & -2.18 & 8.74 & 13.22 & +4.48 \\
\addlinespace[2.5pt]
MOSI & 0 & 10.22 & 30.49 & +20.27 & 8.79 & 32.02 & +23.23 \\
 & 1 & 16.33 & 25.24 & +8.90 & 9.50 & 30.98 & +21.48 \\
 & 2 & 22.18 & 19.38 & -2.81 & 12.88 & 24.00 & +11.13 \\
\addlinespace[2.5pt]
MELD & 0 & 9.88 & 25.99 & +16.11 & 7.17 & 26.96 & +19.79 \\
 & 1 & 13.53 & 21.57 & +8.04 & 7.29 & 23.60 & +16.31 \\
 & 2 & 18.84 & 14.65 & -4.18 & 7.89 & 16.82 & +8.93 \\
\addlinespace[2.5pt]
URF & 0 & 11.33 & 22.46 & +11.13 & 8.69 & 22.74 & +14.05 \\
 & 1 & 15.43 & 19.87 & +4.43 & 10.63 & 20.39 & +9.76 \\
 & 2 & 20.23 & 16.83 & -3.40 & 13.48 & 16.77 & +3.29 \\
\addlinespace[2.5pt]
Daily & 0 & 17.16 & 17.44 & +0.28 & 16.95 & 21.19 & +4.24 \\
 & 1 & 17.87 & 14.69 & -3.18 & 14.12 & 20.06 & +5.93 \\
\addlinespace[2.5pt]
WS & 0 & 14.07 & 16.14 & +2.07 & 14.23 & 17.92 & +3.70 \\
 & 1 & 11.99 & 12.06 & +0.07 & 13.08 & 17.12 & +4.04 \\
 & 2 & 10.04 & 8.16 & -1.88 & 10.53 & 13.39 & +2.86 \\
\addlinespace[2.5pt]
CirCor & 0 & 23.72 & 8.66 & -15.06 & 3.69 & 9.66 & +5.97 \\
 & 1 & 24.81 & 8.40 & -16.41 & 4.36 & 10.32 & +5.97 \\
 & 2 & 26.70 & 8.29 & -18.42 & 5.02 & 10.04 & +5.02 \\
 & 3 & 30.75 & 8.45 & -22.30 & 7.95 & 9.66 & +1.70 \\
\addlinespace[2.5pt]
ADNI & 0 & 7.33 & 16.67 & +9.33 & 3.33 & 18.67 & +15.33 \\
 & 1 & 8.92 & 15.25 & +6.33 & 4.33 & 17.67 & +13.33 \\
 & 2 & 10.83 & 10.50 & -0.33 & 5.33 & 8.67 & +3.33 \\
\bottomrule
\end{tabular}
\endgroup
\end{table}

%% file: tables/supplement_anchor_summary.tex
\begin{table}[!htbp]
\centering
\captionsetup{skip=4pt}
\caption{\textbf{Cooperation gains after missing-input training.} Net accuracy gain over a fixed surviving model, averaged equally over I$^2$MoE, CLARGA, and Set Transformer. Full retains all modalities; Minimal retains one variable modality. Drop is Full minus Minimal, in percentage points.}
\label{tab:supplement_anchor_summary}
\begingroup\fontsize{8.5}{10.2}\selectfont
\setlength{\tabcolsep}{12pt}
\renewcommand{\arraystretch}{1.13}
\begin{tabular}{@{}lrrr@{}}
\toprule
\rowcolor{TblHeader}Dataset & Full & Minimal & Drop \\
\midrule
SIMS & +10.22 & +2.27 & 7.95 \\
MOSI & +11.81 & +3.56 & 8.25 \\
MELD & +16.53 & +8.02 & 8.50 \\
URF & +6.88 & -0.97 & 7.85 \\
\addlinespace[3pt]
Daily & +1.41 & -0.66 & 2.07 \\
WS & +2.23 & -0.28 & 2.51 \\
\addlinespace[3pt]
CirCor & -19.79 & -20.64 & 0.85 \\
ADNI & +8.67 & +3.33 & 5.33 \\
\bottomrule
\end{tabular}
\endgroup
\end{table}

%% file: tables/supplement_anchor_methods.tex
\begin{table}[!htbp]
\centering
\captionsetup{skip=4pt}
\caption{\textbf{Fixed-anchor transitions on SIMS and MOSI.} Harmed and Rescued are percentages of test examples; Net is Rescued minus Harmed, in percentage points. Missing counts removed modalities.}
\label{tab:supplement_anchor_methods_1}
\begingroup\fontsize{8.5}{10.2}\selectfont
\setlength{\tabcolsep}{3pt}\renewcommand{\arraystretch}{1.13}
\begin{minipage}[t]{.485\linewidth}
\vspace{0pt}\centering
\textbf{SIMS}\par\smallskip
\begin{tabular}{@{}lrrrr@{}}
\toprule
\rowcolor{TblHeader}Method & Missing & Harmed & Rescued & Net \\
\midrule
I$^2$MoE & 0 & 7.16 & 18.57 & +11.41 \\
 & 1 & 8.85 & 16.84 & +8.00 \\
 & 2 & 10.83 & 13.44 & +2.61 \\
\addlinespace[3pt]
CLARGA & 0 & 8.00 & 18.15 & +10.16 \\
 & 1 & 8.98 & 17.20 & +8.22 \\
 & 2 & 11.54 & 14.18 & +2.64 \\
\addlinespace[3pt]
Set Trans. & 0 & 8.45 & 17.54 & +9.09 \\
 & 1 & 9.48 & 15.93 & +6.45 \\
 & 2 & 11.73 & 13.28 & +1.55 \\
\bottomrule
\end{tabular}
\end{minipage}
\hfill
\begin{minipage}[t]{.485\linewidth}
\vspace{0pt}\centering
\textbf{MOSI}\par\smallskip
\begin{tabular}{@{}lrrrr@{}}
\toprule
\rowcolor{TblHeader}Method & Missing & Harmed & Rescued & Net \\
\midrule
I$^2$MoE & 0 & 13.36 & 26.68 & +13.31 \\
 & 1 & 13.51 & 25.17 & +11.66 \\
 & 2 & 16.18 & 20.46 & +4.28 \\
\addlinespace[3pt]
CLARGA & 0 & 13.02 & 26.63 & +13.61 \\
 & 1 & 14.09 & 25.12 & +11.03 \\
 & 2 & 15.94 & 19.19 & +3.26 \\
\addlinespace[3pt]
Set Trans. & 0 & 16.04 & 24.54 & +8.50 \\
 & 1 & 17.10 & 22.69 & +5.59 \\
 & 2 & 16.28 & 19.44 & +3.16 \\
\bottomrule
\end{tabular}
\end{minipage}
\endgroup
\end{table}

\begin{table}[!htbp]
\centering
\captionsetup{skip=4pt}
\caption{\textbf{Fixed-anchor transitions on MELD and URF.} Harmed and Rescued are percentages of test examples; Net is Rescued minus Harmed, in percentage points. Missing counts removed modalities.}
\label{tab:supplement_anchor_methods_2}
\begingroup\fontsize{8.5}{10.2}\selectfont
\setlength{\tabcolsep}{3pt}\renewcommand{\arraystretch}{1.13}
\begin{minipage}[t]{.485\linewidth}
\vspace{0pt}\centering
\textbf{MELD}\par\smallskip
\begin{tabular}{@{}lrrrr@{}}
\toprule
\rowcolor{TblHeader}Method & Missing & Harmed & Rescued & Net \\
\midrule
I$^2$MoE & 0 & 7.41 & 25.09 & +17.68 \\
 & 1 & 7.91 & 22.25 & +14.34 \\
 & 2 & 9.00 & 16.72 & +7.72 \\
\addlinespace[3pt]
CLARGA & 0 & 7.62 & 23.72 & +16.10 \\
 & 1 & 8.05 & 21.47 & +13.42 \\
 & 2 & 8.63 & 16.72 & +8.09 \\
\addlinespace[3pt]
Set Trans. & 0 & 7.99 & 23.79 & +15.80 \\
 & 1 & 8.44 & 21.26 & +12.82 \\
 & 2 & 8.30 & 16.55 & +8.26 \\
\bottomrule
\end{tabular}
\end{minipage}
\hfill
\begin{minipage}[t]{.485\linewidth}
\vspace{0pt}\centering
\textbf{URF}\par\smallskip
\begin{tabular}{@{}lrrrr@{}}
\toprule
\rowcolor{TblHeader}Method & Missing & Harmed & Rescued & Net \\
\midrule
I$^2$MoE & 0 & 12.68 & 20.19 & +7.51 \\
 & 1 & 13.92 & 19.32 & +5.40 \\
 & 2 & 17.40 & 16.53 & -0.87 \\
\addlinespace[3pt]
CLARGA & 0 & 13.58 & 19.99 & +6.41 \\
 & 1 & 14.64 & 18.63 & +3.99 \\
 & 2 & 17.00 & 15.63 & -1.38 \\
\addlinespace[3pt]
Set Trans. & 0 & 13.62 & 20.32 & +6.71 \\
 & 1 & 15.11 & 19.06 & +3.96 \\
 & 2 & 17.20 & 16.53 & -0.67 \\
\bottomrule
\end{tabular}
\end{minipage}
\endgroup
\end{table}

\begin{table}[!htbp]
\centering
\captionsetup{skip=4pt}
\caption{\textbf{Fixed-anchor transitions on Daily and WS.} Harmed and Rescued are percentages of test examples; Net is Rescued minus Harmed, in percentage points. Missing counts removed modalities.}
\label{tab:supplement_anchor_methods_3}
\begingroup\fontsize{8.5}{10.2}\selectfont
\setlength{\tabcolsep}{3pt}\renewcommand{\arraystretch}{1.13}
\begin{minipage}[t]{.485\linewidth}
\vspace{0pt}\centering
\textbf{Daily}\par\smallskip
\begin{tabular}{@{}lrrrr@{}}
\toprule
\rowcolor{TblHeader}Method & Missing & Harmed & Rescued & Net \\
\midrule
I$^2$MoE & 0 & 16.67 & 15.82 & -0.85 \\
 & 1 & 14.97 & 12.71 & -2.26 \\
\addlinespace[3pt]
CLARGA & 0 & 16.38 & 18.36 & +1.98 \\
 & 1 & 13.84 & 13.56 & -0.28 \\
\addlinespace[3pt]
Set Trans. & 0 & 16.10 & 19.21 & +3.11 \\
 & 1 & 13.28 & 13.84 & +0.57 \\
\bottomrule
\end{tabular}
\end{minipage}
\hfill
\begin{minipage}[t]{.485\linewidth}
\vspace{0pt}\centering
\textbf{WS}\par\smallskip
\begin{tabular}{@{}lrrrr@{}}
\toprule
\rowcolor{TblHeader}Method & Missing & Harmed & Rescued & Net \\
\midrule
I$^2$MoE & 0 & 13.81 & 15.62 & +1.81 \\
 & 1 & 9.83 & 10.53 & +0.70 \\
 & 2 & 6.49 & 5.51 & -0.98 \\
\addlinespace[3pt]
CLARGA & 0 & 11.58 & 14.85 & +3.28 \\
 & 1 & 8.65 & 10.39 & +1.74 \\
 & 2 & 4.95 & 4.74 & -0.21 \\
\addlinespace[3pt]
Set Trans. & 0 & 13.81 & 15.41 & +1.60 \\
 & 1 & 9.31 & 10.50 & +1.19 \\
 & 2 & 5.58 & 5.93 & +0.35 \\
\bottomrule
\end{tabular}
\end{minipage}
\endgroup
\end{table}

\begin{table}[!htbp]
\centering
\captionsetup{skip=4pt}
\caption{\textbf{Fixed-anchor transitions on CirCor and ADNI.} Harmed and Rescued are percentages of test examples; Net is Rescued minus Harmed, in percentage points. Missing counts removed modalities.}
\label{tab:supplement_anchor_methods_4}
\begingroup\fontsize{8.5}{10.2}\selectfont
\setlength{\tabcolsep}{3pt}\renewcommand{\arraystretch}{1.13}
\begin{minipage}[t]{.485\linewidth}
\vspace{0pt}\centering
\textbf{CirCor}\par\smallskip
\begin{tabular}{@{}lrrrr@{}}
\toprule
\rowcolor{TblHeader}Method & Missing & Harmed & Rescued & Net \\
\midrule
I$^2$MoE & 0 & 6.53 & 10.23 & +3.69 \\
 & 1 & 7.29 & 9.85 & +2.56 \\
 & 2 & 7.67 & 9.28 & +1.61 \\
 & 3 & 7.10 & 8.24 & +1.14 \\
\addlinespace[3pt]
CLARGA & 0 & 6.82 & 8.24 & +1.42 \\
 & 1 & 5.11 & 8.52 & +3.41 \\
 & 2 & 5.30 & 8.05 & +2.75 \\
 & 3 & 6.25 & 7.67 & +1.42 \\
\addlinespace[3pt]
Set Trans. & 0 & 72.16 & 7.67 & -64.49 \\
 & 1 & 72.16 & 7.67 & -64.49 \\
 & 2 & 72.16 & 7.67 & -64.49 \\
 & 3 & 72.16 & 7.67 & -64.49 \\
\bottomrule
\end{tabular}
\end{minipage}
\hfill
\begin{minipage}[t]{.485\linewidth}
\vspace{0pt}\centering
\textbf{ADNI}\par\smallskip
\begin{tabular}{@{}lrrrr@{}}
\toprule
\rowcolor{TblHeader}Method & Missing & Harmed & Rescued & Net \\
\midrule
I$^2$MoE & 0 & 10.00 & 15.33 & +5.33 \\
 & 1 & 8.33 & 14.33 & +6.00 \\
 & 2 & 6.67 & 8.00 & +1.33 \\
\addlinespace[3pt]
CLARGA & 0 & 10.67 & 16.00 & +5.33 \\
 & 1 & 12.00 & 14.00 & +2.00 \\
 & 2 & 6.67 & 11.33 & +4.67 \\
\addlinespace[3pt]
Set Trans. & 0 & 5.33 & 20.67 & +15.33 \\
 & 1 & 6.67 & 16.00 & +9.33 \\
 & 2 & 6.67 & 10.67 & +4.00 \\
\bottomrule
\end{tabular}
\end{minipage}
\endgroup
\end{table}

%% file: tables/motivation_remove2.tex
\begin{table}[!htb]
\centering
\setlength{\tabcolsep}{5pt}
\caption{\textbf{Exact-removal collaboration gains (pp).} Conventional averages four methods equally. Drop is Full minus Remove~2; datasets are equally weighted within each domain.}
\label{tab:motivation_remove2}
\small
\begin{tabular}{llrrrr}
\toprule
Domain & Predictor & Full & Remove 1 & Remove 2 & Drop (pp) \\
\midrule
Affective & Conventional & $14.37$ & $6.56$ & $-3.14$ & $17.51$ \\
 & \luna & $17.00$ & $14.16$ & $6.96$ & $10.04$ \\
Scene & Conventional & $2.07$ & $0.07$ & $-1.88$ & $3.95$ \\
 & \luna & $3.70$ & $4.04$ & $2.86$ & $0.84$ \\
Medical & Conventional & $-2.86$ & $-5.04$ & $-9.37$ & $6.51$ \\
 & \luna & $10.65$ & $9.65$ & $4.17$ & $6.47$ \\
\bottomrule
\end{tabular}
\end{table}

%% file: tables/private_probe_multimethod.tex
\begin{table}[!htb]
\centering
\caption{\textbf{Private-signal decodability across collaboration methods.} Probe $R^2$ (\%), averaged equally over eight datasets. Drop is Minimal minus Full (pp); Full-only/All-views specify the training policy.}
\label{tab:private_probe_multimethod}
\begingroup
\small
\setlength{\tabcolsep}{7pt}
\renewcommand{\arraystretch}{1.12}
\begin{tabular}{@{}lrrrrrr@{}}
\toprule
\rowcolor{TblHeader}
\textbf{Method} & \multicolumn{3}{c}{\textbf{Full-only training}} & \multicolumn{3}{c}{\textbf{All-views training}} \\
\cmidrule(lr){2-4}\cmidrule(lr){5-7}
 & Minimal & Full & Drop & Minimal & Full & Drop \\
\midrule
Logit stack & 7.81 & 6.25 & 1.56 & 7.77 & 6.28 & 1.49 \\
Concat MLP & 40.95 & 20.91 & 20.03 & 40.45 & 20.45 & 20.00 \\
Attn. pool & 57.15 & 28.34 & 28.81 & 58.38 & 33.52 & 24.86 \\
Soft MoE & 56.85 & 23.13 & 33.72 & 57.52 & 27.31 & 30.21 \\
\bottomrule
\end{tabular}
\endgroup
\end{table}

%% file: tables/refresh_diag_jsd.tex
\begin{table}[!htb]
\centering
\caption{\textbf{Prediction differences across availability states.} Cells show $10^3\mathrm{JSD}/\ln 2$ for fixed diagnostic models, with ratios to Adjacent in parentheses. Examples gives the number of evaluated items in each diagnostic cohort. Regression JSD is N/A; Daily's Adjacent and Full-to-minimal sets coincide.}
\label{tab:refresh_diag_jsd}
\begingroup\fontsize{8.5}{10.2}\selectfont
\setlength{\tabcolsep}{3.7pt}
\renewcommand{\arraystretch}{1.13}
\begin{tabular}{@{}lr*{3}{r}@{}}
\toprule
\rowcolor{TblHeader} Dataset & Examples & Adjacent & Same-level & Full-to-minimal \\
SIMS & 4403 & N/A & N/A & N/A \\
MOSI & 2199 & N/A & N/A & N/A \\
MELD & 13704 & 51.43 (1.00$\times$) & 88.86 (1.73$\times$) & 90.23 (1.75$\times$) \\
URF & 9588 & 16.18 (1.00$\times$) & 28.37 (1.75$\times$) & 27.60 (1.71$\times$) \\
Daily & 1019 & 69.63 (1.00$\times$) & 166.61 (2.39$\times$) & 69.63 (1.00$\times$) \\
WS & 2694 & 19.48 (1.00$\times$) & 35.73 (1.83$\times$) & 34.41 (1.77$\times$) \\
CirCor & 585 & 17.49 (1.00$\times$) & 17.06 (0.98$\times$) & 49.30 (2.82$\times$) \\
ADNI & 284 & 31.26 (1.00$\times$) & 35.60 (1.14$\times$) & 56.61 (1.81$\times$) \\
\bottomrule
\end{tabular}
\endgroup
\end{table}

%% file: refs.bib
@article{osborne2024ai,
  title={The AI community building the future? A quantitative analysis of development activity on Hugging Face Hub},
  author={Osborne, Cailean and Ding, Jennifer and Kirk, Hannah Rose},
  journal={Journal of Computational Social Science},
  volume={7},
  number={2},
  pages={2067--2105},
  year={2024},
  publisher={Springer}
}

@article{talebirad2023multi,
  title={Multi-agent collaboration: Harnessing the power of intelligent llm agents},
  author={Talebirad, Yashar and Nadiri, Amirhossein},
  journal={arXiv preprint arXiv:2306.03314},
  year={2023}
}

@article{baltrusaitis2019multimodal,
  title={Multimodal machine learning: A survey and taxonomy},
  author={Baltru{\v{s}}aitis, Tadas and Ahuja, Chaitanya and Morency, Louis-Philippe},
  journal={IEEE transactions on pattern analysis and machine intelligence},
  volume={41},
  number={2},
  pages={423--443},
  year={2018},
  publisher={IEEE}
}

@article{wu2026missingmodalitysurvey,
  title={Deep multimodal learning with missing modality: A survey},
  author={Wu, Renjie and Wang, Hu and Chen, Hsiang-Ting and Carneiro, Gustavo},
  journal={arXiv preprint arXiv:2409.07825},
  year={2024}
}

@article{neverova2016moddrop,
  title={Moddrop: adaptive multi-modal gesture recognition},
  author={Neverova, Natalia and Wolf, Christian and Taylor, Graham and Nebout, Florian},
  journal={IEEE Transactions on Pattern Analysis and Machine Intelligence},
  volume={38},
  number={8},
  pages={1692--1706},
  year={2015},
  publisher={IEEE}
}

@inproceedings{havaei2016hemis,
  title={Hemis: Hetero-modal image segmentation},
  author={Havaei, Mohammad and Guizard, Nicolas and Chapados, Nicolas and Bengio, Yoshua},
  booktitle={International conference on medical image computing and computer-assisted intervention},
  pages={469--477},
  year={2016},
  organization={Springer}
}

@inproceedings{girdhar2023imagebind,
  title={ImageBind one embedding space to bind them all},
  author={Girdhar, Rohit and El-Nouby, Alaaeldin and Liu, Zhuang and Singh, Mannat and Alwala, Kalyan Vasudev and Joulin, Armand and Misra, Ishan},
  booktitle={2023 IEEE/CVF Conference on Computer Vision and Pattern Recognition (CVPR)},
  pages={15180--15190},
  year={2023},
  organization={IEEE}
}

@article{alayrac2022flamingo,
  title={Flamingo: a visual language model for few-shot learning},
  author={Alayrac, Jean-Baptiste and Donahue, Jeff and Luc, Pauline and Miech, Antoine and Barr, Iain and Hasson, Yana and Lenc, Karel and Mensch, Arthur and Millican, Katherine and Reynolds, Malcolm and others},
  journal={Advances in neural information processing systems},
  volume={35},
  pages={23716--23736},
  year={2022}
}

@article{vaswani2017attention,
  title={Attention is all you need},
  author={Vaswani, Ashish and Shazeer, Noam and Parmar, Niki and Uszkoreit, Jakob and Jones, Llion and Gomez, Aidan N and Kaiser, {\L}ukasz and Polosukhin, Illia},
  journal={Advances in neural information processing systems},
  volume={30},
  year={2017}
}

@inproceedings{lee2019set,
  title={Set transformer: A framework for attention-based permutation-invariant neural networks},
  author={Lee, Juho and Lee, Yoonho and Kim, Jungtaek and Kosiorek, Adam and Choi, Seungjin and Teh, Yee Whye},
  booktitle={International conference on machine learning},
  pages={3744--3753},
  year={2019},
  organization={PMLR}
}

@inproceedings{wang2023shaspec,
  title={Multi-modal learning with missing modality via shared-specific feature modelling},
  author={Wang, Hu and Chen, Yuanhong and Ma, Congbo and Avery, Jodie and Hull, Louise and Carneiro, Gustavo},
  booktitle={2023 IEEE/CVF Conference on Computer Vision and Pattern Recognition (CVPR)},
  pages={15878--15887},
  year={2023},
  organization={IEEE}
}

@article{xie2026ccsd,
  title={CCSD: Cross-modal compositional self-distillation for robust brain tumor segmentation with missing modalities},
  author={Xie, Dongqing and Wu, Yonghuang and Ai, Zisheng and Min, Jun and Jiang, Zhencun and Geng, Shaojin and Wang, Lei},
  journal={Knowledge-Based Systems},
  pages={116136},
  year={2026},
  publisher={Elsevier}
}

@inproceedings{maniparambil2025freezealign,
  title={Harnessing frozen unimodal encoders for flexible multimodal alignment},
  author={Maniparambil, Mayug and Akshulakov, Raiymbek and Djilali, Yasser Abdelaziz Dahou and Narayan, Sanath and Singh, Ankit and O'Connor, Noel E},
  booktitle={Proceedings of the IEEE/CVF Conference on Computer Vision and Pattern Recognition},
  pages={29847--29857},
  year={2025}
}

@inproceedings{li2025simmlm,
  title={Simmlm: A simple framework for multi-modal learning with missing modality},
  author={Li, Sijie and Chen, Chen and Han, Jungong},
  booktitle={2025 IEEE/CVF International Conference on Computer Vision (ICCV)},
  pages={24068--24077},
  year={2025},
  organization={IEEE}
}

@article{yun2024flexmoe,
  title={Flex-moe: Modeling arbitrary modality combination via the flexible mixture-of-experts},
  author={Yun, Sukwon and Choi, Inyoung and Peng, Jie and Wu, Yangfan and Bao, Jingxuan and Zhang, Qiyiwen and Xin, Jiayi and Long, Qi and Chen, Tianlong},
  journal={Advances in Neural Information Processing Systems},
  volume={37},
  pages={98782--98805},
  year={2024}
}

@article{xu2026limssr,
  title={LIMSSR: LLM-Driven Sequence-to-Score Reasoning under Training-Time Incomplete Multimodal Observations},
  author={Xu, Huangbiao and Wu, Huanqi and Ke, Xiao and Peng, Yuxin},
  journal={arXiv preprint arXiv:2605.00434},
  year={2026}
}

@inproceedings{baek2026mars,
  title={Residual-Guided Expert Specialization for Incomplete Multimodal Learning},
  author={Baek, Seunghun and Park, Jihwan and Sim, Jaeyoon and Jeong, Minjae and Lee, Hoseok and Kim, Won Hwa},
  booktitle={European Conference on Computer Vision},
  pages={76--93},
  year={2026},
  organization={Springer}
}

@article{meng2026giml,
  title={General Incomplete Multimodal Learning via Dynamic Quality Perception},
  author={Meng, Xiangyu and Wei, Shicai},
  journal={arXiv preprint arXiv:2607.06943},
  year={2026}
}

@inproceedings{liu2022chsims,
  title={Make acoustic and visual cues matter: Ch-sims v2. 0 dataset and av-mixup consistent module},
  author={Liu, Yihe and Yuan, Ziqi and Mao, Huisheng and Liang, Zhiyun and Yang, Wanqiuyue and Qiu, Yuanzhe and Cheng, Tie and Li, Xiaoteng and Xu, Hua and Gao, Kai},
  booktitle={Proceedings of the 2022 international conference on multimodal interaction},
  pages={247--258},
  year={2022}
}

@article{zadeh2016mosi,
  title={Mosi: multimodal corpus of sentiment intensity and subjectivity analysis in online opinion videos},
  author={Zadeh, Amir and Zellers, Rowan and Pincus, Eli and Morency, Louis-Philippe},
  journal={arXiv preprint arXiv:1606.06259},
  year={2016}
}

@inproceedings{poria2019meld,
  title={Meld: A multimodal multi-party dataset for emotion recognition in conversations},
  author={Poria, Soujanya and Hazarika, Devamanyu and Majumder, Navonil and Naik, Gautam and Cambria, Erik and Mihalcea, Rada},
  booktitle={Proceedings of the 57th annual meeting of the association for computational linguistics},
  pages={527--536},
  year={2019}
}

@inproceedings{hasan2019urfunny,
  title={UR-FUNNY: A multimodal language dataset for understanding humor},
  author={Hasan, Md Kamrul and Rahman, Wasifur and Zadeh, AmirAli Bagher and Zhong, Jianyuan and Tanveer, Md Iftekhar and Morency, Louis-Philippe and Hoque, Mohammed Ehsan},
  booktitle={Proceedings of the 2019 conference on empirical methods in natural language processing and the 9th international joint conference on natural language processing (EMNLP-IJCNLP)},
  pages={2046--2056},
  year={2019}
}

@article{zhou2025dailyomni,
  title={Daily-omni: Towards audio-visual reasoning with temporal alignment across modalities},
  author={Zhou, Ziwei and Wang, Rui and Wu, Zuxuan and Jiang, Yu-Gang},
  journal={arXiv preprint arXiv:2505.17862},
  year={2025}
}

@inproceedings{hong2025worldsense,
  title={Worldsense: Evaluating real-world omnimodal understanding for multimodal llms},
  author={Hong, Jack and Yan, Shilin and Cai, Jiayin and Jiang, Xiaolong and Hu, Yao and Xie, Weidi},
  booktitle={International Conference on Learning Representations},
  volume={2026},
  pages={52423--52443},
  year={2026}
}

@article{oliveira2022circor,
  title={The CirCor DigiScope dataset: from murmur detection to murmur classification},
  author={Oliveira, Jorge and Renna, Francesco and Costa, Paulo Dias and Nogueira, Marcelo and Oliveira, Cristina and Ferreira, Carlos and Jorge, Al{\'\i}pio and Mattos, Sandra and Hatem, Thamine and Tavares, Thiago and others},
  journal={IEEE journal of biomedical and health informatics},
  volume={26},
  number={6},
  pages={2524--2535},
  year={2021},
  publisher={IEEE}
}

@article{mueller2005adni,
  title={The Alzheimer’s disease neuroimaging initiative},
  author={Mueller, Susanne G and Weiner, Michael W and Thal, Leon J and Petersen, Ronald C and Jack, Clifford and Jagust, William and Trojanowski, John Q and Toga, Arthur W and Beckett, Laurel},
  journal={Neuroimaging Clinics of North America},
  volume={15},
  number={4},
  pages={869},
  year={2005}
}

@article{xu2025qwenomni,
  author       = {Jin Xu and
                  Zhifang Guo and
                  Jinzheng He and
                  Hangrui Hu and
                  Ting He and
                  Shuai Bai and
                  Keqin Chen and
                  Jialin Wang and
                  Yang Fan and
                  Kai Dang and
                  Bin Zhang and
                  Xiong Wang and
                  Yunfei Chu and
                  Junyang Lin},
  title        = {Qwen2.5-Omni Technical Report},
  journal      = {CoRR},
  volume       = {abs/2503.20215},
  year         = {2025},
  url          = {https://doi.org/10.48550/arXiv.2503.20215},
  doi          = {10.48550/ARXIV.2503.20215},
  eprinttype   = {arXiv},
  eprint       = {2503.20215},
  bibsource    = {dblp computer science bibliography, https://dblp.org}
}

@inproceedings{ye2025omnivinci,
  title={Omnivinci: Enhancing architecture and data for omni-modal understanding llm},
  author={Ye, Hanrong and Yang, Chao-Han Huck and Goel, Arushi and Huang, Wei and Wan, Zhen and Tian, Jinchuan and Cheng, An-Chieh and Zhu, Ligeng and Su, Yuanhang and Lou, Yuming and others},
  booktitle={International Conference on Learning Representations},
  volume={2026},
  pages={56101--56138},
  year={2026}
}

@article{wolpert1992stacking,
  title={Stacked generalization},
  author={Wolpert, David H},
  journal={Neural networks},
  volume={5},
  number={2},
  pages={241--259},
  year={1992},
  publisher={Elsevier}
}

@article{jacobs1991mixtures,
  title={Adaptive mixtures of local experts},
  author={Jacobs, Robert A and Jordan, Michael I and Nowlan, Steven J and Hinton, Geoffrey E},
  journal={Neural computation},
  volume={3},
  number={1},
  pages={79--87},
  year={1991},
  publisher={MIT Press}
}

@inproceedings{wang2024foe,
  title={Fusing models with complementary expertise},
  author={Wang, Hongyi and Maia Polo, Felipe and Sun, Yuekai and Kundu, Souvik and Xing, Eric and Yurochkin, Mikhail},
  booktitle={International Conference on Learning Representations},
  volume={2024},
  pages={45284--45306},
  year={2024}
}

@inproceedings{patapati2026clarga,
  title={Multimodal Graph Representation Learning over Arbitrary Sets of Modalities},
  author={Patapati, Santosh and Srinivasan, Trisanth},
  booktitle={2026 IEEE/CVF Winter Conference on Applications of Computer Vision (WACV)},
  pages={7104--7115},
  year={2026},
  organization={IEEE}
}

@article{tong2025interactiveomni,
  title={Interactiveomni: A unified omni-modal model for audio-visual multi-turn dialogue},
  author={Tong, Wenwen and Guo, Hewei and Ran, Dongchuan and Chen, Jiangnan and Lu, Jiefan and Wang, Kaibin and Li, Keqiang and Zhu, Xiaoxu and Li, Jiakui and Li, Kehan and others},
  journal={arXiv preprint arXiv:2510.13747},
  year={2025}
}

@article{gemmateam2026gemma4,
  author       = {Sherif El Abd and
                  Vaibhav Aggarwal and
                  Robin Algayres and
                  Alek Andreev and
                  Olivier Bachem and
                  Ian Ballantyne and
                  Cormac Brick and
                  Victor Carbune and
                  Michelle Casbon and
                  Mayank Chaturvedi and
                  Aditya Chawla and
                  Victor Cotruta and
                  Alice Coucke and
                  Phil Culliton and
                  Robert Dadashi and
                  Lucas Dixon and
                  Mohamed Elhawaty and
                  Utku Evci and
                  Cl{\'{e}}ment Farabet and
                  Johan Ferret and
                  Filippo Galgani and
                  Sertan Girgin and
                  Jean{-}Bastien Grill and
                  Maarten Grootendorst and
                  Jiaxian Guo and
                  Cassidy Hardin and
                  Yanzhang He and
                  Steven M. Hernandez and
                  Omri Homburger and
                  L{\'{e}}onard Hussenot and
                  Ju{-}yeong Ji and
                  Armand Joulin and
                  Aishwarya Kamath and
                  Parnian Kassraie and
                  Olivier Lacombe and
                  Preethi Lahoti and
                  Ga{\"{e}}l Liu and
                  Gus Martins and
                  Luciano Martins and
                  Tatiana Matejovicova and
                  Ramona Merhej and
                  Nikola Momchev and
                  Sneha Mondal and
                  Ryan Mullins and
                  Sindhu Raghuram Panyam and
                  Shreya Pathak and
                  Sarah Perrin and
                  Andr{\'{e}} Susano Pinto and
                  Etienne Pot and
                  Ang{\'{e}}line Pouget and
                  Alexandre Ram{\'{e}} and
                  Sabela Ramos and
                  Douglas Reid and
                  David Rim and
                  Morgane Rivi{\`{e}}re and
                  Karsten Roth and
                  Louis Rouillard and
                  Omar Sanseviero and
                  Pier Giuseppe Sessa and
                  Shane Settle and
                  Danila Sinopalnikov and
                  Sara Smoot and
                  Piotr Stanczyk and
                  Andreas Steiner and
                  Lawrence Stewart and
                  Ilya O. Tolstikhin and
                  Michael Tschannen and
                  Anton Tsitsulin and
                  Nino Vieillard and
                  Renjie Wu and
                  Pingmei Xu and
                  Haichuan Yang and
                  Edouard Yvinec and
                  Biao Zhang and
                  Li Zhang and
                  Joe Zou and
                  Nicolas Aagnes and
                  Abdelrahman Abdelhamed and
                  Jakub Ad{\'{a}}mek and
                  Shivani Agrawal and
                  Shubham Agrawal and
                  Ibrahim Alabdulmohsin and
                  Jean{-}Baptiste Alayrac and
                  Uri Alon and
                  Chandramouli Amarnath and
                  Ankesh Anand and
                  Chrysovalantis Anastasiou and
                  Setareh Ariafar and
                  Fran{\c{c}}ois{-}Xavier Aubet and
                  Kyriakos Axiotis and
                  Federico Barbero and
                  Joelle K. Barral and
                  Alexei Bendebury and
                  Urs Bergmann and
                  Stanley Bileschi and
                  Kat Black and
                  Mathieu Blondel and
                  Sebastian Borgeaud and
                  Arthur Brazinskas},
  title        = {Gemma 4 Technical Report},
  journal      = {CoRR},
  volume       = {abs/2607.02770},
  year         = {2026},
  url          = {https://doi.org/10.48550/arXiv.2607.02770},
  doi          = {10.48550/ARXIV.2607.02770},
  eprinttype   = {arXiv},
  eprint       = {2607.02770},
  bibsource    = {dblp computer science bibliography, https://dblp.org}
}

@article{zhang2025streamomni,
  title={Stream-omni: Simultaneous multimodal interactions with large language-vision-speech model},
  author={Zhang, Shaolei and Guo, Shoutao and Fang, Qingkai and Zhou, Yan and Feng, Yang},
  journal={arXiv preprint arXiv:2506.13642},
  year={2025}
}

@article{wang2024emu3,
  title={Emu3: Next-token prediction is all you need},
  author={Wang, Xinlong and Zhang, Xiaosong and Luo, Zhengxiong and Sun, Quan and Cui, Yufeng and Wang, Jinsheng and Zhang, Fan and Wang, Yueze and Li, Zhen and Yu, Qiying and others},
  journal={arXiv preprint arXiv:2409.18869},
  year={2024}
}

@article{deng2025bagel,
  title={Emerging properties in unified multimodal pretraining},
  author={Deng, Chaorui and Zhu, Deyao and Li, Kunchang and Gou, Chenhui and Li, Feng and Wang, Zeyu and Zhong, Shu and Yu, Weihao and Nie, Xiaonan and Song, Ziang and others},
  journal={arXiv preprint arXiv:2505.14683},
  year={2025}
}

@article{luo2026torchumm,
  title={Torchumm: A unified multimodal model codebase for evaluation, analysis, and post-training},
  author={Luo, Yinyi and Wang, Wenwen and Bai, Hayes and Zhu, Hongyu and Chen, Hao and He, Pan and Savvides, Marios and Li, Sharon and Wang, Jindong},
  journal={arXiv preprint arXiv:2604.10784},
  year={2026}
}

@inproceedings{xie2025showo,
  title={Show-o: One single transformer to unify multimodal understanding and generation},
  author={Xie, Jinheng and Mao, Weijia and Bai, Zechen and Zhang, David Junhao and Wang, Weihao and Lin, Kevin Qinghong and Gu, Yuchao and Chen, Zhijie and Yang, Zhenheng and Shou, Mike Zheng},
  booktitle={International Conference on Learning Representations},
  volume={2025},
  pages={28240--28264},
  year={2025}
}

@article{lu2024collaboration,
  title={Merge, ensemble, and cooperate! a survey on collaborative strategies in the era of large language models},
  author={Lu, Jinliang and Pang, Ziliang and Xiao, Min and Zhu, Yaochen and Xia, Rui and Zhang, Jiajun},
  journal={arXiv preprint arXiv:2407.06089},
  year={2024}
}

@article{yadav2024moerging,
  title={A survey on model moerging: Recycling and routing among specialized experts for collaborative learning},
  author={Yadav, Prateek and Raffel, Colin and Muqeeth, Mohammed and Caccia, Lucas and Liu, Haokun and Chen, Tianlong and Bansal, Mohit and Choshen, Leshem and Sordoni, Alessandro},
  journal={arXiv preprint arXiv:2408.07057},
  year={2024}
}

@article{wu2024autogen,
  title={Autogen: Enabling next-gen llm applications via multi-agent conversation},
  author={Wu, Qingyun and Bansal, Gagan and Zhang, Jieyu and Wu, Yiran and Li, Beibin and Zhu, Erkang and Jiang, Li and Zhang, Xiaoyun and Zhang, Shaokun and Liu, Jiale and others},
  journal={arXiv preprint arXiv:2308.08155},
  year={2023}
}

@article{li2023camel,
  title={Camel: Communicative agents for" mind" exploration of large language model society},
  author={Li, Guohao and Hammoud, Hasan and Itani, Hani and Khizbullin, Dmitrii and Ghanem, Bernard},
  journal={Advances in neural information processing systems},
  volume={36},
  pages={51991--52008},
  year={2023}
}

@inproceedings{hong2024metagpt,
  author       = {Sirui Hong and
                  Mingchen Zhuge and
                  Jonathan Chen and
                  Xiawu Zheng and
                  Yuheng Cheng and
                  Jinlin Wang and
                  Ceyao Zhang and
                  Zili Wang and
                  Steven Ka Shing Yau and
                  Zijuan Lin and
                  Liyang Zhou and
                  Chenyu Ran and
                  Lingfeng Xiao and
                  Chenglin Wu and
                  J{\"{u}}rgen Schmidhuber},
  title        = {MetaGPT: Meta Programming for {A} Multi-Agent Collaborative Framework},
  booktitle    = {The Twelfth International Conference on Learning Representations,
                  {ICLR} 2024, Vienna, Austria, May 7-11, 2024},
  publisher    = {OpenReview.net},
  year         = {2024},
  url          = {https://openreview.net/forum?id=VtmBAGCN7o},
  bibsource    = {dblp computer science bibliography, https://dblp.org}
}

@inproceedings{gilmer2017neural,
  title={Neural message passing for quantum chemistry},
  author={Gilmer, Justin and Schoenholz, Samuel S and Riley, Patrick F and Vinyals, Oriol and Dahl, George E},
  booktitle={International conference on machine learning},
  pages={1263--1272},
  year={2017},
  organization={Pmlr}
}

@inproceedings{bachlechner2021rezero,
  title={Rezero is all you need: Fast convergence at large depth},
  author={Bachlechner, Thomas and Majumder, Bodhisattwa Prasad and Mao, Henry and Cottrell, Gary and McAuley, Julian},
  booktitle={Uncertainty in artificial intelligence},
  pages={1352--1361},
  year={2021},
  organization={PMLR}
}

@article{yang2025qwen3,
  title={Qwen3 technical report},
  author={Yang, An and Li, Anfeng and Yang, Baosong and Zhang, Beichen and Hui, Binyuan and Zheng, Bo and Yu, Bowen and Gao, Chang and Huang, Chengen and Lv, Chenxu and others},
  journal={arXiv preprint arXiv:2505.09388},
  year={2025}
}

@article{xin2025i2moe,
  title={I2moe: Interpretable multimodal interaction-aware mixture-of-experts},
  author={Xin, Jiayi and Yun, Sukwon and Peng, Jie and Choi, Inyoung and Ballard, Jenna L and Chen, Tianlong and Long, Qi},
  journal={Proceedings of machine learning research},
  volume={267},
  pages={68870},
  year={2025}
}

@inproceedings{whisper,
  title={Robust speech recognition via large-scale weak supervision},
  author={Radford, Alec and Kim, Jong Wook and Xu, Tao and Brockman, Greg and McLeavey, Christine and Sutskever, Ilya},
  booktitle={International conference on machine learning},
  pages={28492--28518},
  year={2023},
  organization={PMLR}
}
